%% file: main.tex
\documentclass[manuscript,screen,natbib=false]{acmart}

\newif\iftechreport
\techreporttrue

\usepackage[style=ACM-Reference-Format,backend=biber,sorting=none]{biblatex}

\AtEveryBibitem{%
  \clearlist{language}
  \clearlist{location}
  \clearlist{institution}
  \clearname{editor}
  \clearname{editorb}
  \clearfield{month}%
  \clearfield{isbn}%
  \clearfield{issn}%
  \clearfield{pages}%
  \clearfield{review}%
  \clearfield{series}%
  \clearlist{address}
  \clearfield{date}
  \clearfield{eprint}

  \ifentrytype{collection}{
    \clearlist{publisher}
    \clearfield{volume}
    \clearfield{number}
    \clearfield{url}
    \clearfield{urlday}
    \clearfield{urlmonth}
    \clearfield{urlyear}
  }{}

  \ifentrytype{incollection}{
    \clearfield{url}
    \clearfield{urlday}
    \clearfield{urlmonth}
    \clearfield{urlyear}
    \clearlist{publisher}
    \clearfield{volume}
    \clearfield{number}
  }{}

  \ifentrytype{inproceedings}{
    \clearlist{publisher}
    \clearfield{volume}
    \clearfield{number}
    \clearfield{url}
    \clearfield{urlday}
    \clearfield{urlmonth}
    \clearfield{urlyear}
  }{}

  \ifentrytype{article}{
    \clearfield{url}
    \clearfield{urlday}
    \clearfield{urlmonth}
    \clearfield{urlyear}
  }{}
}

\newcommand{\myparagraph}[1]{\medskip \noindent \textbf{#1}}

\usepackage{./ressources/ressources}
\usepackage{./ressources/textmacros}
\usepackage{./ressources/my_listings}
\graphicspath{{./ressources/img/}}

\acmConference[TOPS]{Transactions on Privacy and Security}{May, 2021}{}
\acmBooktitle{TOPS}
\acmISBN{}
\title[Binsec/Rel: Symbolic Binary Analyzer for Security]{Binsec/Rel: Symbolic Binary Analyzer for Security with Applications to
  Constant-Time and Secret-Erasure}

\author{Lesly-Ann Daniel}
\email{lesly-ann.daniel@cea.fr}
\affiliation{%
  \institution{Université Paris-Saclay, CEA, List}
  \postcode{F-91120}
  \city{Palaiseau}
  \country{France}
}
\affiliation{%
  \institution{imec-DistriNet, KU Leuven}
  \city{Leuven}
  \country{Belgium}
}

\author{Sébastien Bardin}
\email{sebastien.bardin@cea.fr}
\affiliation{%
  \institution{Université Paris-Saclay, CEA, List}
  \postcode{F-91120}
  \city{Palaiseau}
  \country{France}
}

\author{Tamara Rezk}
\email{tamara.rezk@inria.fr}
\affiliation{%
  \institution{Inria}
  \city{Sophia Antipolis}
  \country{France}
}

\begin{document}
\begin{abstract}
  This paper tackles the problem of designing efficient binary-level verification for a
  subset of information flow properties encompassing \emph{constant-time} and
  \emph{secret-erasure}. These properties are crucial for cryptographic implementations, but
  are generally not preserved by compilers.
  Our proposal builds on relational symbolic execution enhanced with new
  optimizations dedicated to information flow and binary-level analysis,
  yielding a dramatic improvement over prior work based on symbolic
  execution.
  We implement a prototype, \brelse{}, for bug-finding and bounded-verification
  of constant-time and secret-erasure, and perform extensive experiments on a set
  of 338 cryptographic implementations, demonstrating the benefits of our
  approach.
  Using \brelse{}, we also automate two %
  prior manual studies on preservation of constant-time and secret-erasure by
  compilers for a total of \newercontent{4148} and 1156 binaries respectively. %
  Interestingly, our analysis highlights incorrect usages of volatile data
  pointer for secret erasure and shows that scrubbing mechanisms based on
  \emph{volatile function pointers} can introduce additional register spilling
  which might break secret-erasure. We also discovered that \texttt{gcc -O0} and
  backend passes of \texttt{clang} introduce violations of constant-time in
  implementations that were previously deemed secure by a state-of-the-art
  constant-time verification tool operating at LLVM level, showing the
  importance of reasoning at binary-level.
\end{abstract}



\maketitle

\section{Introduction}\label{sec:intro}
\input{intro}

\section{Background}\label{sec:background} %
\input{background}

\section{Motivating Example: Constant-Time Analysis}\label{sec:motivating}
\input{motivating}

\section{Concrete Semantics and Leakage
  Model}\label{sec:concrete_semantics}
\input{concrete_semantics}

\section{Binary-level Relational Symbolic
  Execution}\label{sec:std_relse}
\input{relse}

\section{Experimental Results}\label{sec:expes}
\input{expes}

\section{Discussion}\label{sec:discussion}
\input{discussion}

\section{Related Work}\label{sec:related} %
\input{rw}

\section{Conclusion}\label{sec:conclusion} %
We tackle the problem of designing an automatic and efficient binary-level
analyzer for \emph{information flow} properties, enabling both bug-finding and
bounded-verification on real-world cryptographic implementations. Our approach
is based on \emph{relational symbolic execution} together with original
\emph{dedicated optimizations} reducing the overhead of relational reasoning and
allowing for a significant speedup.
Our prototype, \brelse{}, is shown to be highly efficient compared to
alternative approaches. We used it to perform extensive binary-level
constant-time analysis and secret-erasure for a wide range of cryptographic
implementations, and to automate prior manual studies on the preservation of
constant-time and secret-erasure by compilers.
We highlight incorrect usages of volatile data pointer for secret erasure, and
show that scrubbing mechanisms based on \emph{volatile function pointers} can
introduce additional violation from register spilling.
We also found three constant-time vulnerabilities that slipped through prior
manual and automated analyses, and we discovered that \texttt{gcc -O0} and
backend passes of \texttt{clang} introduce violations of constant-time out of
reach of state-of-the-art constant-time verification tools at LLVM or source
level.

\section*{Acknowledgments} %
We would like to thank Guillaume Girol for his help with setting up Nix virtual
environments, which enable reproducible compilation in our frameworks, as well
as Frédéric Recoules for his help with the final release of the tool. We also
thank the anonymous reviewers for their valuable suggestions, which greatly
helped to improve the paper. This project has received funding from the European
Union Horizon 2020 research and innovation program under grant agreement No
101021727, from ANR grant ANR-20-CE25-0009-TAVA, and from ANR-17-CE25-0014-01
CISC project.

\printbibliography
\iftechreport
\newpage
\appendix
\input{proofs}
\input{appendix}
\fi

\end{document}

%% file: intro.tex
Safety properties~\cite{DBLP:journals/dc/AlpernS87}, such as buffer
overflows, have been extensively studied and numerous efficient tools have been
developed for their
verification~\cite{DBLP:conf/osdi/CadarDE08,DBLP:journals/cacm/GodefroidLM12,DBLP:journals/fac/KirchnerKPSY15,
  DBLP:journals/sttt/HavelundP00, DBLP:conf/popl/JourdanLBLP15,
  DBLP:conf/esop/CousotCFMMMR05,DBLP:journals/ieeesp/AvgerinosBDGNRW18}.
However, safety properties are properties of individual execution traces,
whereas many important security properties are expressed as sets of
traces---i.e., are \emph{hyperproperties}~\cite{DBLP:conf/csfw/ClarksonS08}. In
particular, information flow properties, which regulate the leakage of
information from the secret inputs of a program to public outputs, relate two
execution traces---i.e., are \emph{2-hypersafety properties}. %

\emph{Constant-time} and \emph{secret-erasure} are two examples of
\emph{2-hypersafety properties} that are crucial in cryptographic
implementations.
The constant-time programming discipline (CT) is a software-based countermeasure
to timing and microarchitectural attacks which requires the control flow and the
memory accesses of the program to be independent from the secret
input\footnote{Some versions of constant-time also require that the size of operands of
  variable-time instructions (e.g.~integer division) is independent from
  secrets.}. %
Constant-time has been proven to protect against cache-based timing
attacks~\cite{DBLP:conf/ccs/BartheBCLP14} %
and is widely used to secure cryptographic implementations (e.g.
BearSSL~\cite{BearSSLConstantTimeCrypto},
NaCL~\cite{DBLP:conf/latincrypt/BernsteinLS12},
HACL*~\cite{DBLP:conf/ccs/ZinzindohoueBPB17}, etc). %
Secret-erasure~\cite{DBLP:conf/csfw/ChongM05} (a.k.a.\ data scrubbing or safe
erasure) requires to clear secret data (e.g.\ secret keys) from the memory after
the execution of a critical function, for instance by zeroing the corresponding
memory. It ensures that secret data do not persist in memory longer than
necessary, protecting them against subsequent memory disclosure vulnerabilities.

\myparagraph{Problem.} %
These properties are generally not preserved by
compilers~\cite{DBLP:conf/eurosp/SimonCA18,DBLP:conf/sp/DSilvaPS15,DBLP:conf/csfw/BessonDJ19}.
For example, reasoning about constant-time requires to know whether the code
\lstinline{c=(x<y)-1} will be compiled to branchless code or not, but this
depends on the compiler version and
optimization~\cite{DBLP:conf/eurosp/SimonCA18}.
Similarly, scrubbing operations used for secret-erasure have no effect on the
result of the program and can therefore be optimized away by the
dead-store-elimination pass of the
compiler~\cite{DBLP:conf/uss/YangJOLL17,DBLP:conf/csfw/BessonDJ19,DBLP:conf/sp/DSilvaPS15},
as detailed in CWE-14~\cite{CWE14CompilerRemoval}. Moreover, these scrubbing
operations do not erase secrets that have been copied on the stack by compilers,
e.g.\ from register spilling.

Several CT-analysis tools have been proposed to analyze source
code~\cite{bacelaralmeidaFormalVerificationSidechannel2013,DBLP:conf/esorics/BlazyPT17},
or LLVM
code~\cite{DBLP:conf/uss/AlmeidaBBDE16,DBLP:conf/sp/BrotzmanLZTK19},
but leave the gap opened for violations introduced in the executable
code either by the compiler~\cite{DBLP:conf/eurosp/SimonCA18} or by
closed-source libraries~\cite{DBLP:conf/cans/KaufmannPVV16}.
Binary-level tools for constant-time using dynamic
approaches~\cite{langleyImperialVioletCheckingThat2010,DBLP:conf/memocode/ChattopadhyayBR17,DBLP:conf/uss/WangWLZW17,DBLP:conf/acsac/WichelmannMES18}
can find bugs, but otherwise miss vulnerabilities in unexplored portions of the
code, making them incomplete. Conversely, static
approaches~\cite{DBLP:conf/cav/KopfMO12,DBLP:conf/uss/DoychevFKMR13,DBLP:conf/pldi/DoychevK17}
cannot report precise counterexamples---making them of minor interest when the
implementation cannot be proven secure. %
For secret-erasure there is currently no sound automatic analyzer. Existing
approaches rely on dynamic
tainting~\cite{Secretgrind2020,DBLP:conf/eurosp/SimonCA18} or manual binary-code
analysis~\cite{DBLP:conf/uss/YangJOLL17}. While there has been some work on
security preserving
compilers~\cite{DBLP:conf/csfw/BessonDJ19,DBLP:conf/eurosp/SimonCA18},
they are not always applicable and are ineffective for detecting errors in
existing binaries.

\myparagraph{Challenges.} %
Two main challenges arise in the verification of these properties:
\begin{itemize}
\item First, common verification methods do not directly apply because
  information flow properties like constant-time and secret-erasure are
  not regular safety properties but 2-hypersafety
  properties~\cite{DBLP:conf/csfw/ClarksonS08} (i.e., relating two
  execution traces), and their standard reduction to safety on a
  transformed program,
  \emph{self-composition}~\cite{DBLP:conf/csfw/BartheDR04}, is
  inefficient~\cite{DBLP:conf/sas/TerauchiA05};

  \item Second, it is notoriously difficult to adapt formal methods to
        binary-level because of the lack of structure information (data and
        control) and the need to explicitly reason about the
        memory~\cite{DBLP:conf/fm/DjoudiBG16,DBLP:journals/toplas/BalakrishnanR10}.
\end{itemize}

\noindent A technique that scales well on binary code and that naturally comes
into play for bug-finding and bounded-verification is \emph{symbolic execution}
(SE)~\cite{DBLP:journals/cacm/GodefroidLM12,DBLP:journals/cacm/CadarS13}. While
it has %
proven very successful for standard safety
properties~\cite{DBLP:conf/icse/BounimovaGM13}, its direct adaptation to
2-hypersafety properties through (variants of) self-composition suffers from a
scalability
issue~\cite{DBLP:conf/ccs/BalliuDG14,DBLP:conf/sec/DoBH15,DBLP:conf/forte/MilushevBC12}.
Some recent approaches scale better, but at the cost of sacrificing
either
bounded-verification~\cite{DBLP:conf/uss/WangWLZW17,DBLP:conf/date/SubramanyanMKMF16}
(by doing under-approximations) %
or bug-finding~\cite{DBLP:conf/sp/BrotzmanLZTK19} (by doing
over-approximations). %

The idea of analyzing pairs of executions for the verification of
2-hypersafety is not new (e.g.\ relational Hoare
logic~\cite{DBLP:conf/popl/Benton04},
self-composition~\cite{DBLP:conf/csfw/BartheDR04}, product
programs~\cite{DBLP:conf/fm/BartheCK11}, multiple
facets~\cite{DBLP:conf/popl/AustinF12,DBLP:conf/www/NgoBFRRS18}). In the context of
symbolic execution, it has first been coined as
\emph{ShadowSE}~\cite{DBLP:conf/icse/PalikarevaKC16} for back-to-back
testing, and later as \emph{relational symbolic execution}
(RelSE)~\cite{farinaRelationalSymbolicExecution2019}.
However, because of the necessity to model the memory, RelSE cannot be trivially
adapted to binary-level analysis. In particular, representing the memory as a
large array of bytes prevents sharing between executions and precise
information-flow tracking, which generates \emph{a high number of queries} for
the constraint solver. Hence, a \emph{direct application of RelSE does not
  scale}.

\myparagraph{Proposal.} %
We restrict to a subset of information flow properties relating traces
following the same path---which includes interesting security policies
such as constant-time and secret-erasure.  \emph{We tackle the problem
  of designing an efficient symbolic verification tool for
  constant-time and secret-erasure at binary-level, that leverages the
  full power of symbolic execution without sacrificing correct
  bug-finding nor bounded-verification.}  We present \brelse{}, the
first efficient binary-level automatic tool for bug-finding and
bounded-verification of constant-time and secret-erasure at
binary-level. It is compiler-agnostic, targets x86 and ARM
architectures and does not require source code.

The technique is based on \emph{relational symbolic
  execution}~\cite{DBLP:conf/icse/PalikarevaKC16,farinaRelationalSymbolicExecution2019}:
it models two execution traces following the same path in the same symbolic
execution instance %
and \emph{maximizes sharing between them}. We show via experiments
(\cref{sec:scalability}) that RelSE alone does not scale at binary-level to
analyze constant-time on real cryptographic implementations.
Our key technical insights are (1) to complement RelSE with dedicated
optimizations offering a fine-grained information flow tracking in the memory,
improving sharing at binary-level (2) to use this sharing to track
secret-dependencies and reduce the number of queries sent to the solver.

\brelse{} can analyze about 23 million instructions in 98 min (3860
instructions per second), outperforming similar state of the art
binary-level  symbolic
 analyzers~\cite{DBLP:conf/date/SubramanyanMKMF16,DBLP:conf/uss/WangWLZW17}
(cf.~\cref{tab:comparison_se}, page~\pageref{tab:comparison_se}),
while being still correct and complete.

\myparagraph{Contributions.}  Our contributions are the following:
\begin{itemize}

  \item We design dedicated optimizations for information flow analysis at
        binary-level. First, we complement relational symbolic execution with a
        new \emph{on-the-fly} simplification for \emph{binary-level} analysis,
        to track secret-dependencies and maximize sharing in the memory
        (\cref{sec:row}). Second, we design new simplifications for
        \emph{information flow} analysis: untainting (\cref{sec:untainting}) and
        fault-packing (\cref{sec:fp}). Moreover, we formally prove that our
        analysis is correct for bug-finding and bounded-verification of
        constant-time (\cref{sec:proofs}) and discuss the adaptation of the
        guarantees to other information-flow properties \iftechreport{}
        (\cref{sec:discussion_proofs}); \fi

  \item We propose a tool named \brelse{} for constant-time and secret-erasure
        analysis. Extensive experimental evaluation (338 samples) against
        standard approaches (\cref{sec:scalability}) shows that it can find bugs
        in real-world cryptographic implementations much faster than these
        techniques (\(\times 1000\) speedup) and can achieve bounded-verification
        when they time out, with a performance close to standard SE
        (\(\times 2\) overhead);

  \item In order to prove the effectiveness of \brelse{}, we perform an
        extensive analysis of constant-time at binary-level. In particular, we
        analyze 296 cryptographic binaries previously verified at a higher-level
        (incl.\ codes from HACL*~\cite{DBLP:conf/ccs/ZinzindohoueBPB17},
        BearSSL~\cite{BearSSLConstantTimeCrypto},
        NaCL~\cite{DBLP:conf/latincrypt/BernsteinLS12}), we replay known bugs in 42
        samples (incl.\ Lucky13~\cite{DBLP:conf/sp/AlFardanP13}) %
        and automatically generate counterexamples (\cref{sec:effectiveness});

  \item Simon \emph{et al.}~\cite{DBLP:conf/eurosp/SimonCA18} have demonstrated
        that \texttt{clang}'s optimizations break constant-timeness of code. We
        extend this work in five directions---from 192
        in~\cite{DBLP:conf/eurosp/SimonCA18} to \newercontent{4148} configurations
        (\cref{sec:compilers}):
        \begin{enumerate*}
          \item we automatically analyze the code that was manually checked
          in~\cite{DBLP:conf/eurosp/SimonCA18},
          \item we add new implementations,
          \item we add the \texttt{gcc} compiler and a more recent version of
          \texttt{clang},
          \item we add \newercontent{\texttt{i686}} and ARM,
          \item \newcontent{we investigate the impact of individual
            optimizations---i.e., the \texttt{-x86-cmov-converter} of
            \texttt{clang} and the if-conversion passes of \texttt{gcc}}.
        \end{enumerate*}
        Interestingly, we discovered that \texttt{gcc -O0} and backend passes of
        \texttt{clang} introduce violations of constant-time that cannot be
        detected by LLVM verification tools like
        ct-verif~\cite{DBLP:conf/uss/AlmeidaBBDE16} \newercontent{\emph{even when the
        \texttt{-x86-cmov-converter} is disabled}}. \newcontent{On a positive
        note, we also show that, contrary to \texttt{clang}, \texttt{gcc}
        optimizations tend to help preserve constant-time.} \newcontent{This
        study is open-source and can be easily extended with new compilers and
        programs;} %

  \item Finally, we build the first framework to automatically check the
        preservation of secret-erasure by compilers. We use it to analyze 17
        scrubbing functions---including countermeasures manually analyzed in a
        prior study~\cite{DBLP:conf/uss/YangJOLL17}, compiled with 10
        compilers with different optimization levels, for a total of 1156 binaries
        (\cref{sec:expes-secret-erasure}). Our analysis: %
        \begin{enumerate*}
          \item \newcontent{confirms that the main optimization affecting the
            preservation of secret-erasure is the dead store elimination pass
            (\texttt{-dse}), but also shows that disabling it is not always
            sufficient to preserve secret-erasure},
          \item shows that, while some versions of scrubbing functions based on
          \emph{volatile data pointer} are secure, it is easy to implement this
          mechanism incorrectly, in particular by using a volatile pointer to non-volatile
          data, or passing a pointer to volatile in a function call,
          \item interestingly it also shows that scrubbing mechanisms based on
          \emph{volatile function pointers} can introduce additional register
          spilling that might break secret-erasure with \texttt{gcc -O2} and
          \texttt{gcc -O3},
          \item finally, secret-erasure mechanisms based on dedicated secure
          functions (i.e., \texttt{explicit\_bzero}, \texttt{memset\_s}), memory
          barriers, and weak symbols, are preserved in all tested setups.
        \end{enumerate*}
        This framework is open-source and can be easily extended with new
        compilers and new scrubbing functions;
\end{itemize}

\myparagraph{Discussion.} %
Our technique is shown to be highly efficient on bug-finding and
bounded-verification compared to alternative approaches, paving the way to a
systematic binary-level analysis of information-flow properties on cryptographic
implementations, while our experiments demonstrate the importance of developing
verification tools reasoning at binary-level. %
Besides constant-time and secret-erasure, the tool can be readily adapted to
other 2-hypersafety properties of interest in security (e.g., cache-side
channels, or variants of constant-time taking operand size into account)---as
long as they restrict to pairs of traces following the same path. %

\myparagraph{Availability.} We made \brelse{} open-source at
\url{https://github.com/binsec/rel}, our experiments are available at
\url{https://github.com/binsec/rel_bench}, and in particular, our studies on the
preservation of constant-time and secret-erasure by compilers are available at
\url{https://github.com/binsec/rel_bench/tree/main/properties_vs_compilers}.

\myparagraph{Extension of article~\cite{DBLP:conf/sp/DanielBR20}.} %
This paper is an extension of the article \textit{Binsec/Rel: Efficient
  Relational Symbolic Execution for Constant-Time at
  Binary-Level}~\cite{DBLP:conf/sp/DanielBR20}, with the following
additional contributions:
\begin{itemize}
  \item The leakage model considered in~\cite{DBLP:conf/sp/DanielBR20} restricts
        to constant-time while this work encompasses a more general subset of
        information flow properties. In particular, we define a new leakage
        model and property to capture the notion of \emph{secret-erasure} (cf.\
        \cref{sec:leakage_model,sec:secure_program});

  \item We extend the \brelse{} tool to verify the \emph{secret-erasure} property;

  \item We perform an experimental evaluation on the preservation of
        secret-erasure by compilers (cf. \cref{sec:expes-secret-erasure}). This
        evaluation highlights incorrect usages of volatile data pointers for
        secret erasure, and shows that scrubbing mechanisms based on
        \emph{volatile function pointers} can introduce additional violations
        from register spilling;

  \item \newcontent{Using \brelse{}, we also investigate the role of individual
        compiler optimizations in the preservation of secret-erasure and
        constant-time. For constant-time, we show that the if-conversion passes
        of \texttt{gcc} may help enforce constant-time in ARM binaries. We also
        show that disabling the \texttt{cmov-converter} is not always sufficient
        to preserve constant-time in the backend-passes of \texttt{clang}. For
        secret-erasure, we confirm the key role of the dead store elimination
        pass (\texttt{-dse}), but also show that disabling it does not always
        preserve secret-erasure.}
\end{itemize}

\iftechreport{} %
In addition, we provide full proofs of relative completeness and correctness of
the analysis---whereas simple sketches of proofs were given
in~\cite{DBLP:conf/sp/DanielBR20} (\cref{sec:proofs})---\newercontent{ we
  evaluate the scalability of \brelse{} according to the size of the input
  (\cref{app:scale_size}), and we detail the vulnerabilities introduced by
  \texttt{clang} with examples (\cref{app:ct-sort}). } %
\fi %

%% file: background.tex
In this section, we present the basics of information flow properties
and symbolic execution.
Small examples of constant-time and standard adaptations of symbolic
execution are presented in \cref{sec:motivating}, while formal
definitions of information flow policies (including constant-time and
secret-erasure) are given in \cref{sec:concrete_semantics}.

\subsubsection*{Information flow properties}
Information flow policies regulate the transfer of information between public
and secret domains. To reason about information flow, the program input is
partitioned into two disjoint sets: \emph{low} (i.e., public) and \emph{high}
(i.e., secret).
Typical information flow properties require that the observable output of a
program does not depend on the high input
(\textit{non-interference}~\cite{DBLP:journals/cacm/DenningD77}). Constant-time
and secret-erasure can be expressed as information flow properties.
\emph{Constant-time} requires both the program control flow and the memory
accesses to be independent from high input. \newcontent{It protects against
  timing and microarchitectural attacks (exploiting cache, port contention,
  branch predictors, etc.)}. %
\emph{Secret-erasure} requires specific memory locations (typically the call
stack) to be independent from high input when returning from a critical
function. \newcontent{It ensures that secret data do not persist in memory
  longer than necessary~\cite{DBLP:conf/uss/ChowPGCR04}, protecting these secret
  data against subsequent memory exposure, e.g.\ memory disclosure
  vulnerabilities, access to persistent storage (swap memory).}

Contrary to standard \emph{safety} properties which state that nothing bad can
happen along \emph{one execution trace}, information flow properties relate
\emph{two execution traces}---they are \emph{2-hypersafety
  properties}~\cite{DBLP:conf/csfw/ClarksonS08}.
Unfortunately, the vast majority of symbolic execution
tools~\cite{DBLP:journals/cacm/GodefroidLM12,DBLP:conf/osdi/CadarDE08,DBLP:journals/ieeesp/AvgerinosBDGNRW18,DBLP:journals/tocs/ChipounovKC12,DBLP:conf/sp/Shoshitaishvili16,DBLP:conf/wcre/DavidBTMFPM16}
is designed for safety verification and cannot directly be applied to
2-hypersafety properties.
In principle, 2-hypersafety properties can be reduced to standard
safety properties of a \emph{self-composed
  program}~\cite{DBLP:conf/csfw/BartheDR04} but this reduction alone does not
scale~\cite{DBLP:conf/sas/TerauchiA05}.

\subsubsection*{Symbolic execution}
 Symbolic Execution
(SE)~\cite{DBLP:journals/cacm/King76,DBLP:journals/cacm/CadarS13,DBLP:journals/cacm/GodefroidLM12}
consists in executing a program on \emph{symbolic inputs} instead of
concrete inputs. Variables and expressions of the program are
represented as terms over these symbolic inputs and the current path is
modeled by a \emph{path predicate} (a logical formula), which is the
conjunction of conditional expressions encountered along the
execution. 
 SE is built upon two main steps. %
(1) \emph{Path search}: at each conditional statement the symbolic
execution \emph{forks}, the expression of the condition is added to the
first branch and its negation to the second branch, then the symbolic
execution continues along satisfiable branches; %
(2) \emph{Constraint solving}: the path predicate can be solved with
an off-the-shelf \emph{automated constraint solver}, typically
SMT~\cite{DBLP:conf/woot/VanegueH12}, to generate a concrete input
exercising the path.
 Combining these two steps, SE can explore different program paths
and generate test inputs exercising these paths. It can also check
local assertions in order to \emph{find bugs} or perform
\emph{bounded-verification} (i.e., verification up to a certain
depth).
Dramatic progress in program analysis and constraint solving over
the last two decades have made SE a tool of choice for intensive
testing~\cite{DBLP:conf/icse/BounimovaGM13,DBLP:journals/cacm/CadarS13},
vulnerability
analysis~\cite{DBLP:journals/ieeesp/AvgerinosBDGNRW18,DBLP:conf/ndss/AvgerinosCHB11,DBLP:conf/uss/SchwartzAB11}
and other security-related
analysis~\cite{DBLP:conf/sp/YadegariJWD15,DBLP:conf/sp/BardinDM17}.

\subsubsection*{Binary-level symbolic execution}\label{sec:bl-se}
 Low-level code operates on a set of registers and a single (large)
untyped memory.  During the execution, %
a call stack contains information about the active functions such as
their arguments and local variables. %
A special register \texttt{esp} (stack pointer) indicates the top
address of the call stack and local variables of a function can be
referenced as offsets from the initial
\texttt{esp}\footnote{\texttt{esp} is specific to x86, but this is
  generalizable, e.g.~\texttt{sp} for ARMv7.}.

Binary-level code analysis is notoriously more challenging than source code
analysis~\cite{DBLP:journals/toplas/BalakrishnanR10,DBLP:conf/fm/DjoudiBG16}.
First, evaluation and assignments of source code variables become memory load
and store operations, requiring to reason explicitly about the memory in a very
precise way. Second, the high level control flow structure (e.g. \texttt{for}
loops) is not preserved, and there are indirect jumps to handle (e.g.~instruction
of the form \lstinline{jmp eax}). Fortunately, it turns out that SE is less
difficult to adapt from source code to binary code than other semantic
analysis---due to both the efficiency of SMT solvers and concretization (i.e.,
simplifying a formula by constraining some variables to be equal to their
observed runtime values).
Hence, strong binary-level SE tools do exist and have yielded several highly
promising case
studies~\cite{DBLP:journals/cacm/GodefroidLM12,DBLP:journals/ieeesp/AvgerinosBDGNRW18,DBLP:journals/tocs/ChipounovKC12,DBLP:conf/sp/Shoshitaishvili16,DBLP:conf/wcre/DavidBTMFPM16,DBLP:conf/sp/BardinDM17,DBLP:conf/dimva/SalwanBP18}.
In this paper, we build on top of the binary-analysis platform
\textsc{Binsec}~\cite{DBLP:conf/tacas/DjoudiB15} and in particular its symbolic
execution engine \textsc{Binsec/SE}~\cite{DBLP:conf/wcre/DavidBTMFPM16}.

\newcontent{One of the key components of binary-level symbolic execution is the
  representation of the memory. A first solution, adopted in
  \textsc{Binsec/SE}~\cite{DBLP:conf/wcre/DavidBTMFPM16} and
  \textsc{Bap}~\cite{DBLP:conf/cav/BrumleyJAS11}, is to use a \emph{fully
    symbolic memory model} in which the memory is represented as a symbolic
  array of bytes. Other solutions consist in concretizing (parts of) the memory.
  For instance, angr~\cite{DBLP:conf/sp/Shoshitaishvili16} uses a partially
  symbolic memory model~\cite{DBLP:journals/ieeesp/AvgerinosBDGNRW18} in which
  write addresses are concretized and symbolic loads are encoded as symbolic
  if-the-else expressions. %
  Fully symbolic memory models incur a performance overhead compared to
  partially symbolic (or concrete) memory models. However, they can model all
  possible values that load/write addresses can take---instead of considering
  only a subset of the possible addresses. Hence, they offer better soundness
  guarantees and are better suited for bounded-verification.}

\myparagraph{Logical notations.} \textsc{Binsec/SE} relies on the theory of
bitvectors and arrays, \abv{}~\cite{barrettSMTLIBStandardVersion2017}. %
Values (e.g.~registers, memory addresses, memory content) are modeled with
fixed-size bitvectors~\cite{FixedSizeBitVectorsTheorySMTLIB}. %
We use the type \(\bvtype{m}\), where \(m\) is a constant number, to represent
symbolic bitvector expressions of size \(m\). %
The memory is modeled with a logical array~\cite{ArraysExTheorySMTLIB} of type
\(\memtype{}\) (assuming a $32$ bit architecture). %
A logical array is a function \((Array~\mathcal{I}~\mathcal{V})\) that maps each
index \(i \in \mathcal{I}\) to a value \(v \in \mathcal{V}\).
Operations over arrays are:
\begin{itemize}
\item
  \(select : (Array~\mathcal{I}~\mathcal{V}) \times \mathcal{I}
  \rightarrow \mathcal{V}\) takes an array \(a\) and an index \(i\)
  and returns the value \(v\) stored at index \(i\) in \(a\),
\item
  \(store: (Array~\mathcal{I}~\mathcal{V}) \times \mathcal{I} \times
  \mathcal{V} \rightarrow (Array\ \mathcal{I}\ \mathcal{V})\) takes an
  array \(a\), an index \(i\), and a value \(v\), and returns the
  array \(a\) modified so that \(i\) maps to \(v\).
\end{itemize}
These functions satisfy the following constraints for all
\({a \in(Array~\mathcal{I}~\mathcal{V})}\), \({i \in \mathcal{I}}\),
\({j \in \mathcal{I}}\), \({v \in \mathcal{V}}\):
\begin{itemize}
  \item \(select~(store~a~i~v)~i = v\): a store of a value \(v\) at index \(i\)
        followed by a \(select\) at the same index returns the value \(v\);
  \item \(i \neq j \implies select~(store~a~i~v)~j = select~a~j\): a store at
        index \(i\) does not affect values stored at other indexes \(j\).
\end{itemize}

%% file: motivating.tex
Consider the constant-time policy applied to the toy program in
\Cref{list:motivating}. The outcome of the conditional instruction at line~\ref{line:motivating:cf} and the
memory access at line~\ref{line:motivating:mem} are \emph{leaked}. We say that a \emph{leak is
  insecure} if it depends on the secret input. Conversely, a
\emph{leak is secure} if it does not depend on the secret
input. Constant-time holds for a program if there is no insecure leak.

\begin{figure}[htbp]
    \centering
    \begin{subfigure}[t]{.45\textwidth}
      \input{./ressources/listings/motivating}  
    \end{subfigure}
    \hfill %
    \begin{subfigure}[t]{.5\textwidth}
        \input{./ressources/listings/limitation_relse}
    \end{subfigure}
\end{figure}

\textbf{Example.} Consider two executions of this program with the same public
input: \((x,y)\) and \((x',y')\) where \(y = y'\). Intuitively, we can see that
the leakages produced at line~\ref{line:motivating:cf}, \(y = 0\) and \(y' = 0\), are necessarily equal
in both executions because \(y = y'\); hence this leak does not depend on the
secret input and is secure. On the contrary, the leakages \(x\) and \(x'\) at
line~\ref{line:motivating:mem} can differ in both executions (e.g.\ with \(x = 0\) and \(x' = 1\));
hence this leak depends on the secret input and is insecure.

\emph{The goal of an automatic analysis is to prove that the leak at
line~\ref{line:motivating:cf} is secure and to return  concrete input showing that the leak
at line~\ref{line:motivating:mem} is insecure.}

\subsection{Symbolic Execution and Self-Composition (SC)} %
Symbolic execution can be adapted to the case of constant-time, following the
self-composition principle. Instead of self-composing the program, we rather
self-compose the formula with a renamed version of itself plus a precondition
stating that the low inputs are
equal~\cite{DBLP:conf/iccsw/Phan13}.
Basically, this amounts to model \emph{two different executions
  following the same path and sharing the same low input} in a single
formula.
At each conditional statement, \emph{exploration queries} are sent to
the solver to determine satisfiable branches. %
Additionally, \emph{insecurity queries} specific to constant-time are sent
before each control-flow instruction and memory access to determine whether they
depend on the secret---if an insecurity query is satisfiable then a
constant-time violation is found.

As an illustration, let us consider the program in \Cref{list:motivating}.
First, we assign symbolic values to \lstinline{x} and \lstinline{y} and use
symbolic execution to generate a formula of the program until the first
conditional jump (line~\ref{line:motivating:cf}), resulting in the formula:
\(x = \beta ~\wedge~ y = \lambda ~\wedge~ c = (\lambda \neq 0)\). Second,
self-composition is applied on the formula with precondition
\(\lambda = \lambda'\) to constrain the low inputs to be equal in both
executions. Finally, a postcondition \(c \neq c'\) asks whether the value of the
condition can differ, resulting in the following insecurity query:
\begin{equation*}
  \lambda = \lambda' ~\wedge~
  \left(\begin{aligned}
      x = \beta ~\wedge~ y = \lambda ~\wedge~ c = (\lambda \neq 0) ~\wedge~ \\
      x' = \beta' ~\wedge~ y' = \lambda' ~\wedge~ c' = (\lambda' \neq 0) \\
    \end{aligned}\right)
  ~\wedge~ c \neq c'
\end{equation*}

This formula is sent to an SMT-solver. If the solver returns
\textsc{unsat}, meaning that the query is not satisfiable, then the condition
does not differ in both executions and thus is secure. Otherwise, it means that
the outcome of the condition depends on the secret and the solver returns a
counterexample satisfying the insecurity query. %
Here, the SMT-solver Z3~\cite{DBLP:conf/tacas/MouraB08}
answers that the query is \textsc{unsat} and we can conclude that the leak is
secure. %
With the same method, the analysis finds that the leak at line~\ref{line:motivating:mem} is insecure,
and returns two inputs (0,0) and (1,0), respectively leaking 0 and 1, as a
counterexample.

\myparagraph{Limits.} Basic self-composition suffers from two
weaknesses:
\begin{itemize}
  \item It generates insecurity queries at each control-flow instruction and
        memory access. Yet, as seen in the previous example, insecurity queries
        could be spared when expressions do not depend on
        secrets. %
  \item The whole original formula is duplicated so the size of the
        self-composed formula is twice the size of the original formula. Yet,
        because the parts of the program which only depend on public input are
        equal in both executions, the self-composed formula contains
        redundancies that are not exploited.
\end{itemize}

\subsection{Relational Symbolic Execution (RelSE)}  %
RelSE improves over self-composition by maximizing \emph{sharing} between the
pairs of
executions~\cite{DBLP:conf/icse/PalikarevaKC16,farinaRelationalSymbolicExecution2019}.
RelSE models two executions of a program \(P\) in the same symbolic execution
instance, let us call them \(p\) and \(p'\). During RelSE, variables of \(P\)
are mapped to \emph{relational expressions} which are either \emph{pairs} of
expressions or \emph{simple} expressions. Variables that \emph{must be equal} in
\(p\) and \(p'\) (i.e., the low inputs) are represented as \emph{simple}
expressions whereas those that \emph{may be different} (i.e., the secret input)
are represented as \emph{pairs} of expressions. %
Secret-dependencies are propagated (in a conservative way) through symbolic
execution using these relational expressions: if the evaluation of an expression
only involves simple operands, its result will be a simple expression, meaning
that it does not depend on secret, whereas if it involves a pair of
expressions, its result will be a pair of expressions.
This representation offers two main advantages. %
First, this enables sharing redundant parts of \(p\) and \(p'\), reducing the
size of the self-composed formula. Second, variables mapping to simple
expressions cannot depend on secret, which makes it possible to spare
insecurity queries. %

As an example, let us perform RelSE of the toy program in
\Cref{list:motivating}. Variable \lstinline{x} is assigned a pair of
expressions ${\pair{\beta}{\beta'}}$ and \lstinline{y} is assigned a
simple expression %
$\simple{\lambda}$. Note that the precondition that public variables
are equal is now implicit since we use the same symbolic variable in
both executions. At line~\ref{line:motivating:cf}, the conditional expression is evaluated to
$c = \simple{\lambda \neq 0}$ and we need to check that the leakage of
$c$ is secure. Since $c$ maps to a simple expression, we know by
definition that it does not depend on the secret, hence we can spare
the insecurity query.

\newercontent{Finally, when a control-flow instruction depends on a pair of
  expressions ${\pair{\varphi}{\varphi'}}$, an insecurity query
  \(\varphi \neq \varphi'\) is sent to the solver. If it is satisfiable, a
  vulnerability is reported and RelSE continues with the constraint
  ${\varphi = \varphi'}$ so the same vulnerability is not reported twice;
  otherwise the insecurity query is unsatisfiable, meaning that
  ${\varphi = \varphi'}$. %
  In both cases, the value of the control-flow instruction is the same in both
  executions and RelSE only needs to model pairs of executions following the
  same path.}

\emph{RelSE maximizes sharing between both executions and tracks
  secret-dependencies enabling to spare insecurity queries and reduce
  the size of the formula. }

\subsection{Challenge of binary-level analysis} %
Recall that,\bsse{} represents the memory as a special variable of type
\(\memtype\). Consequently, it is not possible to directly store relational
expressions in it. In order to store high inputs at the beginning of the
execution, we have to duplicate it. In other words the \emph{memory is always
  duplicated}. %
Consequently, every \(select\) operation will evaluate to a duplicated
expression, preventing to spare queries in many situations.

As an illustration, consider the compiled version of the previous program, given
in~\Cref{lst:limitation_std_relse}. The steps of RelSE on this program are given
in \cref{fig:relse_motivating2}. %
When the secret input is stored in memory at line~\ref{line:limitation_relse:init_high}, the array
representing the memory is duplicated. This propagates to the load expression in
\dbainline{eax} at line~\ref{line:limitation_relse:eax} and to the conditional expression at line~\ref{line:limitation_relse:cf_leak}. %
Intuitively, at line~\ref{line:limitation_relse:cf_leak}, \dbainline{eax} should be equal to the simple
expression \(\simple{\lambda}\) in which case we could spare the insecurity
query like in the previous example. %
However, because dependencies cannot be tracked in the array representing the
memory, \dbainline{eax} evaluates to a pair of \(select\) expression and we have to
send the insecurity query to the solver.

\input{./ressources/figures/limitation_relse}

\myparagraph{Practical impact.} \Cref{tab:motivating} reports the
performance of constant-time analysis on an implementation of elliptic curve
Curve25519-donna~\cite{DBLP:conf/pkc/Bernstein06}. %
\textit{Both SC and RelSE fail to prove the program secure in less
  than 1h.}  %
\textit{RelSE} does reduce the number of queries compared to
\textit{SC}, but it is not sufficient.

\input{./ressources/tables/motivating}

\myparagraph{Our solution.} To mitigate this issue, we propose
dedicated simplifications for binary-level relational symbolic
execution that allow a precise tracking of secret-dependencies \emph{in
  the memory} (details in~\cref{sec:optims}). In the particular
example of~\cref{tab:motivating}, our prototype \brelse{} \emph{proves that the code is secure} in less than 20 minutes. Our
simplifications simplify all the queries, resulting in a
\(\times 2000\) speedup compared to standard RelSE and SC in terms of
number of instructions explored per second.

%% file: ressources/listings/motivating.tex
\begin{lstlisting}[caption={Toy program with one control-flow leak and one memory leak.},label={list:motivating}]
x = secret_input();
y = public_input();
if (y != 0) return 0; // leak y = 0 <@\label{line:motivating:cf}@>
return tab[z];        // leak z <@\label{line:motivating:mem}@>
\end{lstlisting}

%% file: ressources/listings/limitation_relse.tex
\begin{dba}[caption={Compiled version of \Cref{list:motivating}, where
      \(\pair{\beta}{\beta'}\) (resp.
      \(\simple{\lambda}\)) denotes  a high (resp. low)
      input.},label={lst:limitation_std_relse}]
store ebp-8 $\pair{\beta}{\beta'}$ // store high input <@\label{line:limitation_relse:init_high}@>
store ebp-4 $\simple{y}$      // store low input <@\label{line:limitation_relse:init_low}@>
eax := load ebp-4    // assign $\simple{\lambda}$ to eax <@\label{line:limitation_relse:eax}@>
ite eax ? $\mathtt{l_{1}}$ : $\mathtt{l_{2}}$        // leak $\simple{\lambda = 0}$ <@\label{line:limitation_relse:cf_leak}@>
[...]
\end{dba}

%% file: ressources/figures/limitation_relse.tex
\begin{figure}[!tbp]
  \centering
  
\begin{tabular}{@{}rl@{~~}r@{~~}l@{~~}l@{~~}l@{}}
\((init)\) & \multicolumn{3}{l@{}}{\dbainline{mem} \(\mapsto \simple{\mu_0}\)
             and \dbainline{ebp} \(\mapsto \simple{ebp}\)}\\
   \((1)\) & \dbainline{mem} \(\mapsto \pair{\mu_1}{\mu_1'}\)
             & where & \(\mu_1 \mydef store(\mu_0,ebp-8,\beta)\)
              &  and & \(\mu_1' \mydef store(\mu_0,ebp-8,\beta')\)\\
   \((2)\) & \dbainline{mem} \(\mapsto \pair{\mu_2}{\mu_2'}\)
           &  where & \(\mu_2 \mydef store(\mu_1,ebp-4,\lambda)\)
              & and & \(\mu_2' \mydef store(\mu_1',ebp-4,\lambda)\)\\
   \((3)\) & \dbainline{eax} \(\mapsto \pair{\alpha}{\alpha'}\)
           &  where & \(\alpha \mydef select(\mu_2,ebp-4)\)
              & and & \(\alpha' \mydef select(\mu_2',ebp-4)\)\\
   \((4)\) & \multicolumn{3}{l@{}}{\(leak~\pair{\alpha \neq 0}{\alpha' \neq 0}\)}
\end{tabular}
\caption{RelSE of program in \Cref{lst:limitation_std_relse} where
  \dbainline{mem} is the memory variable, \dbainline{ebp} and \dbainline{eax}
  are registers, \(\mu_0, \mu_1, \mu_1', \mu_2, \mu_2'\) are symbolic array
  variables, and \(ebp, \beta, \beta', \lambda, \alpha, \alpha'\) are symbolic
  bitvector variables}\label{fig:relse_motivating2}
\end{figure}

%% file: ressources/tables/motivating.tex
\begin{table}
  \begin{minipage}[t]{0.55\textwidth}
  \vspace{0pt}
\begin{tabular}{lrrrrc}
  \toprule
  \multicolumn{1}{l}{Version} &
  \multicolumn{1}{c}{\#Instr} &
  \multicolumn{1}{c}{\#Query} &
  \multicolumn{1}{c}{Time} &
  \multicolumn{1}{c}{\#Instr/s} &
  \multicolumn{1}{c}{Status}\\
  \midrule
  \textit{SC} (e.g.\ \cite{DBLP:conf/date/SubramanyanMKMF16})
   & 11k & 9051 &   \hourglass{} & 3 & \hourglass{} \\
  \textit{RelSE} (e.g.\ \cite{farinaRelationalSymbolicExecution2019})
  & 13k & 5486 &   \hourglass{} & 4 & \hourglass{} \\
  \midrule
  \brelse{}
  & 10M &  0 & 1166 & 8576 & \ccmark{} \\
  \bottomrule
\end{tabular}
\end{minipage}\hfill
\begin{minipage}[t]{0.43\textwidth}
  \vspace{0pt}
  \caption{Performance of CT-analysis on \texttt{donna} compiled with
    \texttt{gcc-5.4 -O0}, in terms of number of explored unrolled
    instructions (\#Instr), number of queries (\#Query), execution time in
    seconds (Time), instructions explored per second (\#Instr/s), and status
    (Status) set to secure (\ccmark) or timeout (\hourglass{}) set to
    3600s.}\label{tab:motivating}
\end{minipage}
\end{table}

%% file: concrete_semantics.tex
We present the leakage models  in an intermediate language called
Dynamic Bitvectors Automatas (DBA)~\cite{bardinBINCOAFrameworkBinary2011}.

\subsection{Dynamic Bitvectors Automatas} %
DBA~\cite{bardinBINCOAFrameworkBinary2011}, shown in \cref{table:dba_lang}, is the  representation
used in \binsec{}~\cite{DBLP:conf/wcre/DavidBTMFPM16} to model
programs and perform its analysis. 
\input{./ressources/figures/dba}

Let \(\instrset\) denote the set of instructions and \(\locset\) the
set of program locations. %
A program \(\prog{} : \locset \rightarrow \instrset\) is a map from
locations to instructions. %
Values \texttt{bv} and variables \texttt{v} range over the set of
fixed-size bitvectors \(\bvset{n} := {\{0,1\}}^n\) (set of \(n\)-bit
words). %
A concrete configuration is a tuple
\(\cconf{\locvar}{\cregmap}{\cmem}\) where: %
\begin{itemize}
\item \(\locvar \in \locset\) is the current location, and
  \(\locmap{l}\) returns the current instruction,
\item \(\cregmap : \varset{} \to \bvset{n} \) is a register map that maps
  variables to their bitvector value,
\item \(\cmem : \bvset{32} \to \bvset{8}\) is the memory, mapping 32-bit
  addresses to bytes and accessed by operators \load{} and \store{}.
\end{itemize}
The initial configuration is given by
\(\cconfvar_0 \mydef \cconf{\locvar_0}{\cregmap_0}{\cmem_0}\) with
\(\locvar_0\) the address of the entrypoint of the program, \(r_0\) an
arbitrary register map, and \(m_0\) an arbitrary memory. Let
\(\haltlocset \subseteq \locset\) the set of halting program locations
such that
\(\locvar \in \haltlocset \iff \locmap{\locvar} = \texttt{halt}\).
For the evaluation of indirect jumps, we define a partial one-to-one
correspondence from bitvectors to program locations,
\(\toloc : \bvset{32} \rightharpoonup \locset\). If a bitvector
corresponds to an illegal location (e.g.\ non-executable address),
\(\toloc\) is undefined.

\subsection{Leakage Model}\label{sec:leakage_model}
The behavior of programs is modeled with an instrumented operational semantics
in which each transition is labeled with an explicit notion of leakage. Building
on \citeauthor{DBLP:conf/csfw/BartheGL18}'s
framework~\cite{DBLP:conf/csfw/BartheGL18}, the semantics is
parameterized with leakage functions, which permits to consider several leakage
models.

The set of program leakages, denoted \(\leakset\), is defined according to the
leakage model.
A transition from a configuration \(c\) to a configuration \(c'\)
produces a leakage \(\leakvar \in \leakset\), denoted
\(c \cleval{\leakvar} c'\). %
Analogously, the evaluation of an expression \(e\) in a configuration
\(\cconf{\locvar}{\cregmap}{\cmem}\), produces a leakage
\(\leakvar \in \leakset\), denoted
\(\ceconf{\cregmap}{\cmem}{e} \ceeval{\leakvar} bv\). %
The leakage of a multistep execution is the concatenation of leakages,
denoted \(\concat\), produced by individual steps. We use
\(\cleval{\leakvar}^k\) with \(k\) a natural number to denote \(k\)
steps in the concrete semantics.

\input{./ressources/figures/dba_semantics_full}

The concrete semantics is given in \cref{fig:dba_semantics_full} and
is parameterized with leakage functions %
\(\leakfunc_{\unop}: \bvset{} \to \leakset\), %
\(\leakfunc_{\binop}: \bvset{} \times \bvset{} \to \leakset\),
\(\leakfunc_{@}: \bvset{32} \to \leakset\), %
\(\leakfunc_{pc}: \locset \to \leakset\), %
\(\leakfunc_{\bot}: \locset \to \leakset\), %
\(\leakfunc_{\mu}: \bvset{32} \times \bvset{8} \to \leakset\). %
A \emph{leakage model} is an instantiation of the leakage functions. We consider
the \emph{program counter}, \emph{memory obliviousness}, \emph{size
  noninterference} and \emph{constant-time}, leakage models defined
in~\cite{DBLP:conf/csfw/BartheGL18}. In addition, we define the
\emph{operand noninterference} and \emph{secret-erasure} leakage models.

\myparagraph{Program counter~\cite{DBLP:conf/csfw/BartheGL18}.} The programs
counter leakage model leaks the control flow of the program. %
The leakage of a program is a list of program location:
\(\leakset \mydef List(\locset)\). %
The outcome of conditional jumps and the address of indirect jumps %
is leaked: \(\leakfunc_{pc}(\locvar) = [\locvar]\). %
Other instructions produce an empty leakage.

\myparagraph{Memory
  obliviousness~\cite{DBLP:conf/csfw/BartheGL18}.} The memory
obliviousness leakage model leaks the sequence of memory addresses accessed
along the execution. %
The leakage of a program is a list of 32-bit bitvectors representing addresses
of memory accesses: \(\leakset \mydef List(\bvset{32})\). %
The addresses of memory load and stores are leaked: \(\leakfunc_{@}(e) =
[e]\). %
Other instructions produce an empty leakage.

\myparagraph{Operand noninterference.} The operand noninterference leakage model
leaks the value of operands (or part of it) for specific operators that execute
in non constant-time. %
The leakage of a program is a list of bitvector values:
\(\leakset \mydef List(\bvset{})\). Functions \(\leakfunc_{\unop}\) and
\(\leakfunc_{\binop}\) are defined according to architecture specifics. For
instance, in some architectures, the execution time of shift or rotation
instructions depends on the shift or rotation count\footnote{See
  \url{https://bearssl.org/constanttime.html}}. In this case, we can define
\(\leakfunc_{<<}(bv_1,bv_2) = [bv_2]\). %
Other instructions produce an empty leakage.

\myparagraph{Size
  noninterference~\cite{DBLP:conf/csfw/BartheGL18}.} The size
noninterference leakage model is a special case of operand noninterference where
the size of the operand is leaked. For instance, knowing that the execution time
of the division depends on the size of its operands, we can define
\(\leakfunc_{\div}(bv_1,bv_2) = [size(bv_1),size(bv_2)]\).
  
\myparagraph{Constant-time~\cite{DBLP:conf/csfw/BartheGL18}.} The
constant-time leakage model combines the program counter and the memory
obliviousness security policies. The set of leakage is defined as
\(\leakset \mydef List(\locset~\cup~\bvset{32})\). %
The control flow is leaked %
\(\leakfunc_{pc}(\locvar) = [\locvar]\), %
as well as the memory accesses %
\(\leakfunc_{@}(e) = [e]\). %
Other instructions produce an empty leakage.
Note that some definitions of constant-time also include size
noninterference~\cite{DBLP:conf/csfw/BartheGL18} or operand
noninterference~\cite{BearSSLConstantTimeCrypto}.

\myparagraph{Secret-erasure.} %
The secret-erasure leakage model leaks the index and value of every store
operation---values that are overwritten are filtered-out from the leakage trace
\newcontent{(as we formalize later in \cref{def:secret-erasure})}. %
With regard to secret dependent control-flow, we define a conservative notion of
secret-erasure forbidding to branch on secrets---thus including the program
counter policy.
The leakage of a program is a list of locations and pairs of bitvector values:
\(\leakset \mydef List(\locset~\cup~(\bvset{32} \times \bvset{8}))\).
The control flow is leaked %
\(\leakfunc_{pc}(\locvar) = [\locvar]\), %
as well as the end of the program %
\(\leakfunc_{\bot}(\locvar) = [\locvar]\), %
and the list of store operations %
\(\leakfunc_{\mu}(bv, bv') = [(bv, bv')]\). %
Other instructions produce an empty leakage.

\subsection{Secure program}\label{sec:secure_program} %
Let \(\highvarset \subseteq \varset\) be the set of high (secret)
variables and \(\lowvarset = \varset \setminus \highvarset\) be the
set of low (public) variables. Analogously, we define
\(\highmemset \subseteq \bvset{32}\) (resp.
\(\lowmemset = \bvset{32} \setminus \highmemset\)) as the addresses
containing high (resp.\ low) input in the initial memory. %
The \emph{low-equivalence relation} over concrete configurations
\(\cconfvar\) and \(\cconfvar'\), denoted
\(\cconfvar \loweq \cconfvar'\), is defined as the equality of low
variables and low parts of the memory. %
Formally, two configurations %
\(\cconfvar \mydef \cconf{\locvar}{\cregmap}{\cmem}\) and %
\(\cconfvar' \mydef \cconf{\locvar'}{\cregmap'}{\cmem'}\) are
low-equivalent if and only if %
for all variable \(v \in \lowvarset\),
\(\cregmap\ v = \cregmap'\ v\) and for all address
\(a \in \lowmemset\),
\(\cmem\ a = \cmem'\ a\).

\newcontent{Security is expressed as a form of observational noninterference
  that is parameterized by the leakage model. Intuitively it guarantees that
  low-equivalent configurations produce the same observations, according to the
  leakage model:}
\begin{definition}[Observational noninterference (ONI)]
  A program is observationally noninterferent if and only if for all
  low-equivalent initial configurations \(\cconfvar_0\) and
  \(\cconfvar'_0\), and for all \(k \in
  \mathbb{N}\), %
  \begin{equation*}
    \cconfvar_0 \loweq \cconfvar_0'\ %
    ~\wedge~ \cconfvar_0 \cleval{\leakvar}^k \cconfvar_k %
    ~\wedge~ \cconfvar'_0 \cleval{\leakvar'}^k \cconfvar'_k%
    \implies \filter(\leakvar) = \filter(\leakvar') %
  \end{equation*}
  The property is parameterized by a function,
  \(\filter : \leakset \to \leakset\), that further restricts the
  leakage.
\end{definition}

\begin{definition}[Constant-time]\label{def:ct} A program is
  constant-time (CT) if it is ONI in the constant-time leakage model with
  \(\filter\) set to the identity function.
\end{definition}

\begin{definition}[Secret-erasure]\label{def:secret-erasure}
  A program enforces secret-erasure if it is ONI in the secret-erasure leakage
  model with \(\filter\) set to the identity function for control-flow leakages
  and only leaking store values at the end of the program
  (\(\locvar \in \haltlocset\)), restricting to values that have not been
  overwritten by a more recent store.
  Formally, \(\filter(\leakvar) = \filter'(\leakvar, m_{\varepsilon})\) where
  \(m_{\varepsilon}\) is the empty partial function from \(\bvset{32}\) to
  \(\bvset{8}\) and \(\filter'(\leakvar, m_{acc})\) is defined as:
  \newercontent{\begin{mathpar}
    \inferrule*[left=filter-empty]{}{\filter'(\varepsilon, m_{acc}) = \varepsilon}\and
    \inferrule*[left=filter-store]{}{\filter'((\mathtt{a}, \mathtt{v}) \concat \leakvar, m_{acc}) = filter'(\leakvar, m_{acc}[\mathtt{a} \mapsto \mathtt{v}])}\and
    \inferrule*[left=filter-cf]{\locvar \not\in \haltlocset}{\filter'(\locvar \concat \leakvar, m_{acc}) = \locvar \concat filter'(\leakvar, m_{acc})}\and
    \inferrule*[left=filter-halt]{\mathtt{a_i} \in dom(m_{acc}) \\ \locvar \in \haltlocset}{\filter'(\locvar \concat \leakvar, m_{acc}) = m_{acc}(\mathtt{a_0}) \concat \dots \concat m_{acc}(\mathtt{a_n})}
  \end{mathpar}}
  \newercontent{Intuitively, \(m_{acc}\) is a function used to accumulate values written to
  the memory and leak them at the end of a program. %
  The \textsc{filter-store} rule accumulates a store operation \((a, c)\) from
  the leakage trace into the function \(m_{acc}\). Notice that because
  \(m_{acc}\) is a function, if \(m_{acc}(\mathtt{a})\) is already defined, its
  value will be replaced by \(\mathtt{v}\) after
  \(m_{acc}[\mathtt{a} \mapsto \mathtt{v}]\). The \textsc{filter-cf} rule adds
  control-flow label to the final leakage trace. Finally, the
  \textsc{filter-halt} rule is evaluated when a final location is reached and
  leaks all the store values accumulated in \(m_{acc}\). %
  For example,
  \(\filter((\mathtt{a}, \mathtt{x}) \concat (\mathtt{b}, \mathtt{y}) \concat (\mathtt{a}, \mathtt{z}) \concat \locvar_\bot)\)
  where \(\locvar_\bot \in \haltlocset\) will return the leakage
  \(\mathtt{y} \cdot \mathtt{z}\).}
\end{definition}

%% file: ressources/figures/dba.tex
\begin{figure}[htbp]
  \centering
      \begin{empheq}[box=\widefbox]{gather*}
      \begin{array}{@{}r@{~~}r@{~~}l@{\qquad}r@{~~}r@{~~}l@{}}
        prog  & ::= & \varepsilon ~~|~~ stmt~prog &
        stmt  & ::= & <\locvar,inst>\\
      \end{array}\\
      inst ~~::=~~ \assign{\var{v}}{expr}
             ~~|~~ \store{expr}{expr}
             ~~|~~ \dbaite{e}{\locvar_1}{\locvar_2}
             ~~|~~ \goto{expr}
             ~~|~~ \goto{\locvar}
             ~~|~~ \halt\\
      expr ~~::=~~ \var{v}
             ~~|~~ \bv{}
             ~~|~~ \unop~expr
             ~~|~~ expr~\binop~expr
             ~~|~~ \load{expr}\\
      \begin{array}{@{}r@{~~}r@{~~}l@{\qquad}r@{~~}r@{~~}l@{}}
        \unop & ::= & \neg ~~|~~ - &
        \binop & ::= & + ~~|~~ \times ~~|~~ \leq ~~|~~ \dots \\
      \end{array}
    \end{empheq}
    \protect\caption{The syntax of DBA programs, where \(\locvar, \locvar_{1}\)
      and \(\locvar_{2}\) are program locations, \(\var{v}\) is a variable
      and \(\bv{}\) is a value.}\label{table:dba_lang}
\end{figure}

%% file: ressources/figures/dba_semantics_full.tex
\begin{figure*}[htbp]
\centering

\begin{mybox}[]{Expr}
  \begin{mathpar}

    \inferrule*[left={cst}]{ }{\ceconf{\cregmap}{\cmem}{\bv{}} \eeval{\epsilon} \bv{}}
    \quad
    
    \inferrule*[left={var}]{ }{\ceconf{\cregmap}{\cmem}{\texttt{v}}
      \eeval{\epsilon} \cregmap\ \texttt{v}} \and

    \inferrule*[left={unop}]{
      \ceconf{\cregmap}{\cmem}{e} \eeval{\leakvar} \bv{}%
    }{
      \ceconf{\cregmap}{\cmem}{\unop\ e} \eeval{\leakvar \concat \leakfunc_{\unop}(\bv{})} \unop\ \bv{}
    }

    \inferrule*[left={binop}]{
      \ceconf{\cregmap}{\cmem}{e_1} \eeval{\leakvar_1} \bv{}_{\texttt{1}}\\
      \ceconf{\cregmap}{\cmem}{e_2} \eeval{\leakvar_2} \bv{}_{\texttt{2}}
    }{
      \ceconf{\cregmap}{\cmem}{e_1\ \binop\ e_2}
      \eeval{\leakvar_1 \concat \leakvar_2 \concat \leakfunc_{\binop}(\bv{}_1,\bv{}_2)}
      \bv{}_{\texttt{1}}\ \binop\ \bv{}_{\texttt{2}}
    }

    \inferrule*[left={load}]
    {
      \ceconf{\cregmap}{\cmem}{e} \ceeval{\leakvar} \bv{} \\
    }
    {
      \ceconf{\cregmap}{\cmem}{\load{}\ e} %
      \ceeval{\leakvar \concat \leakfunc_{@}(\bv{})} \cmem~\bv{}
    }\and
  \end{mathpar}
\end{mybox}

\begin{mybox}[]{Instr}
  \begin{mathpar}

    \inferrule*[lab={halt}]{
      \locmap{\locvar} = \halt{}
    }{
      \cconf{\locvar}{\cregmap}{\cmem} \cleval{\leakfunc_{\bot}(\locvar)}
      \cconf{\locvar}{\cregmap}{\cmem}
    }\and
    
    \inferrule*[lab={s\_jump}]{
      \locmap{\locvar} = \goto{\locvar'}
    }{
      \cconf{\locvar}{\cregmap}{\cmem} \cleval{\leakfunc_{pc}(\locvar')}
      \cconf{\locvar'}{\cregmap}{\cmem}
    }\and
    
    \inferrule*[lab={d\_jump}]{
      \locmap{l} = \goto{e} \\
      \ceconf{\cregmap}{\cmem}{e} \ceeval{\leakvar} \bv{} \\
      \locvar' \mydef \toloc(\bv{})
    }
    {
      \cconf{\locvar}{\cregmap}{\cmem} %
      \cleval{\leakvar \concat \leakfunc_{pc}(\locvar')} %
      \cconf{\locvar'}{\cregmap}{\cmem}
    }\and

    \inferrule*[lab={ite-true}]{
      \locmap{l} = \dbaite{e}{\locvar_1}{\locvar_2}\\
      \ceconf{\cregmap}{\cmem}{e} \ceeval{\leakvar} \bv{} \\
      \bv{} \neq 0
    }
    {
      \cconf{\locvar}{\cregmap}{\cmem} %
      \cleval{\leakvar \concat \leakfunc_{pc}(\locvar{}_1)} %
      \cconf{\locvar{}_1}{\cregmap}{\cmem}
    } \mkern-15mu

    \inferrule*[lab={ite-false}]{
      \locmap{l} = \dbaite{e}{\locvar_1}{\locvar_2}\\
      \ceconf{\cregmap}{\cmem}{e} \ceeval{\leakvar} \bv{} \\
      \bv{} = 0
    }
    {
      \cconf{\locvar}{\cregmap}{\cmem} %
      \cleval{\leakvar \concat \leakfunc_{pc}(\locvar{}_2)} %
      \cconf{\locvar{}_2}{\cregmap}{\cmem}
    }\and

    \inferrule*[lab={assign}]{
      \locmap{l} = \assign{\var{v}}{e}\\
      \cconf{\locvar}{\cregmap}{\cmem}{e} \ceeval{t} \bv{}
    }{
      \cconf{\locvar}{\cregmap}{\cmem} \cleval{t}
      \cconf{\locvar+1}{\cregmap[\var{v} \mapsto{} \bv{}]}{\cmem}
    }\mkern-15mu
    
    \inferrule*[lab={store}]{
      \locmap{l} = \store{e}{e'}\\
      \ceconf{\cregmap}{\cmem}{e} \ceeval{\leakvar} \bv{} \\
      \ceconf{\cregmap}{\cmem}{e'} \ceeval{\leakvar'} \bv{}' \\
    }
    { %
      \cconf{\locvar}{\cregmap}{\cmem} %
      \cleval{\leakvar' \concat \leakvar \concat \leakfunc_{@}(\bv{}) \concat \leakfunc_{\mu}(\bv{},\texttt{bv'})} %
      \cconf{\locvar+1}{\cregmap}{\cmem[\bv{} \mapsto \bv']} %
    }
  \end{mathpar}
\end{mybox}

\caption{Concrete evaluation of DBA instructions and
  expressions.}\label{fig:dba_semantics_full}
\end{figure*}

%% file: relse.tex
Binary-level symbolic execution relies on the quantifier-free theory
of fixed-size bitvectors and arrays
(\abv{}~\cite{barrettSMTLIBStandardVersion2017}). %
We let \(\beta\), \(\beta'\), \(\lambda\), \(\varphi\), range over the
set of formulas $\formulaset$ in the \abv{} logic.  A
\emph{relational} formula \(\rel{\varphi}\) is either a \abv{} formula
\(\simple{\varphi}\) or a pair \(\pair{\varphi_l}{\varphi_r}\) of two
\abv{} formulas.  We denote \(\lproj{\rel{\varphi}}\) (resp.\
\(\rproj{\rel{\varphi}}\)), the projection on the left (resp.\ right)
value of \(\rel{\varphi}\). If \(\rel{\varphi} = \simple{\varphi}\),
then \(\lproj{\rel{\varphi}}\) and \(\rproj{\rel{\varphi}}\) are both
defined as \(\varphi\).  Let \(\rlift{\formulaset}\) be the set of
relational formulas and \(\rlift{\bvtype{n}}\) be the set of
relational symbolic bitvectors of size $n$.

\myparagraph{Symbolic configuration.}\label{sec:symbolic-configuration1} %
Our symbolic evaluation restricts to pairs of traces following the same
path---which is sufficient for constant-time and our definition of
secret-erasure. Therefore, a symbolic configuration only needs to consider a single
program location \(l \in Loc\) at any point of the execution. %
A \emph{symbolic configuration} is of the form
\(\iconfold{l}{\regmap}{\smem}{\pc{}}\) where:
\begin{itemize}
\item \(l \in Loc\) is the current program point,
\item \(\regmap{} : \varset{} \rightarrow \rlift{\formulaset}\) is a
  symbolic register map, mapping variables from a set \(\varset{}\) to
  their symbolic representation as a relational formula in
  \(\rlift{\formulaset}\),
\item \(\smem : \memtype \times \memtype\) is the symbolic memory---a
  pair of arrays of values in \(\bvtype{8}\) indexed by addresses in
  \(\bvtype{32}\),
\item \(\pc{} \in \formulaset\) is the path predicate---a conjunction
  of conditional statements and assignments encountered along a path.
\end{itemize}

\myparagraph{Symbolic evaluation} of instructions, denoted
\(\sconfvar \ieval{} \sconfvar'\) where $\sconfvar$ and $\sconfvar'$
are symbolic configurations, is given in
\Cref{fig:eval_instr_sha_full}.
The evaluation of an expression \(expr\) to a relational formula
\(\rel{\varphi}\), is denoted
\(\econfold{\regmap}{\smem}{\pc}{expr} \eeval{} \rel{\varphi}\). %
A model \(M\) assigns concrete values to symbolic variables. %
The satisfiability of a formula \(\pi\) with a model \(M\) is denoted
$M \sat{\pi}$. In the implementation, an SMT-solver is used to determine
satisfiability of a formula and obtain satisfying model, denoted
$M \solver{\pi}$. Whenever the model is not needed for our purposes, we leave it
implicit and simply write $\sat{\pi}$ or $\solver{\pi}$ for satisfiability.

The symbolic evaluation is parameterized by \emph{symbolic leakage predicates} %
\(\sleakfunc_{\unop}, \sleakfunc_{\binop}, \sleakfunc_{@}, \sleakfunc_{dj}, \sleakfunc_{ite}\)
and %
\(\sleakfunc_{\bot}\) %
which are instantiated according to the leakage model (details on the
instantiation will be given in \cref{sec:leakage_predicates}). %
\newcontent{Symbolic leakage predicates take as input a path predicate and
  expressions that can be leaked, and return \(true\) if and only if no secret
  data can leak.} The rules of the symbolic evaluation are guarded by these
symbolic leakage predicates: a rule can only be evaluated if the associated
leakage predicate evaluates to \(true\), \newcontent{meaning that no secret can
  leak}. If a symbolic leakage predicate evaluates to \(false\) then a secret
leak is detected and \newcontent{the analysis is stuck}. Detailed explanations
of (some of) the symbolic evaluation rules follow:

\newcontent{\rulename{cst} is the evaluation of a constant \texttt{bv} and
  returns the corresponding symbolic bitvector as a simple expression
  \(\simple{bv}\).}

\rulename{load} is the evaluation of a load expression. %
It returns a pair of logical \(select\) formulas from the pair of
symbolic memories \(\smem\) (the box in the hypotheses should be
ignored for now, it will be explained in~\cref{sec:optims}). Note that
the returned expression is \emph{always duplicated} as the \(select\)
must be performed in the left and right memories independently.

\rulename{d\_jump} is the evaluation of an indirect jump. %
It finds a concrete value $l'$ for the jump target, and updates the path
predicate and the next location. Note that this rule is nondeterministic as
\(l'\) can be any concrete value satisfying the path constraint. %

\rulename{ite-true} is the evaluation of a conditional jump when the
expression evaluates to \(true\) (the \(false\) case is
analogous). %
If the condition guarding the \(true\)-branch is satisfiable, the rule
updates the path predicate and the next location to explore it.

\rulename{assign} is the evaluation of an assignment. It allocates a
fresh symbolic variable to avoid term-size explosion, and updates the
register map and the path predicate. %
The content of the box in the hypothesis and the rule
\rulename{canonical-assign} should be ignored for now and will be
explained in~\cref{sec:optims}.

\rulename{store} is the evaluation of a store instruction. %
It evaluates the index and value of the store and updates the symbolic
memories and the path predicate with a logical \(store\) operation.

\input{./ressources/figures/eval_instr_sha_full}

\subsection{Security evaluation}\label{sec:security_eval} %
\newcontent{For the security evaluation, we start by defining a general
  predicate, $\secleak$, which takes as an input a path predicate and a
  relational expression that is leaked, and returns \(true\) if and only if no
  secret data can leak (cf.\ \cref{sec:secleak}). Then, we use this $\secleak$
  predicate to instantiate symbolic leakage predicates
  \(\sleakfunc_{\unop}, \sleakfunc_{\binop}, \sleakfunc_{@}, \sleakfunc_{dj}, \sleakfunc_{ite}\)
  and %
  \(\sleakfunc_{\bot}\) according to the leakage model (cf.\
  \cref{sec:leakage_predicates}).}

\subsubsection{Predicate \(\secleak\)}\label{sec:secleak}
We define a predicate
$\secleak : \rlift{\formulaset} \times \formulaset \to Bool$ which ensures that
a relational formula does not differ in its right and left components, meaning
that it can be leaked securely:
\begin{alignat*}{2}
  \secleak(\rel{\varphi}, \pc) &
  &=
  \begin{cases}
    true & {\sf if~} \rel{\varphi} = \simple{\varphi}  \\
    true & {\sf if~} \rel{\varphi} = \pair{\varphi_l}{\varphi_r}
    \wedge \unsatsolver{\pi \wedge \varphi_l \neq \varphi_r}  \\
    false & {\sf otherwise}
  \end{cases} 
\end{alignat*}

By definition, a simple expression \(\simple{\varphi}\) does not depend on
secrets and can be leaked securely. Thus it \emph{spares an insecurity query} to
the solver. %
However, a duplicated expression \(\pair{\varphi_l}{\varphi_r}\) \emph{may}
depend on secrets. Hence \emph{an insecurity query must be sent to the solver}
to ensure that the leak is secure.

\subsubsection{Instantiation of leakage predicates}\label{sec:leakage_predicates}
Symbolic leakage predicates are instantiated according to the concrete leakage
models defined in \cref{sec:leakage_model}. %
Note that the analysis can be easily be extended to other leakage models by
defining symbolic leakage predicates accordingly.

\myparagraph{Program counter.} Symbolic leakage predicates ensure that the
outcome of control-flow instructions and the addresses of indirect jumps are the
same in both executions: %
\(\sleakfunc_{dj}(\pc, \rel{\varphi}) = \secleak(\rel{\varphi}, \pc)\) and
\(\sleakfunc_{ite}(\pc, \rel{\varphi}) = \secleak(\rlift{eq_0}\ \rel{\varphi}, \pc)\)
where \(eq_0\ x\) returns \(true\) if \(x = 0\) and \(false\) otherwise, and
\(\rlift{eq_0}\) is the lifting of \(eq_0\) to relational formulas. Other
symbolic leakage predicates evaluate to true.

\myparagraph{Memory obliviousness.} Symbolic leakage predicates ensure that store and
load indexes are the same in both executions: %
\(\sleakfunc_{@}(\pc, \rel{\varphi}) = \secleak(\rel{\varphi}, \pc)\). %
Other symbolic leakage predicates evaluate to true.

\myparagraph{Operand noninterference.} Symbolic leakage predicates ensure that
operands (or part of them) are the same in both executions for specific
operators that execute in non constant-time. %
For instance, for architectures in which the execution time of shift depends on
the shift count,
\(\sleakfunc_{<<}(\pc, \rel{\varphi},\rel{\phi}) =
\secleak(\rel{\varphi}, \pc)\).
Other symbolic leakage predicates evaluate to true.

\myparagraph{Size noninterference} (special case of operand noninterference).
Symbolic leakage predicates ensure that the size of operands is the same in both
executions for specific operators that execute in non constant-time. For
instance for the division, we have
\(\sleakfunc_{\div}(\pc, \rel{\varphi}, \rel{\psi}) = \secleak(\rlift{size}\ \rel{\varphi}, \pc)\),
where \(size : \bvtype{} \to \bvtype{}\) is a function that returns the size of
a symbolic bitvector and \(\rlift{size}\) its lifting to relational expressions.
Other symbolic leakage predicates evaluate to true.
  
\myparagraph{Constant-time.} This policy is a combination of the program counter
and the memory obliviousness policies. Symbolic leakage predicates
\(\sleakfunc_{dj}\) and \(\sleakfunc_{ite}\) are defined like in the program
counter policy, while \(\sleakfunc_{@}\) is defined like in the memory
obliviousness policy. Other symbolic leakage predicates evaluate to true.

\myparagraph{Secret-erasure.} %
At the end of the program, a symbolic leakage predicate ensures that the parts
of memory that have been written by the program are the same in both executions:
\begin{equation*}
  \sleakfunc_{\bot}(\pc, \smem) =
  \bigwedge\limits_{\iota \in addr(\smem)} \secleak(\pair{select(\lproj{\smem},\iota)}{select(\rproj{\smem},\iota)}, \pc)
\end{equation*}
where \(addr(\smem)\) is the list of store indexes in \(\smem\).

\subsubsection{Specification of high and low input.} %
By default, the content of the memory and registers is low so the user has to
specify memory addresses that initially contain secret inputs. Addresses of high
variables can be specified as offsets from the initial stack pointer
\texttt{esp} (which requires manual reverse engineering), or using dummy
functions to annotate secret variables at source level (which is easier but only
applies to libraries or requires access to source code).

\subsubsection{Bug-finding.} A vulnerability is found when the function
\(\secleak(\rel{\varphi}, \pc)\) evaluates to \emph{false}. In this case, the
insecurity query is satisfiable and the solver returns a model \(M\) such that %
\(M \solver{\pi \wedge (\lproj{\rel{\varphi}} \neq \rproj{\rel{\varphi}})}\).
The model $M$ assigns concrete values to variables that satisfy the insecurity
query. Therefore it can be returned as a concrete counterexample that triggers
the vulnerability, along with the current location of the vulnerability.

\subsection{Optimizations for binary-level symbolic
  execution}\label{sec:optims}

Relational symbolic execution does not scale in the context of
binary-level analysis (see \textit{RelSE} in
\Cref{tab:scale_total_summary}). In order to achieve better
scalability, we enrich our analysis with an optimization, called
\emph{on-the-fly-read-over-write} (\textit{FlyRow} in
\cref{tab:scale_optims}), based on
\emph{read-over-write}~\cite{DBLP:conf/lpar/FarinierDBL18}.  This
optimization simplifies expressions and resolves load operations ahead
of the solver, often avoiding to resort to the duplicated memory and
allowing to spare insecurity queries. %
We also enrich our analysis with two further optimizations, called
\emph{untainting} and \emph{fault-packing} (\textit{Unt} and
\textit{FP} in \cref{tab:scale_optims}), specifically targeting RelSE
for information flow analysis.

\subsubsection{On-the-fly read-over-write}\label{sec:row}

Solver calls are the main bottleneck of symbolic execution, and reasoning about
\(store\) and \(select\) operations in arrays is particularly
challenging~\cite{DBLP:conf/lpar/FarinierDBL18}. Read-over-write
(Row)~\cite{DBLP:conf/lpar/FarinierDBL18} is a simplification for the theory of
arrays that efficiently resolves \(select\) operations. It is particularly
efficient in the context of binary-level analysis where the memory is
represented as an array and formulas contain many \(store\) and \(select\)
operations.
The standard read-over-write optimization~\cite{DBLP:conf/lpar/FarinierDBL18}
has been implemented as a solver-pre-processing, simplifying a formula
before sending it to the solver. While it has proven to be very
efficient to simplify individual formulas of a single
execution~\cite{DBLP:conf/lpar/FarinierDBL18}, we show in \cref{sec:comparison_se}
that it does not scale in the context of relational reasoning, where
formulas model two executions and a lot of queries are sent to the
solver. %

Thereby, we introduce \emph{on-the-fly read-over-write} (\textit{FlyRow}) to
track secret-dependencies in the memory and spare insecurity queries in the
context of information flow analysis. By keeping track of \emph{relational
  \(store\) expressions} along the execution, it can resolve \(select\)
operations---often avoiding to resort to the duplicated memory---and drastically
reduces the number of queries sent to the
solver, %
improving the performance of the analysis.

\myparagraph{Memory Lookup.} %
The symbolic memory can be seen as the history of the successive \(store\)
operations beginning with the initial memory \(\mu_0\). %
Therefore, a memory \(select\) can be resolved by going back up the history and
comparing the index to load, with indexes previously stored. %
Our FlyRow optimization consists in replacing selection in the memory
(\Cref{fig:eval_instr_sha_full}, \rulename{load} rule, boxed hypothesis) by a
new function %
\(\lookup : (\memtype \times \memtype) \times \rlift{\bvtype{32}} \to \rlift{\bvtype{8}}\)
which takes a relational memory and a relational index, and returns the
relational bitvector value stored at that index. %
\iftechreport%
\newercontent{For simplicity, we
  define the function for simple indexes and detail the lifting to relational
  indexes in \cref{app:rel-lookup}}:
\fi
\begin{alignat*}{2}
  \lookup(\smem_0, \iota) =& ~~\pair{select(\lproj{{\smem_0}}, \iota)}{select(\rproj{{\smem_0}}, \iota)}\\
  \lookup(\smem_n, \iota) =&
  \begin{cases}
    \simple{\varphi_l}  & {\sf if~} \compare(\iota,\kappa)  \wedge \compare(\varphi_l,\varphi_r) \\
    \pair{\varphi_l}{\varphi_r}  & {\sf if~} \compare(\iota,\kappa) \wedge \neg\compare(\varphi_l,\varphi_r)\\
    \lookup(\smem_{n-1}, \iota)& {\sf if~} \neg\compare(\iota,\kappa)\\ 
    \pair{select(\lproj{{\smem_n}}, \iota)}{select(\rproj{{\smem_n}}, \iota)}& {\sf if~} \compare(\iota,\kappa) = \bot
  \end{cases} \\
  \text{~where~} \smem_n \mydef& \pair{store(\lproj{{\smem_{n-1}}},\kappa,\varphi_l)}{store(\rproj{{\smem_{n-1}}},\kappa,\varphi_r)} \\
\end{alignat*}
where \(\compare(\iota,\kappa)\) is a comparison function relying on
\emph{syntactic term equality}, which returns true (resp.\ false) only
if \(\iota\) and \(\kappa\) are equal (resp.\ different) in any
interpretation. If the terms are not comparable, it is undefined,
denoted \(\bot\).

\begin{example}[Lookup]\label{ex:lookup}
  Let us consider the memory: \input{./ressources/figures/memory}

  \begin{itemize}
  \item A call to \(\lookup(\rel{\mu}, ebp - 4)\) returns \(\lambda\).
    \item A call to \(\lookup(\rel{\mu}, ebp - 8)\) first compares the indexes
          \([ebp-4]\) and \([ebp-8]\). Because it can determine that these
          indexes are \emph{syntactically distinct}, the function moves to the
          second element, determines the syntactic equality of indexes and
          returns \(\pair{\beta}{\beta'}\).
    \item A call to \(\lookup(\rel{\mu}, esp)\) tries to compare the indexes
          \([ebp-4]\) and \([esp]\). Without further information, the equality
          or disequality of \(ebp\) and \(esp\) cannot be determined, therefore
          the lookup is aborted and the \(select\) operation cannot be
          simplified.
  \end{itemize}
\end{example}

\myparagraph{Term rewriting.} %
To improve the conclusiveness of syntactic equality checks for the
read-over-write, the terms are assumed to be in \emph{normalized} form
\(\beta + o\) where \(\beta\) is a base (i.e., an expression on
symbolic variables) and \(o\) is a constant offset. %
The comparison of two terms \(\beta + o\) and \(\beta' + o'\) in
normalized form can be efficiently computed as follows: if the bases
\(\beta\) and \(\beta'\) are syntactically equal, then return
\(o = o'\), otherwise the terms are not comparable. %
In order to apply \textit{FlyRow}, we normalize all the formulas
created during the symbolic execution using \emph{rewriting rules}
similar as those defined in~\cite{DBLP:conf/lpar/FarinierDBL18}.  An excerpt
of these rules is given in \cref{fig:normalize}. %
Intuitively, these rewriting rules put symbolic variables at the
beginning of the term and the constants at the end
(see~\cref{ex:normalize}).

\input{./ressources/figures/normalize}

\begin{example}[Normalized formula]\label{ex:normalize}
  \(\normalize\ ((eax + 4) + (ebx + 4)) = (eax + ebx) + 8 \)
\end{example}

In order to increase the conclusiveness of \textit{FlyRow}, we also need
\emph{variable inlining}. However, inlining all variables is not a viable option
as it would lead to an exponential term size growth. %
Instead, we define a \emph{canonical form} \(x + o\) where \(x\) is a bitvector
variable, and \(o\) is a constant bitvector offset, and we only inline formulas
that are in canonical form (see rule \rulename{canonical-assign} in
\cref{fig:eval_instr_sha_full}). It enables rewriting of most of the memory
accesses on the stack which, are of the form \lstinline{ebp + bv}, while
avoiding term-size explosion.

\subsubsection{Untainting}\label{sec:untainting} %
After the evaluation of a rule with the predicate $\secleak$ for a
duplicated expression \(\pair{\varphi_l}{\varphi_r} \), we know that
the equality \(\varphi_l = \varphi_r\) holds in the current
configuration. From this equality, we can deduce useful information
about variables that must be equal in both executions. We can then
propagate this information to the register map and memory in order to
spare subsequent insecurity queries concerning these variables. %
For instance, consider the leak of the duplicated expression
\(\pair{x_l + 1}{x_r + 1}\), where \(x_l\) and \(x_r\) are symbolic
variables. If the leak is secure, we can deduce that \(x_l = x_r\) and
replace all occurrences of \(x_r\) by \(x_l\) in the rest of the
symbolic execution.

We define in~\cref{fig:untainting_rules} a function
\(\untaint(\regmap,\smem, \rel{\varphi})\) which takes a register map
\(\regmap\), a memory \(\smem\), and a duplicated expression \(\rel{\varphi}\).
It deduces variable equalities from \(\rel{\varphi}\), propagate them in
\(\regmap\) and \(\smem\), and returns a pair of updated register map and memory
\((\regmap', \smem')\). %
Intuitively, if the equality of variables \(x_l\) and \(x_r\) can be deduced
from \(\secleak(\rel{\varphi}, \pc)\), the \(untaint\) function replaces
occurrences of \(x_r\) by \(x_l\) in the memory and the register map. As a
result, a duplicated expression \(\pair{x_l}{x_r}\) would be replaced by the
simple expression \(\simple{x_l}\) in the rest of the execution\footnote{We
  implement untainting with a cache of ``untainted variables'' that are
  substituted in the program copy during symbolic evaluation of expressions.}.%

\input{./ressources/figures/untainting}

\subsubsection{Fault-packing}\label{sec:fp} %
Symbolic evaluation generates a large number of insecurity checks for
some leakage models (e.g.\ memory obliviousness, constant-time). The
fault-packing (\textit{FP}) optimization gathers these insecurity
checks along a path and postpones their resolution to the end of the
basic block.

\begin{example}[Fault-packing]
  For example, let us consider a basic-block with a path predicate
  \(\pc\). If there are two memory accesses along the basic block that
  evaluate to \(\pair{\lproj{\varphi}}{\rproj{\varphi}}\) and
  \(\pair{\lproj{\phi}}{\rproj{\phi}}\), we would normally generate
  two insecurity queries
  \((\pc \wedge \lproj{\varphi} \neq \rproj{\varphi})\) and
  \((\pc \wedge \lproj{\phi} \neq \rproj{\phi})\)---one for each
  memory access. Fault-packing regroups these checks into a single
  query
  \(\big(\pc \wedge ((\lproj{\varphi} \neq \rproj{\varphi}) \lor
  (\lproj{\phi} \neq \rproj{\phi}))\big)\) sent to the solver at the
  end of the basic block.
\end{example}

This optimization reduces the number of insecurity queries sent to the solver
and thus helps improving performance. However it degrades the precision of the
counterexample: while checking each instruction individually precisely points to
vulnerable instructions, fault-packing reduces accuracy to vulnerable basic
blocks only. Note that even though disjunctive constraints are usually harder to
solve than pure conjunctive constraints, those introduced by \textit{FP} are
very simple---they are all evaluated under the same path predicate and are not
nested. Therefore, they never end up in a performance degradation
(see~\cref{tab:scale_optims}).

\subsection{Theorems}~\label{sec:proofs} %
\iftechreport%
Theorems and proof are given for the constant-time property. Adaptation of the
theorems and proofs for other leakage models are discussed in
\cref{sec:discussion_proofs}. %
\fi

In order to define properties of our symbolic execution, we use
$\cleval{}^k$ (resp.\ $\ieval{}^k$), with $k$ a natural number, to
denote $k$ steps in the concrete (resp.\ symbolic) evaluation.
\begin{proposition}\label{hyp:ct_upk}
  If a program \(\prog{}\) is constant-time up to \(k\) then for all
  \(i \leq k\), \(\prog{}\) is constant-time up to \(i\).
\end{proposition}

\begin{hypothesis}\label{hyp:abvconc}
  Through this section we assume that theory \abv{} is correct and complete
  w.r.t.\ our concrete evaluation.
\end{hypothesis}

The satisfiability problem for the theory \abv{} is
decidable~\cite{DBLP:journals/toplas/NelsonO79}. Therefore we make the
following hypothesis on the solver:
\begin{hypothesis}\label{hyp:solver}
  We suppose that the SMT solver for \abv{} is correct, complete and
  always terminates. Therefore for a \abv{} formula \(\varphi\),
  \(M \sat \pc \iff M \solver \pc\).
\end{hypothesis}

\begin{hypothesis}\label{hyp:stuck}
  We assume that the program \(\prog{}\) is defined on all locations computed
  during the symbolic execution---notably by the function \(\toloc\) in rule
  \rulename{d\_jump}. Under this hypothesis, and because the solver always
  terminates (\cref{hyp:solver}), symbolic execution is stuck if and only if a
  leakage predicate evaluates to false. In this case, an expression
  \(\rel{\varphi}\) is leaked such that \(\secleak(\rel{\varphi}, \pc)\)
  evaluates to \(false\) and the solver returns a model \(M\) such that
  \({M \sat \pc \wedge (\lproj{\rel{\varphi}} \neq \rproj{\rel{\varphi}})}\)
  (from \cref{hyp:solver}). %
\end{hypothesis}

\begin{proposition}\label{hyp:deterministic}
  Concrete semantics is deterministic, c.f.\ rules of the concrete
  semantics in \cref{fig:dba_semantics_full}.
\end{proposition}

\begin{hypothesis}\label{hyp:concrete-stuck}
  We restrict our analysis to safe programs \newcontent{(e.g.\ no division by 0,
    illegal indirect jump, segmentation fault)}. Under this hypothesis, concrete
  execution never gets stuck. %
\end{hypothesis}

\begin{definition}[\(\concsym{p}{M}\)]\label{def:concsym}
  We define a concretization relation $\concsym{p}{M}$ between
  concrete and symbolic configurations, where \(M\) is a model and
  \(p \in \{l,r\}\) is a projection on the left or right side of a
  symbolic configuration.  Intuitively, the relation
  $c\! \concsym{p}{M}\! s$ is the concretization of the \(p\)-side of
  the symbolic state \(s\) with the model \(M\).
  Let \(c \mydef \cconf{\locvar_1}{\cregmap}{\cmem}\) and
  \(s \mydef \iconfold{\locvar_2}{\regmap}{\smem}{\pc}\). Formally
  $c \concsym{p}{M} s$ holds iff \(M \sat \pc\),
  \(\locvar_1 = \locvar_2\) and for all expression \(e\), either the
  symbolic evaluation of \(e\) gets stuck or we have
  \begin{equation*}
    \econfold{\regmap}{\smem}{\pc}{e} \eeval{} \rel{\varphi} ~\wedge~ %
    (M(\proj{\rel{\varphi}}) = \mathtt{bv} \iff c~e \ceeval{} \mathtt{bv}) %
  \end{equation*}
\end{definition}

Notice that because both sides of an initial configuration \(s_0\) are
low-equivalent, the following proposition holds:
\begin{proposition}\label{prop:loweq}
  For all concrete configurations \(\cconfvar_0\) and \(\cconfvar_0'\)
  such that
  \(\cconfvar_0 \concsym{l}{M} s_0 ~\wedge~ \cconfvar'_0
  \concsym{r}{M} s_0\), then \(\cconfvar_0 \loweq \cconfvar_0'\).
\end{proposition}

The following lemma expresses that when the symbolic evaluation is
stuck on a state \(s_k\), there exist concrete configurations derived
from \(s_k\) which produce distinct leakages.
\begin{restatable}{lemma}{stuckinsecure}\label{lemma:stuckinsecure}
  Let \(s_k\) be a symbolic configuration obtained after \(k\) steps. If \(s_k\)
  is stuck, then there exists a model \(M\) such that for each concrete
  configurations \(c_k \concsym{l}{M} s_k\) and \(c_k' \concsym{r}{M} s_k\), the
  executions from \(c_k\) and \(c_k'\) produce distinct leakages.
\end{restatable}
\begin{proofoverview}
  \iftechreport{}(Full proof in \cref{app:stuckinsecure}) \fi %
  The proof goes by case analysis on the symbolic evaluation of \(s_{k}\).
  \newcontent{Let \(s_{k}\) be a symbolic configuration that is stuck (i.e., a
    symbolic leakage predicate evaluates to \(false\) with a model \(M\)), then
    \(s_{k}\) can be concretized using the model \(M\), producing concrete
    states \(c_k\) and \(c_k'\) such that \(c_{k} \cleval{\leakvar} c_{k+1}\)
    and \(c_{k}' \cleval{\leakvar'} c_{k+1}'\). Finally, because the symbolic
    leakage model does not over-approximate the concrete leakage, i.e., each
    symbolic leak corresponds to a concrete leak, we have
    \(\leakvar \neq \leakvar'\).}
\end{proofoverview}

The following lemma expresses that when symbolic evaluation does not
get stuck up to \(k\), then for each pair of concrete executions
following the same path up to \(k\), there exists a corresponding
symbolic execution.
\begin{restatable}[]{lemma}{completenesslemma}\label{lemma:completeness}
  Let $s_0$ be a symbolic initial configuration for a program $P$ that does not
  get stuck up to \(k\). For every concrete states $c_0$, $c_k$, $c_0'$, $c_k'$
  and model $M$ such that
  ${c_0 \concsym{l}{M} s_0} ~\wedge~ {c_0' \concsym{r}{M} s_0}$, %
  if $c_0 \cleval{\leakvar}^k c_k$ and $c_0' \cleval{\leakvar'}^k c_k'$ follow
  the same path, then there exists a symbolic configuration \(s_k\) and a model
  \(M'\) such that: %
  \[s_0 \ieval{}^k s_k ~\wedge~ %
    c_k \concsym{l}{M'} s_k ~\wedge~ c_k' \concsym{r}{M'} s_k\]
\end{restatable}

\begin{proofoverview}
  \iftechreport{}(Full proof in \cref{app:completenesslemma}) \fi %
  The proof goes by induction on the number of steps \(k\). For each concrete
  step \(c_{k-1} \ceval{} c_{k}\) and \(c_{k-1}' \ceval{} c_{k}'\), we show
  that, as long as they follow the same path, there is a symbolic step from
  \(s_{k-1}\) to a state \(s_{k}'\) that models \(c_{k}\) and \(c_{k}'\). This
  follows from the fact that our symbolic execution does not make
  under-approximations.
\end{proofoverview}

\subsubsection{Correctness of RelSE}
The following theorem claims the correctness of our symbolic
execution, stating that for each symbolic execution and model \(M\)
satisfying the path predicate, the concretization of the symbolic
execution with \(M\) corresponds to a valid concrete execution (no
over-approximation).

\begin{restatable}[Correctness of RelSE]{theorem}{correctness}\label{thm:correctness}
  For every symbolic configurations $s_0$, $s_k$ such that
  \(s_0 \ieval{}^k s_k\) and for every concrete configurations
  \(c_0\), \(c_k\) and model \(M\), such that
  \(c_0 \concsym{p}{M} s_0\) and \(c_k \concsym{p}{M} s_k\),
  there exists a concrete execution \(c_0 \cleval{}^k c_k\).
\end{restatable}
\begin{proofoverview} %
  \iftechreport{}(Full proof in \cref{app:correctness}) \fi %
  The proof goes by induction on the number of steps \(k\). For each symbolic step
  \(s_{k-1} \ieval{} s_{k}\) and model \(M_{k}\) such that
  \(c_{k-1} \concsym{p}{M_{k}} s_{k-1}\) and \(c_{k} \concsym{p}{M_{k}} s_{k}\),
  there exists a step \(c_{k-1} \ceval{} c_{k}\) in concrete execution. For each
  rule, we show that there exists a unique step from \(c_{k-1}\) to a state
  \(c_{k}'\) (from \cref{hyp:concrete-stuck,hyp:deterministic}), and, because
  there is no over-approximation in symbolic execution, \(c_{k}'\) satisfies
  \(c_{k}' \concsym{p}{M_{k}} s_{k}\).
\end{proofoverview}

\subsubsection{Correct bug-finding for CT}
The following theorem expresses that when the symbolic execution gets
stuck, then the program is not constant-time.
\begin{restatable}[Bug-Finding for CT]{theorem}{bugfinding}\label{thm:bf}
  Let $s_0$ be an initial symbolic configuration for a program $\prog$. If
  symbolic evaluation gets stuck in a configuration \(s_k\) then $\prog$
  is not constant-time at step \(k\). Formally, if there is a symbolic
  evaluation \(s_0 \ieval{}^k s_k\) such that \(s_k\) is stuck, then
  there exists a model \(M\) and concrete configurations
  \(\cconfvar_0 \concsym{l}{M} s_0\), %
  \(\cconfvar_0' \concsym{r}{M} s_0 \), %
  \(\cconfvar_k \concsym{l}{M} s_k \) and %
  \(\cconfvar_k' \concsym{r}{M} s_k\) such that, %
  \begin{equation*}%
    \cconfvar_0 \loweq \cconfvar_0' ~\wedge~%
    \cconfvar_0  \cleval{\leakvar}^k \cconfvar_k \cleval{\leakvar_{k}} \cconfvar_{k+1} ~\wedge~ %
    \cconfvar_0' \cleval{\leakvar'}^k \cconfvar'_k \cleval{\leakvar_{k}'} \cconfvar_{k+1} %
    \wedge \leakvar_{k} \neq \leakvar_{k}' %
  \end{equation*}
\end{restatable}
\begin{proof} Let us consider symbolic configurations \(s_0\) and \(s_k\) such
  that \(s_0 \ieval{}^k s_k\) and \(s_k\) is stuck. %
  From \cref{lemma:stuckinsecure}, there is a model \(M\) and concrete
  configurations \(c_k\) and \(c_k'\) such that \(c_{k} \concsym{l}{M} s_{k}\)
  and \(c_{k}' \concsym{r}{M} s_{k}\), and \(c_{k} \cleval{\leakvar_k} c_{k+1}\)
  and \(c_{k}' \cleval{\leakvar_k'} c_{k+1}'\) with
  \(\leakvar_k \neq \leakvar_k'\). %
  Additionally, let \(c_0, c_0'\) be concrete configurations such that
  \(c_0 \concsym{l}{M} s_0\) and \(c_0' \concsym{r}{M} s_0\). %
  From \cref{prop:loweq}, we have \(c_0 \loweq c_0'\), and from \cref{thm:correctness},
  there are concrete executions \(c_0 \cleval{\leakvar}^{k} c_{k}\) and
  \(c_0' \cleval{\leakvar'}^{k} c_{k}'\). %
  Therefore, we have
  \(\cconfvar_0 \cleval{\leakvar}^k \cconfvar_k \cleval{\leakvar_{k}} \cconfvar_{k+1}\)
  and %
  \(\cconfvar_0' \cleval{\leakvar'}^k \cconfvar'_k \cleval{\leakvar_{k}'} \cconfvar_{k+1}'\)
  with \(c_0 \loweq c_0'\) and \(\leakvar_k \neq \leakvar_k'\), meaning that
  \(\prog\) is not constant-time at step \(k\).
\end{proof}

\subsubsection{Relative completeness of RelSE}
The following theorem claims the completeness of our
symbolic execution relatively to an initial symbolic state. If the
program is constant-time up to \(k\), then for each pair of concrete
executions up to \(k\), there exists a corresponding symbolic
execution (no under-approximation).
Notice that our definition of completeness differs from standard definitions of
completeness in SE~\cite{DBLP:journals/cacm/CadarS13}. Here, completeness up to
\(k\) only applies to programs that are constant-time up to \(k\). This directly
follows from the fact that our symbolic evaluation blocks on errors while
concrete execution continues.

\begin{restatable}[Relative Completeness of
  RelSE]{theorem}{completeness}\label{thm:completeness}
  Let \(P\) be a program constant-time up to \(k\) and $s_{0}$ be a
  symbolic initial configuration for $P$. For every concrete states
  $c_0$, $c_k$, $c_0'$, $c_k'$, and model $M$ such that
  ${c_0 \concsym{l}{M} s_0} ~\wedge~ {c_0' \concsym{r}{M} s_0}$, %
  if $c_0 \cleval{\leakvar}^k c_k$ and $c_0' \cleval{\leakvar}^k c_k'$
  then there exists a symbolic configuration \(s_k\) and a model
  \(M'\) such that: %
  \[s_0 \ieval{}^k s_k ~\wedge~ %
    c_k \concsym{l}{M'} s_k ~\wedge~ c_k' \concsym{r}{M'} s_k\]
\end{restatable}

\begin{proof} First, note that from \cref{thm:bf} and the hypothesis
  that \prog{} is constant-time up to \(k\), we know that symbolic
  evaluation from \(s_0\) does not get stuck up to \(k\). Knowing
  this, we can apply \cref{lemma:completeness} which directly entails
  \cref{thm:completeness}.
\end{proof}

\subsubsection{Correct bounded-verification for CT}
Finally, we prove that if symbolic execution does not get stuck due to a
satisfiable insecurity query, then the program is constant-time.
\begin{restatable}[Bounded-Verification for
  CT]{theorem}{boundedverif}\label{thm:bv}
  Let $s_0$ be a symbolic initial configuration for a program $P$. If
  the symbolic evaluation does not get stuck, then $P$ is
  constant-time w.r.t.\ $s_0$. Formally, if for all $k$,
  $s_0 \ieval{}^k s_k$ then for all initial configurations
  \(\cconfvar_0\) and \(\cconfvar_0'\) and model \(M\) such that
  \(\cconfvar_0 \concsym{l}{M} s_0\), and
  \(\cconfvar'_0 \concsym{r}{M} s_0\),
  \begin{equation*}
    \cconfvar_0  \cleval{\leakvar}^k \cconfvar_k ~\wedge~ %
    \cconfvar_0' \cleval{\leakvar'}^k \cconfvar'_k %
    \implies \leakvar = \leakvar'
  \end{equation*}
  Additionally, if \(s_0\) is fully symbolic, then \(P\) is
  constant-time.
\end{restatable}
\begin{proofoverview} \iftechreport{}(Full proof in \cref{app:bv}) \fi %
  The proof goes by induction on the number of steps. If the program is constant-time up
  to \(k-1\) (induction hypothesis) then from \cref{lemma:completeness} there is
  a symbolic execution for any configurations \(c_{k-1}\) and \(c_{k-1}'\). If
  these configurations produce distinct leakages, then symbolic execution stuck
  at step at step \(k-1\) which is a contradiction. This relies on the fact that
  the symbolic leakage model does not under-approximate the concrete leakage.
\end{proofoverview}

\iftechreport%
\subsection{Adapting theorems and proofs for other leakage
  models}\label{sec:discussion_proofs}
Theorems and proofs in \cref{sec:proofs} are given for the
constant-time property. In this section we discuss how the theorems
and proofs given in \cref{sec:proofs} can be adapted to other leakage
models.
 
Correctness of our symbolic execution (\cref{thm:correctness}) holds
regardless of the leakage model considered. Indeed, we showed that our
symbolic execution makes no over-approximation, without using the
leakage model. Moreover, we can show that (\cref{thm:correctness})
still holds for other leakage models because symbolic leakage
predicates cannot remove constraints from the symbolic state (and
therefore cannot introduce over-approximations).

Bug-finding (\cref{thm:bf}) can also be easily adapted to other leakage models
as long as the symbolic leakage model does not over-approximate the concrete
leakage model. In particular, it still holds for secret-erasure. The adaptation
of \cref{thm:bf} to secret-erasure only requires to show that
\cref{lemma:stuckinsecure} holds for the \rulename{halt} rule. %

Completeness (\cref{thm:completeness}) follows from \cref{lemma:completeness}
and \cref{thm:bf} and thus can be adapted to other leakage models on two
conditions.
First, because our symbolic semantics is blocking on errors, it only applies to
secure programs and its proof relies on the absence of false alarm---which is
given as long as the symbolic leakage model does not over-approximate the
concrete leakage model (\cref{thm:bf}). %
Second, \cref{lemma:completeness} only applies to pairs of concrete executions
following the same path. Therefore, \cref{thm:completeness} only holds for
leakage models leaking the control-flow (i.e., that include the program counter
leakage model). Note that these two conditions are met in the case of
secret-erasure.

Bounded-verification (\cref{thm:bv}) can be adapted to other leakage models on
two conditions. First, because it builds on \cref{lemma:completeness} which only
applies to pairs of concrete executions following the same path, it only holds
for leakage models leaking the control-flow (i.e., that include the program
counter leakage model). Second, it requires to show that the symbolic leakage
model does not under-approximate the concrete leakage model: if a leakage occurs
in concrete execution then this leakage is captured in symbolic execution. These
conditions hold for our definition of secret-erasure, we must just adapt the
proof for the \rulename{halt} rule as the \(\filter\) function delays the
leakage of store values upon termination.
\fi

%% file: ressources/figures/eval_instr_sha_full.tex
\begin{figure*}[htbp]
\begin{mybox}[]{Expr}
  \begin{mathpar}

    \inferrule*[left={cst}]{ }{\econfold{\regmap}{\smem}{\pc}{\var{bv}} \eeval{} \simple{bv}}
    \quad
    \inferrule*[left={var}]{ }{\econfold{\regmap}{\smem}{\pc}{\var{v}}
      \eeval{} \regmap\ \texttt{v}} \and

    \inferrule*[left={unop}]{
      \econfold{\regmap}{\smem}{\pc}{e} \eeval{} \rel{\phi} \\
      \rel{\varphi} \mydef{} \sunop{}\ \rel{\phi}\\
      \sleakfunc_{\unop}(\pc, \rel\phi)
    }{
      \econfold{\regmap}{\smem}{\pc}{\unop{}\ e} \eeval{} \rel{\varphi}
    }

    \inferrule*[left={binop}]{
      \econfold{\regmap}{\smem}{\pc}{e_1} \eeval{} \rel{\phi}\\
      \econfold{\regmap}{\smem}{\pc}{e_2} \eeval{} \rel{\psi}\\
      \rel{\varphi} \mydef{} \rel{\phi{}}\ \sbinop{}\ \rel{\psi}\\
      \sleakfunc_{\binop}(\pc, \rel\phi, \rel\psi)
    }{
      \econfold{\regmap}{\smem}{\pc}{e_1\ \binop{}\ e_2} \eeval{} \rel{\varphi}
    }

    \inferrule*[left={load}]{
      \econfold{\regmap}{\smem}{\pc}{e_{idx}} \eeval{} \rel{\iota}\\
      \rulebox{\rel{\varphi} \mydef{} \pair{select(\lproj{\smem}, \lproj{\rel{\iota}}) }{
          select(\rproj{\smem}, \rproj{\rel{\iota}})}} \\
      \sleakfunc_{@}(\pc, \rel{\iota})
    }{
      \econfold{\regmap}{\smem}{\pc}{\load{}\ e_{idx}} \eeval{} \rel{\varphi}
    }\and
  \end{mathpar}
\end{mybox}

\begin{mybox}[]{Instr}
  \begin{mathpar}

    \inferrule*[left={halt}]{
      \locmap{l} = \halt{}\\
      \sleakfunc_{\bot}(\pc, \smem)
    }{
      \iconfold{l}{\regmap}{\smem}{\pc{}} \ieval{}
      \iconfold{l}{\regmap}{\smem}{\pc{}}
    }\and

    \inferrule*[left={s\_jump}]{
      \locmap{l} = \goto{l'}
    }{
      \iconfold{l}{\regmap}{\smem}{\pc{}} \ieval{}
      \iconfold{l'}{\regmap}{\smem}{\pc{}}
    }\and

    \inferrule*[left={d\_jump}]{
      \locmap{l} = \goto{e}\\
      \econfold{\regmap}{\smem}{\pc}{e}  \eeval{} \rel{\varphi}  \\
      M\solver{\pi{} \wedge \lproj{\rel{\varphi{}}} = \rproj{\rel{\varphi{}}}}  \\ %
      l' \mydef \toloc(M(\lproj{\rel{\varphi{}}})) \\
      \pi{}' \mydef{} \pi{} \wedge{} (\lproj{\rel{\varphi{}}} = \rproj{\rel{\varphi{}}} = M(\lproj{\rel{\varphi{}}}))\\
      \sleakfunc_{dj}(\pc, \rel{\varphi})\\
    }{
      \iconfold{l}{\regmap}{\smem}{\pc} \ieval{}
      \iconfold{l'}{\regmap}{\smem}{%
       \pi{}'}
    }\and

    \inferrule*[lab={ite-true}]{
      \locmap{l} = \dbaite{e}{l_{true}}{l_{false}}\\
      l' \mydef  l_{true} \\
      \econfold{\regmap}{\smem}{\pc}{e } \eeval{} \rel{\varphi} \\
      \pc{}' \mydef{} \pc{} \wedge{} %
      (\lproj{\rel{\varphi{}}} \neq 0) \wedge (\rproj{\rel{\varphi{}}} \neq 0)\\
      \solver{\pc'} \\
      \sleakfunc_{ite}(\pc, \rel{\varphi})
   }{
      \iconfold{l}{\regmap}{\smem}{\pc{}} \ieval{}
      \iconfold{l'}{\regmap}{\smem}{\pc{}'}
    }\and

    \inferrule*[lab={ite-false}]{
      \locmap{l} = \dbaite{e}{l_{true}}{l_{false}}\\
      l' \mydef  l_{false} \\
      \econfold{\regmap}{\smem}{\pc}{e } \eeval{} \rel{\varphi} \\
      \pc{}' \mydef{} \pc{} \wedge{} %
      (\lproj{\rel{\varphi{}}} = 0) \wedge (\rproj{\rel{\varphi{}}} = 0)\\
      \solver{\pc'} \\
      \sleakfunc_{ite}(\pc, \rel{\varphi})
   }{
      \iconfold{l}{\regmap}{\smem}{\pc{}} \ieval{}
      \iconfold{l'}{\regmap}{\smem}{\pc{}'}
    }\and

    \inferrule*[left={assign}]{
      \locmap{l} = \assign{\var{v}}{e}\\
      \econfold{\regmap}{\smem}{\pc}{e} \eeval{} \rel{\varphi} \\
      \rulebox{\neg \canonical(\rel{\varphi})} \\
      \rel{\varphi}' \mydef \fresh\\
      \regmap' \mydef{} \regmap[v \mapsto{} \rel{\varphi}'] \\
      \pc' \mydef{} \pc{} \wedge{} (\lproj{\rel{\varphi}}' = \lproj{\rel{\varphi}})
      \wedge{} (\rproj{\rel{\varphi}}' = \rproj{\rel{\varphi}})
    }{
      \iconfold{l}{\regmap}{\smem}{\pc{}} \ieval{}
      \iconfold{l+1}{\regmap'}{\smem}{\pc'}
    }\and

    \inferrule*[left={canonical-assign}]{
      \locmap{l} = \assign{\var{v}}{e}\\
      \econfold{\regmap}{\smem}{\pc}{e} \eeval{} \rel{\varphi} \\
      \canonical(\rel{\varphi}) \\
      \regmap' \mydef{} \regmap[v \mapsto{} \rel{\varphi}]
    }{
      \iconfold{l}{\regmap}{\smem}{\pc{}} \ieval{}
      \iconfold{l+1}{\regmap'}{\smem}{\pc}
    }\and
    
    \inferrule*[lab={store}]{
      \locmap{l} = \store{e_{idx}}{e_{val}} \\
      \econfold{\regmap}{\smem}{\pc}{e_{idx}}  \eeval{} \rel{\iota} \\
      \econfold{\regmap}{\smem}{\pc}{e_{val}} \eeval{} \rel{\nu} \\
      \smem' \mydef{}  \pair{ store(\lproj{\smem},\lproj{\rel{\iota}},\lproj{\rel{\nu}})}{store(\rproj{\smem},\rproj{\rel{\iota}},\rproj{\rel{\nu}})} \\
      \pc' \mydef{} \pc{} \wedge \lproj{\smem}' = store(\lproj{\smem},\lproj{\rel{\iota}},\lproj{\rel{\nu}})
        \wedge \rproj{\smem}' = store(\rproj{\smem},\rproj{\rel{\iota}},\rproj{\rel{\nu}})\\
      \sleakfunc_{@}(\pc, \rel{\iota})
    }{
      \iconfold{l}{\regmap}{\smem}{\pc{}} \ieval{}
      \iconfold{l + 1}{\regmap}{\smem'}{\pc{}'}
    }
  \end{mathpar}
\end{mybox}
\caption{Symbolic evaluation of DBA instructions and expressions where
  \(\fresh\) returns a fresh symbolic variable,
  \(\canonical(\protect\rel{\varphi})\) is \(true\) if \(\protect\rel{\varphi}\)
  is in canonical form; and \(\sunop\) (resp. \(\sbinop\)) is the logical
  counterpart of the concrete operator \(\unop\) (resp. \(\binop\)), lifted to
  relational expressions.}
\label{fig:eval_instr_sha_full}
\end{figure*}

%% file: ressources/figures/memory.tex
\tikzstyle{box1}=[draw, fill=gray!10, text centered, minimum width = 2em, minimum height=1.8em,scale=0.8]
\tikzstyle{box2}=[draw, text centered, minimum width = 1.5em, minimum height=1.8em,scale=0.8]
\begin{center}
  \begin{tikzpicture}[node distance=0cm,outer sep = 0pt]
    \node (I) [] {};
    \node (L) [left=0pt of I] {\(\smem ~~=\)};
    \node (S)  [right=0pt of L] {};
    \node (C)  [box1,right = 0pt of S] {\(ebp-4\)};
    \node (C1) [box2,anchor=north west] at (C.north east) {$\simple{\lambda}$};
    \node (B)  [box1,right = 1.8em of C1] {\(ebp-8\)};
    \node (B1) [box2,anchor=north west] at (B.north east) {$\pair{\beta}{\beta'}$};
    \node (D)  [box1,right = 1.8em of B1] {\(esp\)};
    \node (D1) [box2,anchor=north west] at (D.north east) {$\simple{ebp}$};
    \node (A)  [right = 1.8em of D1] {\([~]\)};
    \draw[->] (C1) -- (B);
    \draw[->] (B1) -- (D);
    \draw[->] (D1) -- (A);
  \end{tikzpicture}
\end{center}

%% file: ressources/figures/normalize.tex
\begin{figure}[ht]
  \begin{equation*}    
    \begin{aligned}
      \normalize\ (c + t) &= t + c \\
      \normalize\ ((t + c) + c') &= t + (c + c')\\
      \normalize\ ((t + c) + (t' + c')) &= (t + t') + (c + c') \\     
    \end{aligned}\qquad
    \begin{aligned}
      \normalize\ (-(t + c)) &= (-t) - c \\
      \normalize\ ((t + c) + t') &= (t + t') + c \\
      \normalize\ ((t + c) - (t' + c')) &= (t - t') + (c - c')
    \end{aligned}
  \end{equation*}  
  \caption{Rewriting rules for normalization (non exhaustive).  All
    expressions belong to the set \(\bvtype{}\) where \(c, c'\) are
    bitvector constants and \(t, t'\) are arbitrary bitvector
    terms. Note that \((c + c')\) is a constant value, not a
    term.}\label{fig:normalize}
\end{figure}

%% file: ressources/figures/untainting.tex
\begin{figure}[ht]
  \begin{equation*}    
  \begin{aligned}
    untaint(\regmap,\smem,\pair{x_l}{x_r}) &= (\regmap[x_r \backslash x_l],\smem[x_r \backslash x_l])\\
    \begin{rcases}
      untaint(\regmap,\smem,\pair{\neg t_l}{\neg t_r})\\
      untaint(\regmap,\smem,\pair{-t_l}{-t_r})\\
    \end{rcases} &= untaint(\regmap,\smem,\pair{t_l}{t_r})
  \end{aligned}\qquad
  \begin{aligned}
    \begin{rcases}
      untaint(\regmap,\smem,\pair{t_l + k}{t_r + k})\\
      untaint(\regmap,\smem,\pair{t_l - k}{t_r - k})\\
      untaint(\regmap,\smem,\pair{t_l :: l}{t_r :: l})\\
    \end{rcases} &= untaint(\regmap,\smem,\pair{t_l}{t_r})
  \end{aligned}
\end{equation*}  
\caption{Untainting rules where \(x_l, x_r\) are bitvector variables of the same
  size, \(t_l, t_r, k, l\) are bitvector terms such that \(t_l, t_r, k\) have
  the same size, \(::\) indicates the concatenation of bitvectors, and
  $f[x_r \backslash x_l]$ indicates that the variable \(x_r\) is substituted
  with \(x_l\) in $f$.}\label{fig:untainting_rules}
\end{figure}

%% file: expes.tex
\myparagraph{Implementation.} \input{implem}

\myparagraph{Research questions.}
We investigate  the following research questions:
\begin{description}
  \item[RQ1.] \textbf{Effectiveness: constant-time analysis on real-world cryptographic code.} Is \brelse{} able to perform constant-time analysis on real cryptographic binaries, for both bug finding and bounded-verification?

  \item[RQ2.] \textbf{Genericity.} Is \brelse{}  generic enough to encompass several
        architectures and compilers?

  \item[RQ3.] \textbf{Comparison with standard approaches.} How does
        \brelse{} scale compared to traditional approaches based on self-composition (SC) and RelSE?

  \item[RQ4.] \textbf{Impact of simplifications.} What are the respective  impacts of
        our different simplifications?

  \item[RQ5.] \textbf{Comparison vs. SE.} What is the overhead of \brelse{} compared
        to standard symbolic execution (SE), and can our simplifications be useful for standard SE?

  \item[RQ6.] \textbf{Effectiveness: large scale analysis of scrubbing functions.}
        Is \brelse{} able to verify the secret-erasure property on a large number of binaries?
\end{description}

\myparagraph{Setup.} Experiments were performed on a laptop with an
Intel(R) Core(TM) i5-2520M CPU @ 2.50GHz processor and 32GB of RAM.
Similarly to related work (e.g.~\cite{DBLP:conf/uss/DoychevFKMR13}),
\texttt{esp} is initialized to a concrete value, we start the analysis from the
beginning of the \texttt{main} function, we statically allocate data structures
and the length of keys and buffers is fixed. When not stated otherwise, programs
are compiled for a x86 (32bit) architecture with their default compiler setup.

\myparagraph{Legend.} Throughout this section, \#\(\text{I}\) denotes the number
of static instructions of a program, \#\(\text{I}_{unr}\) is the number of
unrolled instructions explored by the analysis, P is the number of program paths
explored, Time is the execution time give in seconds and \bug{} is the number of
bugs (vulnerable instructions) found.
Status is set to \ccmark{} for secure (exhaustive exploration),
\cxmark{} for insecure, or \hourglass{} for timeout (set to 1 hour).
Additionally, for each program, we report the type of operation
performed and the length of the secret key (Key) and message (Msg)
when applicable (in bytes).

\subsection{Effectiveness of \brelse{} (RQ1)}\label{sec:effectiveness}
We carry out two experiments to assess the effectiveness of our
technique:
\begin{enumerate*}
  \item bounded-verification of secure cryptographic primitives previously
  verified at source- or LLVM-level~\cite{DBLP:conf/esorics/BlazyPT17,DBLP:conf/uss/AlmeidaBBDE16,DBLP:conf/ccs/ZinzindohoueBPB17}
  (\cref{sec:bounded-verification}),
\item automatic replay of known bug
  studies~\cite{DBLP:conf/eurosp/SimonCA18,DBLP:conf/uss/AlmeidaBBDE16,DBLP:conf/sp/AlFardanP13} (\cref{sec:bug-finding}).
\end{enumerate*}
Overall, our study encompasses 338 representative code samples for a
total of 70k machine instructions and 22M unrolled instructions (i.e.,
instructions explored by \brelse{}).

\subsubsection{Bounded-Verification.}\label{sec:bounded-verification}
We analyze a large range of \emph{secure} constant-time cryptographic
primitives (296 samples, 64k instructions), comprising:
\begin{itemize}
\item Several basic constant-time utility functions such as selection
  functions~\cite{DBLP:conf/eurosp/SimonCA18}, sort
  functions~\cite{ImdeasoftwareVerifyingconstanttime}
  and utility functions from
  HACL*~\cite{DBLP:conf/ccs/ZinzindohoueBPB17} %
  and OpenSSL~\cite{OpenSSL}, compiled with \texttt{clang} (versions 3.0, 3.9 and 7.1), and \texttt{gcc} (versions 5.4 and 8.3) and for optimizations levels \texttt{O0} and \texttt{O3};

\item A set of representative constant-time cryptographic primitives
  already studied in the literature on source
  code~\cite{DBLP:conf/esorics/BlazyPT17} or
  LLVM~\cite{DBLP:conf/uss/AlmeidaBBDE16}, including implementations of TEA~\cite{DBLP:conf/fse/WheelerN94},
  Curve25519-donna~\cite{DBLP:conf/pkc/Bernstein06}, \texttt{aes} and
  \texttt{des} encryption functions taken from
  BearSSL~\cite{porninBearSSL}, cryptographic primitives from
  libsodium~\cite{DBLP:conf/latincrypt/BernsteinLS12}, and the constant-time
  padding remove function \texttt{tls-cbc-remove-padding}, extracted from
  OpenSSL~\cite{DBLP:conf/uss/AlmeidaBBDE16};

\item A set of functions from the HACL*
  library~\cite{DBLP:conf/ccs/ZinzindohoueBPB17}.
\end{itemize}

Results are reported in~\cref{tab:bounded-verif}. For each program, \brelse{} is
able to perform an exhaustive exploration without finding any violations of
constant-time in less than 20 minutes. Note that exhaustive exploration is
possible because in cryptographic programs, fixing the input size bounds loops. %
\iftechreport%
\newercontent{Additionally, the scalability of \brelse{} according to the size
  of the input data is evaluated in \cref{app:scale_size} and unbounded loops
  are discussed in \cref{sec:discussion}}. %
\fi
These results show that \brelse{} can
perform bounded-verification of real-world cryptographic implementations at
binary-level in a reasonable time, which was impractical with previous
approaches based on self-composition or standard RelSE
(see~\cref{sec:scalability}). \emph{Moreover, this is the first automatic
  constant-time analysis of these cryptographic libraries at the binary-level.}

\input{./ressources/tables/bv}

\subsubsection{Bug-Finding.}\label{sec:bug-finding}
We take three known bug studies from the
literature~\cite{DBLP:conf/eurosp/SimonCA18,ImdeasoftwareVerifyingconstanttime,DBLP:conf/sp/AlFardanP13}
and replay them automatically at binary-level (42 samples, 6k
instructions), including:
(1) binaries compiled from constant-time sources of a selection
function~\cite{DBLP:conf/eurosp/SimonCA18} and sort
functions~\cite{ImdeasoftwareVerifyingconstanttime},
(2) non-constant-time versions of \texttt{aes} and \texttt{des} from
BearSSL~\cite{porninBearSSL},
(3) the non-constant-time version of OpenSSL's
\texttt{tls-cbc-remove-padding}\footnote{\url{https://github.com/openssl/openssl/blob/OpenSSL_1_0_1/ssl/d1_enc.c}\label{fnote:lucky13}}
responsible for the famous Lucky13 attack~\cite{DBLP:conf/sp/AlFardanP13}.

Results are reported in~\cref{tab:bug-finding} with \emph{fault-packing
  disabled} to report vulnerabilities at the instruction level. All
bugs have been found within the timeout.
Interestingly, we found 3 \emph{unexpected binary-level vulnerabilities
  (from secure source codes) that slipped through prior analysis}:
\begin{itemize}
  \item function \texttt{ct\_select\_v1}~\cite{DBLP:conf/eurosp/SimonCA18} was
        deemed secured through binary-level manual inspection, still we confirm
        that any version of \texttt{clang} with \texttt{-O3} introduces a
        secret-dependent conditional jump which violates constant-time;
  \item functions \texttt{ct\_sort} and \texttt{ct\_sort\_mult}, verified by
        ct-verif~\cite{DBLP:conf/uss/AlmeidaBBDE16} (LLVM
        bitcode compiled with \texttt{clang}), are vulnerable when compiled
        with \texttt{gcc -O0} or \texttt{clang -O3 -m32 -march=i386} (details in
        \cref{sec:compilers}).
\end{itemize}

\input{./ressources/tables/bf}

\myparagraph{Conclusion (RQ1).} We perform an extensive analysis over 338
samples of representative cryptographic primitive studied in the
literature~\cite{DBLP:conf/esorics/BlazyPT17,DBLP:conf/ccs/ZinzindohoueBPB17,DBLP:conf/uss/AlmeidaBBDE16}
Overall, it demonstrates that \brelse{} does scale to realistic applications for
both bug-finding and bounded-verification. As a side-result, we also proved
CT-secure 296 binaries of interest.

\subsection{Preservation of Constant-Time by Compilers (RQ2).}\label{sec:compilers} %
\newcontent{In this section, we present an easily extensible framework, based on \brelse{},
to check constant-time for small programs under multiple compiler
setups\footnote{\url{https://github.com/binsec/rel_bench/tree/main/properties_vs_compilers/ct}}.}
Using this framework, we replay a prior \emph{manual}
study~\cite{DBLP:conf/eurosp/SimonCA18}, which analyzed whether \texttt{clang}
optimizations break the constant-time property, for 5 different versions of a
selection function (\texttt{ct-select}). %
We reproduce their analysis in an \emph{automatic} manner and extend it
significantly, adding:
29 new functions, 3 newer version of \texttt{clang} (7.1.0, 9.0.1 and
11.0.1), the \texttt{gcc} compiler, and 2 new architectures (i.e., \texttt{i686}
and \texttt{arm}, while only \texttt{i386} was considered in the initial
study)---for a total of \newercontent{4148} configurations (192 in the initial
study).

\newcontent{Additionally, we investigate the impact of individual optimizations
  on the preservation of constant-time. For \texttt{clang}, we target the
  \texttt{-x86-cmov-converter} which converts x86 \dbainline{cmov} instructions
  into branches when profitable and which is known to play a role in the
  preservation of constant-time~\cite{cmov-conversion}. In particular, we
  evaluate the impact of selectively disabling this optimization, by passing the
  flags \texttt{-O3 -mllvm -x86-cmov-converter=0} to \texttt{clang}, which we
  denote \texttt{O3\textsuperscript{-}}.
  For \texttt{gcc}, we target the if-conversion (i.e., \texttt{-fif-conversion
    -fif-conversion2 -ftree-loop-if-convert}), which transforms conditional
  jumps into branchless equivalent. In particular, we evaluate the impact of
  selectively \emph{enabling} this optimization, by passing the flags
  \texttt{-O0 -fif-conversion -fif-conversion2 -ftree-loop-if-convert} to
  \texttt{gcc}, (denoted \texttt{O0\textsuperscript{+}}); and the impact of
  selectively \emph{disabling} this optimization using \texttt{-O3
    -fno-if-conversion -fno-if-conversion2 -fno-tree-loop-if-convert} (denoted
  \texttt{O3\textsuperscript{-}}).
  \newercontent{Bear in mind that the \texttt{i386} architecture does not feature
    \dbainline{cmov} instructions but \texttt{i686} does.}
  Results are presented in \cref{tab:verif_compilers}. Results for
  \texttt{O3\textsuperscript{-}} are not applicable to \texttt{clang-3.0} and
  \texttt{clang-3.9} (denoted - in the table) as these versions do not recognize
  the \texttt{-x86-cmov-converter} argument.}

\input{./ressources/tables/ct-vs-compilers}

We confirm the main conclusion of Simon \emph{et
  al.}~\cite{DBLP:conf/eurosp/SimonCA18} that \texttt{clang} is more likely to
optimize away constant-time protections as the optimization level increases.
However, \emph{contrary to their work}, our experiments show that newer
versions of \texttt{clang} are not necessarily more likely than older
ones to break constant-time (e.g. \texttt{ct\_sort} is compiled to a
non-constant-time code with \texttt{clang-3.9} but not with
\texttt{clang-7.1}).

Surprisingly, in contrast with \texttt{clang}, \texttt{gcc} optimizations tend
to remove branches and thus, are less likely to introduce vulnerabilities in
constant-time code. Especially, \texttt{gcc} for ARM produces \emph{secure
  binaries from the insecure source codes}. Indeed, the compiler takes advantage
of the many ARM conditional instructions to remove conditional jumps in
\texttt{sort\_naive} and \texttt{naive\_select}. %
\newercontent{This also applies to the \texttt{i686} architecture but only for
  \texttt{naive\_select}.} %
\newcontent{We conclude that the if-conversion passes of \texttt{gcc} play a
  role here, as disabling them (\texttt{O3\textsuperscript{-}}) produces
  insecure binaries.} %
\newercontent{However, the fact that \texttt{O0\textsuperscript{+}} is still
  insecure shows that if-conversion passes must be combined with other
  optimizations (at least \texttt{O1}) to effectively remove conditional jumps.}

\newercontent{Finally, we found that constant-time sort functions, taken from
  the benchmark of the \texttt{ct-verif}~\cite{DBLP:conf/uss/AlmeidaBBDE16}
  tool, can be compiled to insecure binaries for two different reasons %
  \iftechreport%
  (both detailed in \cref{app:ct-sort}):
  \fi
  \begin{itemize}
  \item For the \texttt{i386} architecture and old compilers, conditional
        \dbainline{select} LLVM instructions are compiled to conditional jumps
        because target architectures do not feature \dbainline{cmov}
        instructions. These violations are introduced in \emph{backend passes}
        of \texttt{clang}, making them of reach of LLVM verification tools like
        \texttt{ct-verif}\footnote{We did confirm that \texttt{ct-verif} with
        the setting \texttt{--clang-options="-O3 -m32 -march=i386"} does not
        report the vulnerability.};
    \item More interestingly, we found that for more recent architectures
          featuring \dbainline{cmov} (i.e., \texttt{i686}), the use of
          \dbainline{cmov} might \emph{introduce secret-dependent memory
          accesses}. Indeed, the compiler introduces a secret-dependent pointer
          selection, done with \dbainline{cmov}, which results in a memory-based
          leak when the pointer is dereferenced.
\end{itemize}
We also remark that disabling the \texttt{-x86-cmov-converter} does not change
anything in our settings. } %

\myparagraph{Conclusion (RQ2).} %
This study shows that \brelse{} is generic in the sense that it can be applied
with different versions and options of \texttt{clang} and \texttt{gcc}, over x86
and ARM. We also get the following interesting results:
\begin{itemize}
  \item \newcontent{We found that, contrary to \texttt{clang}, \texttt{gcc}
        optimizations tend to help enforcing constant-time---\texttt{gcc -O3}
        preserves constant-time in all our examples. \texttt{gcc} even
        sometimes produces secure binaries from insecure sources thanks to the
        if-conversion passes};
  \item We found that backend passes of \texttt{clang} can introduce
        vulnerabilities in codes that are secure at the LLVM level;
  \item \newercontent{We found that \texttt{clang} use of \dbainline{cmov}
        instructions might introduce secret-dependent memory accesses};
  \item \newcontent{Finally, this study shows that the preservation of
        constant-time by compilers depends on multiple factors and cannot simply
        rely on enabling/disabling optimizations. Instead, compiler-based
        hardening~\cite{DBLP:conf/ccs/BorrelloDQG21,DBLP:conf/uss/RaneLT15} or
        property preservation~\cite{DBLP:conf/eurosp/SimonCA18} seem promising
        directions, in which \brelse{} could be used for validation.}
\end{itemize}

\subsection{Comparison against Standard Techniques (RQ3,RQ4,RQ5)}\label{sec:scalability}
We compare \brelse{} with standard techniques based on
self-composition and relational symbolic execution (\textit{RelSE})
(\cref{sec:comparison_std}), then we analyze the preformance of our different
simplifications (\cref{sec:perfs-simpl}), and finally we investigate the overhead of
\brelse{} compared to standard SE, and whether our simplifications are
useful for SE (\cref{sec:comparison_se}).

Experiments are performed on the programs introduced in
\cref{sec:effectiveness} for bug-finding and
bounded-verification %
(338 samples, 70k instructions).
We report the following metrics: total number of unrolled instruction
\#\(\text{I}_{unr}\), number of instruction explored per seconds
(\#\(\text{I}_{unr}\)/s), total number of queries sent to the solver (\#Q),
number of exploration (resp.\ insecurity) queries (\(\text{\#Q}_{\text{e}}\)),
(resp.\ \(\text{\#Q}_{\text{i}}\)), total execution time (T), timeouts
(\hourglass), programs proven secure (\ccmark), programs proven insecure
(\cxmark), unknown status (\csmark). Timeout is set to 3600 seconds.

\subsubsection{Comparison vs. Standard Approaches (RQ3).}\label{sec:comparison_std}
We evaluate \brelse{} against \textit{SC} and
\textit{RelSE}. %
Since no implementation of these methods fits our particular use-cases,
we implement them directly in \binsec{}. \textit{RelSE} is obtained by
disabling \brelse{} optimizations (\cref{sec:optims}), while
\textit{SC} is implemented on top of \textit{RelSE} by duplicating low
inputs instead of sharing them and adding the adequate preconditions.
Results are given in \cref{tab:scale_total_summary}.

\input{./ressources/tables/scale_total_summary}

While \textit{RelSE} performs slightly better than \textit{SC}
(\newcontent{\(1.6 \times\) speedup in terms of \#\(\text{I}_{unr}/s\)}) thanks
to a noticeable reduction of the number of queries (approximately 50\%), both
techniques are not efficient enough on binary code:
\textit{RelSE} times out in 13 cases and achieves an analysis speed of only 6.2
instructions per second while \textit{SC} is worse.
\brelse{} completely outperforms both previous approaches:
\begin{itemize}
  \item The optimizations implemented in \brelse{} drastically reduce the
        number of queries sent to the solver (\(57\times\) less insecurity
        queries than \textit{RelSE});
  \item \brelse{} reports no timeout, \newcontent{is \(1000\times\) faster than
        \textit{RelSE} and \(1600\times\) faster than \textit{SC} in terms of
        \#\(\text{I}_{unr}/s\)};
  \item \brelse{} can perform bounded-verification of large programs (e.g.
        \texttt{donna}, \texttt{des-ct}, \texttt{chacha20}, etc.) that were out
        of reach of prior approaches.
\end{itemize}

\subsubsection{Performance of Simplifications
  (RQ4).}\label{sec:perfs-simpl} %
We evaluate the performance of our individual optimizations: on-the-fly
read-over-write (\textit{FlyRow}), untainting (\textit{Unt}) and fault-packing
(\textit{FP}). Results are reported in \cref{tab:scale_optims}:

\input{./ressources/tables/scale_optims}

\begin{itemize}
  \item \textit{FlyRow} is the major source of improvement in \brelse{},
        drastically reducing the number of queries sent to the solver and
        allowing a \(718\times\) speedup compared to \textit{RelSE} in terms of
        \#\(\text{I}_{unr}/s\);

  \item Untainting and fault-packing do have a positive impact on
        \textit{RelSE}---untainting alone reduces the number of queries by
        almost 50\%, the two optimizations together yield a \(2\times\) speedup;

  \item Yet, their impact is more modest once \textit{FlyRow} is activated:
        untainting leads to a very slight slowdown, while fault-packing achieves
        a \(1.4\times\) speedup.
\end{itemize}

\noindent Still, \textit{FP} can be interesting on some particular programs,
when the precision of the bug report is not the priority. Consider for instance
the non-constant-time version of \texttt{aes} in BearSSL (i.e.,
\texttt{aes-big}): \brelse{} without \textit{FP} reports 32 vulnerable
instructions in 1580 seconds, while \brelse{} with \textit{FP} reports 2
vulnerable \emph{basic blocks} (covering the 32 vulnerable instructions) in only
146 seconds (almost \(11 \times\) faster).

\subsubsection{Comparison vs. Standard SE (RQ5).}\label{sec:comparison_se} %
We investigate the overhead of \brelse{} compared to standard symbolic
execution (\textit{SE}); evaluate whether on-the-fly read-over-write
(\textit{FlyRow}) can improve performance of \textit{SE}; and also compare
\textit{FlyRow} to a recent implementation of
read-over-write~\cite{DBLP:conf/lpar/FarinierDBL18} (\textit{PostRow}),
implemented posterior to symbolic-execution as a formula pre-processing. %
Standard symbolic-execution is directly implemented in the \textsc{Rel} module
and models a single execution of the program with exploration queries but
\emph{without insecurity queries}. %

\input{./ressources/tables/sse_relse_overhead2}

\begin{itemize}
  \item \brelse{}, compared to our best setting for symbolic execution
        (\textit{SE+FlyRow}), only has an overhead of \(2\times\) in terms of
        \#\(\text{I}_{unr}/s\). Hence constant-time comes with an acceptable
        overhead on top of standard symbolic execution. This is consistent with
        the fact that our simplifications discard most insecurity queries,
        letting only the exploration queries which are also part of
        symbolic-execution;

  \item For RelSE, \textit{FlyRow} completely outperforms \textit{PostRow}. %
        First, \textit{PostRow} is not designed for relational verification and
        duplicates the memory. Second, \textit{PostRow} simplifications are not
        propagated along the execution and must be recomputed for every query,
        producing a significant simplification overhead. On the contrary,
        \textit{FlyRow} models a single memory containing relational values and
        propagates along the symbolic execution.

  \item \textit{FlyRow} also improves the performance of standard SE by a factor
        \(643\) in our experiments, performing much better than \textit{PostRow}
        (\(430\times\) faster).
\end{itemize}

\myparagraph{Conclusion (RQ3, RQ4, RQ5).}
\brelse{} performs significantly better than previous approaches to
relational symbolic execution (\(1000\times\) speedup vs.~\textit{RelSE}). The
main source of improvement is the on-the-fly read-over-write simplification
(\textit{FlyRow}), which yields a \(718\times\) speedup vs.~\textit{RelSE} and
sends \(57 \times\) less insecurity queries to the solver.
Note that, in our context, \textit{FlyRow} outperforms state-of-the-art
binary-level simplifications, as they are not designed to efficiently cope with
relational properties and introduce a significant simplification-overhead at
every query.
Fault-packing and untainting, while effective over \textit{RelSE}, have a much
slighter impact once \textit{FlyRow} is activated; fault-packing can still be
useful on insecure programs. %
Finally, in our experiments, \textit{FlyRow} significantly improves performance
of standard symbolic-execution (\(643 \times\) speedup).

\subsection{Preservation of Secret-Erasure by Compilers (RQ6)}\label{sec:expes-secret-erasure}
Secret-erasure is usually enforced using \emph{scrubbing functions}---functions
that overwrite a given part of the memory with dummy
values. %
In this section we present a framework to automatically check the preservation
of secret-erasure for multiple scrubbing functions and compilers. This framework
is open
source\footnote{\url{https://github.com/binsec/rel_bench/tree/main/properties_vs_compilers/secret-erasure}}
and \textit{can be easily extended} with \textit{new compilers} and \textit{new
  scrubbing functions}. Using \brelse{}, we analyze 17 scrubbing functions; with
multiple versions of \texttt{clang} (3.0, 3.9, 7.1.0, 9.0.1 and 11.0.1) and
\texttt{gcc} (5.4.0, 6.2.0, 7.2.0, 8.3.0 and 10.2.0); and multiple optimization
levels (\texttt{-O0}, \texttt{-O1}, \texttt{-O2} and \texttt{-O3}). %
\newcontent{We also investigate the impact of disabling individual optimizations
  (those related to the dead-store-elimination pass) on the preservation of
  secret-erasure (cf.\ \cref{sec:compiler-optims})}. %
This accounts for a total of \newcontent{\emph{1156 binaries}} and extends a
prior \emph{manual} study on scrubbing
mechanisms~\cite{DBLP:conf/uss/YangJOLL17}.

In this section, clang-all-versions (resp.
gcc-all-versions) refer to all the aforementioned clang (resp. gcc) versions;
and in tables \ccmark{} indicates that a program is secure and \cxmark{} that it
is insecure w.r.t secret-erasure.

\subsubsection{Naive implementations}
First, we consider naive (insecure) implementations of scrubbing functions: %
\begin{itemize}
  \item \lstinline{loop}: naive scrubbing function that uses a simple \lstinline{for}
        loop to set the memory to 0,
  \item \lstinline{memset}: uses the \lstinline{memset} function from the Standard C
        Library,
  \item \lstinline{bzero}: function defined in \texttt{glibc} to set memory to 0.
\end{itemize}

\begin{minipage}[t]{.5\linewidth}
\vspace{0pt}
\input{./ressources/tables/naive_implementations}

\end{minipage}\hfill
\begin{minipage}[t]{.45\linewidth}
\vspace{0pt}
\noindent\textbf{Results} (cf. \cref{tab:naive_implementations}). As expected,
without appropriate countermeasures, these naive implementation of scrubbing
functions are all optimized away by all versions of \texttt{clang} and
\texttt{gcc} at optimization level \texttt{-O2} and \texttt{-O3}. Additionally,
as highlighted in \cref{tab:naive_implementations}, \lstinline{bzero} is also
optimized away at optimization level \texttt{-O1} with \texttt{gcc-7.2.0} and
older versions\footnotemark.
\end{minipage}
\footnotetext{This is because the function calls to \lstinline{scrub}
  and \lstinline{bzero} are inlined in \texttt{gcc-7.2.0} and older versions,
  making the optimization possible whereas the call to \lstinline{scrub} is not
  inlined in \texttt{gcc-8.3.0} and older versions.}

\subsubsection{Volatile function pointer}
The \lstinline{volatile} type qualifier indicates that the value of an object
may change at any time, preventing the compiler from optimizing memory accesses
to volatile objects. This mechanism can be exploited for secure secret-erasure
by using a volatile function pointer for the scrubbing function (e.g.\
eventually redirecting to \lstinline{memset}). Because the function may change,
the compiler cannot optimize it away. \Cref{lst:volatile_func} illustrates the
implementation of this mechanism in
OpenSSL~\cite{OpenSSL_cleanse}. %

\begin{center}
\begin{minipage}[t]{.55\linewidth}
\vspace{0pt}
\centering
\input{./ressources/listings/volatile_function}
\end{minipage}%
\hfill
\begin{minipage}[t]{.4\linewidth}
\vspace{0pt}
\centering
\input{./ressources/tables/volatile_function}
\end{minipage}
\end{center}

\noindent\textbf{Results} (cf. \cref{tab:volatile_function}). \brelse{} reports that, for all versions of
\texttt{gcc}, the secret-erasure property is not preserved at optimization
levels \texttt{-O2} and \texttt{-O3}.
Indeed, the caller-saved register \dbainline{edx} is pushed on the stack before the
call to the volatile function. However, it contains secret data which are
spilled on the stack and not cleared afterwards.
\textit{This shows that our tool can find violations of secret erasure from
  register spilling.}
We conclude that while the volatile function pointer mechanism is effective for
preventing the scrubbing function to be optimized away, it may also
\emph{introduce unnecessary register spilling that might break secret-erasure}.

\subsubsection{Volatile data pointer} %
The volatile type qualifier can also be used for secure secret-erasure by
marking the data to scrub as volatile before erasing it. We analyze several
implementations based on this mechanism:
\begin{itemize}
  \item \texttt{ptr\_to\_volatile\_loop} casts the pointer \lstinline{buf} to a
        pointer-to-volatile \lstinline{vbuf} (cf.\ \cref{lst:volatile_ptr},
        line~\ref{line:ptr-to-volatile}) before scrubbing data from \lstinline{vbuf}
        using a simple \lstinline{for} or \lstinline{while} loop. This is a commonly
        used technique for scrubbing memory, used for instance in
        Libgcrypt~\cite{Libgcrypt_wipememory}, wolfSSL~\cite{wolfSSL_ForceZero},
        or sudo~\cite{sudo_explicit_bzero};
  \item \texttt{ptr\_to\_volatile\_memset} is similar to
        \texttt{ptr\_to\_volatile\_loop} but scrubs data from memory using
        \lstinline{memset}. Note that this implementation is insecure as the
        \lstinline{volatile} type qualifier is discarded by the function
        call---\lstinline{volatile char *} is not compatible with \lstinline{void *};
  \item \texttt{volatile\_ptr\_loop} (resp. \texttt{volatile\_ptr\_memset})
        casts the pointer \lstinline{buf} to a volatile pointer
        \lstinline{vbuf}---but pointing to non volatile data (cf.\
        \cref{lst:volatile_ptr}, line~\ref{line:volatile-ptr}) before scrubbing data
        from \lstinline{vbuf} using a simple \lstinline{for} or \lstinline{while} loop
        (resp. \lstinline{memset})\footnote{Although we did not find this
        implementation in real-world cryptographic code, we were curious about
        how the compiler would handle this case.};
  \item \texttt{vol\_ptr\_to\_vol\_loop} casts the pointer \lstinline{buf} to a
        volatile pointer-to-volatile \lstinline{vbuf} (cf.\
        \cref{lst:volatile_ptr}, line~\ref{line:volatile-ptr-to-volatile}) before
        scrubbing data from \lstinline{vbuf} using a simple \lstinline{for} or
        \lstinline{while} loop. It is the fallback scrubbing mechanism used in
        \texttt{libsodium}~\cite{libsodium_memzero} and in
        HACL*~\cite{hacl_memzero} cryptographic libraries;
  \item \texttt{vol\_ptr\_to\_vol\_memset} is similar to
        \texttt{vol\_ptr\_to\_vol\_loop} but uses \lstinline{memset} instead of a
        \lstinline{loop}.
\end{itemize}

\input{./ressources/listings/volatile_ptr}
\input{./ressources/tables/volatile_ptr}

\noindent\textbf{Results} (cf. \cref{tab:volatile_ptr}). First, our
experiments show that using \emph{volatile pointers to non-volatile data does
  not reliably prevent the compiler from optimizing away the scrubbing
  function}. Indeed, \texttt{gcc} optimizes away the scrubbing function at
optimization level \texttt{-O2} and \texttt{-O3} in both \texttt{volatile\_ptr}
implementations. Second, using a pointer to volatile works in the \lstinline{loop}
version (i.e., \texttt{ptr\_to\_volatile\_loop} and
\mbox{\texttt{vol\_ptr\_to\_vol\_loop}}) but not in the \lstinline{memset} versions
(i.e., \mbox{\texttt{ptr\_to\_volatile\_memset}} and
\texttt{vol\_ptr\_to\_vol\_memset}) as the function call to \lstinline{memset}
discards the volatile qualifier.

\subsubsection{Memory barriers}\label{section:memory_barriers} %
Memory barriers are inline assembly statements which indicate the compiler that
the memory could be read or written, forcing the compiler to preserve preceding
store operations. We study four different implementations of memory barriers:
three implementations from safeclib~\cite{Safeclib_memory_barriers}, plus the
approach recommended in a prior study on scrubbing
mechanisms~\cite{DBLP:conf/uss/YangJOLL17}.
\begin{itemize}
  \item \texttt{memory\_barrier\_simple} (cf.\ \cref{lst:memory_barriers},
        line~\ref{line:memory_barriers:simple}) is the implementation used in
        \texttt{explicit\_bzero} and the fallback implementation used in
        safeclib. As pointed by
        \citeauthor{DBLP:conf/uss/YangJOLL17}~\cite{DBLP:conf/uss/YangJOLL17},
        this barrier works with \texttt{gcc}~\cite{gcc_extended_asm} but might
        not work with \texttt{clang}, which might optimize away a call to
        \lstinline{memset} or a loop before this
        barrier~\cite{clang_bug15495}---although we could not reproduce the
        behavior in our experiments;
  \item \texttt{memory\_barrier\_mfence} (cf.\ \cref{lst:memory_barriers},
        line~\ref{line:memory_barriers:fence}) is similar to
        \texttt{memory\_barrier\_simple} with an additional \dbainline{mfence}
        instruction for serializing memory. It is used in safeclib when
        \dbainline{mfence} instruction is available;
  \item \texttt{memory\_barrier\_lock} (cf.\ \cref{lst:memory_barriers},
        line~\ref{line:memory_barriers:lock}) is similar to
        \texttt{memory\_barrier\_mfence} but uses a \dbainline{lock} prefix for
        serializing memory. It is used in safeclib on \texttt{i386}
        architectures;
  \item \texttt{memory\_barrier\_ptr} (cf.\ \cref{lst:memory_barriers},
        line~\ref{line:memory_barriers:ptr}) is a more resilient approach than
        \texttt{memory\_barrier\_simple}, recommended in the study of
        \citeauthor{DBLP:conf/uss/YangJOLL17}~\cite{DBLP:conf/uss/YangJOLL17},
        and used for instance in libsodium
        \texttt{memzero}~\cite{libsodium_memzero}. It makes the pointer
        \lstinline{buf} visible to the assembly code, preventing prior store
        operation to this pointer from being optimized away.
\end{itemize}

\input{./ressources/listings/memory_barriers}
\noindent\textbf{Results.} For all the implementation of memory barriers that we
tested, we did not find any vulnerability---even with the version deemed
insecure in prior
study~\cite{DBLP:conf/uss/YangJOLL17}\footnote{\newcontent{As explained in a
    bug report~\cite{clang_bug15495}, \texttt{memory\_barrier\_simple} is not
    reliable because \texttt{clang} might consider that the inlined assembly
    code does not access the buffer (e.g.\ by fitting all of the buffer in
    registers). The fact that we were not able to reproduce this bug in our
    setup is due to differences in programs (in our program the address of the
    buffer escapes because of function calls whereas it is not the case in the
    bug report); it does not mean that this barrier is secure (it is not).}}.

\subsubsection{Weak symbols} %
Weak symbols are specially annotated symbols (with
\lstinline[]{__attribute__((weak))}) whose definition may
change at link time. An illustration of a weak function symbol is given in
\cref{lst:weak_symbols}. The compiler cannot optimize a store operation
preceding the call to \lstinline{_sodium_dummy_symbol} because its definition
may change and could access the content of the buffer. This mechanism, is used
in libsodium \lstinline{memzero}~\cite{libsodium_memzero} when weak symbols are
available.

\input{./ressources/listings/weak_symbols}

\noindent\textbf{Results.} \brelse{} did not find any vulnerability with
weak-symbols.

\subsubsection{Off-the-shelf implementations}
Finally, we consider two secure implementations of scrubbing functions proposed
in external libraries, namely \texttt{explicit\_bzero} and \texttt{memset\_s}.
\texttt{explicit\_bzero} is a function defined in \texttt{glibc} to set memory
to 0, with additional protections to not be optimized away by the compiler.
Similarly, \texttt{memset\_s} is a function defined in the optional Annex K
(bound-checking interfaces) of the C11 standard, which sets a memory region to a
given value and should not be optimized away. We take the implementation of
\texttt{safeclib}~\cite{Safeclib_memset_s}, compiled with its default Makefile
for a \texttt{i386} architecture.
Both implementations both rely on a memory barrier (see
\cref{section:memory_barriers}) to prevent the compiler from optimizing
scrubbing operations.

\noindent\textbf{Results.} \brelse{} did not find any vulnerability with these functions.

\subsubsection{Impact of disabling individual optimizations}\label{sec:compiler-optims}
\newcontent{In order to understand what causes compilers to introduce violations
  of secret-erasure, we selectively disable the \texttt{-dse} (i.e., dead store
  elimination) option in \texttt{clang} and the \texttt{-dse} and
  \texttt{-tree-dse} (i.e., dead store elimination on tree) in \texttt{gcc}.}

\myparagraph{Results.} \newcontent{For \texttt{clang-3.9}, \texttt{clang-7.1.0},
  \texttt{clang-9.0.1} and \texttt{clang-11.0.1}\footnote{\texttt{clang-3.0} is
    omitted in this study because we were not able to run the LLVM
    optimizer (\texttt{opt}) for \texttt{clang-3.0} in order to disable the
    \texttt{-dse} optimization.}, disabling the \texttt{-dse} transform pass
  makes all our samples secure. This points towards the hypothesis that the
  \texttt{-dse} transform pass is often responsible for breaking secret-erasure
  and that, in some cases, disabling it might be sufficient to preserve
  secret-erasure\footnote{However, we strongly suspect that this conclusion does
    not generalize to all programs, for instance to programs that violate
    secret-erasure because of register spilling.}.}

\newcontent{The results for \texttt{gcc} are given in table
  \cref{tab:secret-erasure_no-dse_gcc}. Firstly, we observe that both
  \texttt{-dse} and \texttt{-tree-dse} play a role in the preservation of
  secret-erasure. Indeed, for \texttt{bzero}, disabling \texttt{dse} is
  sufficient for obtaining a secure binary, while for
  \texttt{volatile\_ptr\_memset} and \texttt{vol\_ptr\_to\_vol\_memset},
  \texttt{-tree-dse} must be disabled. On the contrary, for \texttt{memset} and
  \texttt{ptr\_to\_volatile\_memset}, it is necessary to disable both
  optimizations. %
  Secondly, we observe that there are other factors that affect the
  preservation of secret-erasure. Indeed, the \texttt{volatile\_func} program is
  still insecure because of register spilling. Additionally, \texttt{loop} and
  \texttt{volatile\_ptr\_loop} are also insecure because the loop is still
  optimized away.}

\input{./ressources/tables/secret-erasure_no-dse_gcc}

%% file: implem.tex
We implemented our relational symbolic execution, \brelse{}, on top of
the binary-level analyzer \binsec{}~\cite{DBLP:conf/wcre/DavidBTMFPM16}.
\brelse{} takes as input an x86 or ARM executable, a specification of
high inputs and an initial memory configuration (possibly fully
symbolic). It performs bounded exploration of the
program under analysis (up to a user-given depth), and reports
identified violations together with counterexamples (i.e., initial
configurations leading to the vulnerabilities).
In case no violation is reported, if the initial configuration is
fully symbolic and the program has been explored exhaustively then the
program is \emph{proven} secure.
 \brelse{} is composed of a \emph{relational symbolic exploration}
module and an \emph{insecurity analysis} module. The symbolic
exploration module chooses the path to explore, updates the symbolic
configuration, builds the path predicate and ensure that it is
satisfiable. The insecurity analysis module builds insecurity queries
and ensures that they are not satisfiable. It can be configured
according to the leakage model.
 We explore the program in a depth-first search manner and we rely on
the Boolector SMT-solver~\cite{niemetzBoolectorSystemDescription2014},
currently the best on theory
\abv{}~\cite{SMTCOMP,DBLP:conf/lpar/FarinierDBL18}.

%% file: ressources/tables/bv.tex
\begin{table}[htbp]
  \centering
  \setlength{\tabcolsep}{4pt}
  \begin{tabularx}{\textwidth}{lXrrrrrlrr}
    \toprule
    \multicolumn{2}{l}{Description and number of binaries\(^{\dagger}\)}
    & \multicolumn{1}{r}{\(\approx\text{\#I}\)}
    & \multicolumn{1}{r}{\#\(\text{I}_{unr}\)}
    & \multicolumn{1}{r}{P}
    & \multicolumn{1}{r}{Time}
    & \multicolumn{1}{r}{Status}
    & \multicolumn{1}{l}{Type}
    & \multicolumn{1}{r}{Key}
    & \multicolumn{1}{r}{Msg}\\
    \midrule
    \multirow{4}{*}{Utility}
    & ct-select (\(\times 29\))  & 1015 &  1507 & 29 &0.2 & 29 \(\times\) \ccmark{} &
    \multirow{4}{*}{Utility functions} %
    & - & \newercontent{9} \\
    & ct-sort   (\(\times 12\))  & 2400 &  1782 &  12 & 0.2 & 12 \(\times\) \ccmark{} & & - & \newercontent{12}\\
    & Hacl*     (\(\times 110\)) & 3850 & 90953 & 110 & 7.6 & 110 \(\times\) \ccmark{} & & - & \newercontent{2-200}\\
    & OpenSSL   (\(\times 130\)) & 4550 &  5113 & 130 & 0.9 & 130 \(\times\) \ccmark{} & & - & \newercontent{4-12}\\
    \midrule
    \multirow{2}{*}{Tea}
    & decrypt \texttt{-O0} & 290 & 953 & 1 & 0.1 & \ccmark{} &
    \multirow{2}{*}{Block cipher} &
    \multirow{2}{*}{16} &
    \multirow{2}{*}{8} \\
    & decrypt \texttt{-O3} & 250 & 804 & 1 & 0.1 & \ccmark{} &\\
    \midrule
    \multirow{2}{*}{Donna}
        & \texttt{-O0} & 7083 & 10.2M & 1 &  1008.5 & \ccmark{} &
    \multirow{2}{*}{Elliptic curve} &
    \multirow{2}{*}{32} &
    \multirow{2}{*}{-} \\
        & \texttt{-O3} & 4643 &  2.7M & 1 &  347.1 & \ccmark{} \\
    \midrule
    \multirow{4}{*}{Libsodium}
    & salsa20  &  1627 & 38.0k & 1 &  3.5 & \ccmark{} & Stream cipher & 32 & 256 \\
    & chacha20 &  2717 & 12.3k & 1 &  1.5 & \ccmark{} & Stream cipher & 32 & 256 \\
    & sha256   &  4879 & 48.4k & 1 &  4.5 & \ccmark{} & Secure hash & - & 256 \\
    & sha512   & 16312 & 92.0k & 1 &  8.1 & \ccmark{} & Secure hash & - & 256 \\
    \midrule
    \multirow{4}{*}{Hacl*}
    & chacha20   & 1221 &   6.7k & 1 &    4.3 & \ccmark{} & Stream cipher & 32 & 256 \\
    & sha256     & 1279 &  21.0k & 1 &    2.7 & \ccmark{} & Secure hash & - & 256 \\
    & sha512     & 2013 &  47.5k & 1 &    5.2 & \ccmark{} & Secure hash & - & 256 \\
    & curve25519 & 8522 &   9.4M & 1 &  927.8 & \ccmark{} & Elliptic curve & 32 & - \\
    \midrule
    \multirow{2}{*}{BearSSL}
    & aes-ct-cbcenc\(^{\ddagger}\)   & 357 &  3.5k & 1 &   0.5 & \ccmark{} & Block cipher & 240 & 32 \\
    & des-ct-cbcenc\(^{\ddagger}\)   & 682 & 19.9k & 1 & 12.1 & \ccmark{} & Block cipher & 384 & 16 \\
    \midrule
    \multirow{1}{*}{OpenSSL}
    & tls-remove-padding-patch & 424 & 35.7k & 520 & 438.0 & \ccmark{} & Remove padding & - & 63 \\
    \midrule
    \textbf{Total} & 296 binaries & 64114 & 22.8M & 815 & 2772.7 & 296 \(\times\) \ccmark{} & - & - & - \\
    \bottomrule
  \end{tabularx}
  \caption{Bounded verification for constant-time cryptographic implementations.
    \(^{\dagger}\) A line in which the number of binaries is not indicated
    corresponds to 1 binary. \(^{\ddagger}\) \texttt{aes} set to 2 rounds and \texttt{des} set to 2 iterations.}\label{tab:bounded-verif}
\end{table}

%% file: ressources/tables/bf.tex
\begin{table}[!htbp]
  \centering
  \setlength{\tabcolsep}{3.8pt}
  \begin{tabularx}{\textwidth}{lXrrrrcrrlll}
    \toprule
    \multicolumn{2}{l}{Description and nb.\ of binaries\(^{\dagger}\)}
    & \multicolumn{1}{r}{\(\approx\text{\#I}\)}
    & \multicolumn{1}{r}{\#\(\text{I}_{unr}\)}
    & \multicolumn{1}{r}{P}
    & \multicolumn{1}{r}{Time}
    & CT
    & Status
    & \bug{}
    & \multicolumn{1}{l}{Type}
    & \multicolumn{1}{l}{Key}
    & \multicolumn{1}{l}{Msg}\\
    \midrule
    \multirow{2}{*}{Utility}
    & ct-select (\(\times 21\)) &  735 &  767 & 42 & 0.4 & Y & 21 \(\times\) \cxmark{} (1 new) & 21 & \multirow{2}{*}{Utility func.} & - & - \\
    & ct-sort  (\(\times 18\)) & 3600 & 7513 & 138 & 13.6 & Y & 18 \(\times\) \cxmark{} (2 new) & 44 &  & - & - \\
    \midrule
    \multirow{2}{*}{BearSSL}
    & aes-big-cbcenc\(^{\ddagger}\) & 375 &   876 & 1 & 1651.8 & N & \cxmark{} & 32 & Block cipher & 240 & 32 \\
    & des-tab-cbcenc\(^{\ddagger}\) & 365 &  5187 & 1 & 4.4 & N & \cxmark{} &  8 & Block cipher & 384 & 16 \\
    \midrule
    OpenSSL & tls-rem-pad-lucky13                               
       & 950 &  7866 & 375 & 700.3 & N & \cxmark & 6 & Remove pad. & - & 63 \\
    \midrule
    \textbf{Total} & 42 binaries & 6025 & 22209 & 557 & 2370.5 & - & 42 \(\times\) \cxmark & 111 & - & - & -  \\
    \bottomrule
  \end{tabularx}
  \caption{Bug-finding for constant-time in cryptographic implementations.
    \(^{\dagger}\) A line in which the number of binaries is not indicated
    corresponds to 1 binary. \(^{\ddagger}\) \texttt{aes} set to 2 rounds and \texttt{des} set to 2 iterations.}\label{tab:bug-finding}
\end{table}

%% file: ressources/tables/ct-vs-compilers.tex
\begin{table}[!htbp]
  \footnotesize
  \setlength{\tabcolsep}{.5pt}
 \centering
  \begin{tabularx}{\textwidth}{X *{6}{*{5}{c} @{\hskip 6pt}} *{2}{*{6}{c} @{\hskip 6pt}} *{6}{c}}
    \toprule
    \multicolumn{1}{l}{\textbf{compiler}}
    & \multicolumn{5}{c}{clang-3.0}
    & \multicolumn{5}{c}{clang-3.9}
    & \multicolumn{5}{c}{clang-7.1}
    & \multicolumn{5}{c}{clang-7.1}
    & \multicolumn{5}{c}{clang-9/11}
    & \multicolumn{5}{c}{clang-9/11}
    & \multicolumn{6}{c}{gcc-all}
    & \multicolumn{6}{c}{gcc-all}
    & \multicolumn{6}{c}{gcc-10.3}\\
    \multicolumn{1}{l}{\textbf{arch}}
    & \multicolumn{5}{c}{i386/i686}
    & \multicolumn{5}{c}{i386/i686}
    & \multicolumn{5}{c}{i386}
    & \multicolumn{5}{c}{i686}
    & \multicolumn{5}{c}{i386}
    & \multicolumn{5}{c}{i686}
    & \multicolumn{6}{c}{i386}
    & \multicolumn{6}{c}{i686}
    & \multicolumn{6}{c}{arm}
    \\
    \multicolumn{1}{l}{\textbf{opt-level}}
    & \texttt{0} & \texttt{1} & \texttt{2} & \texttt{3} & \texttt{3\textsuperscript{-}}
    & \texttt{0} & \texttt{1} & \texttt{2} & \texttt{3} & \texttt{3\textsuperscript{-}}
    & \texttt{0} & \texttt{1} & \texttt{2} & \texttt{3} & \texttt{3\textsuperscript{-}}
    & \texttt{0} & \texttt{1} & \texttt{2} & \texttt{3} & \texttt{3\textsuperscript{-}}
    & \texttt{0} & \texttt{1} & \texttt{2} & \texttt{3} & \texttt{3\textsuperscript{-}}
    & \texttt{0} & \texttt{1} & \texttt{2} & \texttt{3} & \texttt{3\textsuperscript{-}}
    & \texttt{0} & \texttt{0\textsuperscript{+}} & \texttt{1} & \texttt{2} & \texttt{3} & \texttt{3\textsuperscript{-}}
    & \texttt{0} & \texttt{0\textsuperscript{+}} & \texttt{1} & \texttt{2} & \texttt{3} & \texttt{3\textsuperscript{-}}
     & \texttt{0} & \texttt{0\textsuperscript{+}} & \texttt{1} & \texttt{2} & \texttt{3} & \texttt{3\textsuperscript{-}}\\
    \midrule
    \input{./ressources/tables/clang-gcc} %
  \end{tabularx}
  \caption{Preservation of constant-time for several programs compiled with
    \texttt{gcc} or \texttt{clang} for \texttt{i386}, \texttt{i686} or \texttt{arm}
    architectures and optimization levels \texttt{O0},
    \texttt{O0\textsuperscript{+}} \texttt{O1}, \texttt{O2}, \texttt{O3} or
    \texttt{O3\textsuperscript{-}}. \ccmark{} indicates that a program is secure
    whereas \textsc{\textcolor{red}{c}} (resp. \textsc{\textcolor{red}{m}})
    indicates that a program is insecure due to secret-dependent control-flow
    (resp.\ memory access). \texttt{gcc-all} denotes versions 5.4.0, 6.2.0,
    7.2.0, 8.3.0 and 10.2.0 of \texttt{gcc}.}\label{tab:verif_compilers}
\end{table}

%% file: ressources/tables/clang-gcc.tex
\texttt{ct\_select\_v1} & \ccmark{} & \ccmark{} & \textsc{\textcolor{red}{c}} & \textsc{\textcolor{red}{c}} & - & \ccmark{} & \ccmark{} & \textsc{\textcolor{red}{c}} & \textsc{\textcolor{red}{c}} & - & \ccmark{} & \ccmark{} & \textsc{\textcolor{red}{c}} & \textsc{\textcolor{red}{c}} & \textsc{\textcolor{red}{c}} & \ccmark{} & \ccmark{} & \textsc{\textcolor{red}{m}} & \textsc{\textcolor{red}{m}} & \textsc{\textcolor{red}{m}} & \ccmark{} & \ccmark{} & \textsc{\textcolor{red}{c}} & \textsc{\textcolor{red}{c}} & \textsc{\textcolor{red}{c}} & \ccmark{} & \ccmark{} & \textsc{\textcolor{red}{m}} & \textsc{\textcolor{red}{m}} & \textsc{\textcolor{red}{m}} & \ccmark{} & \ccmark{} & \ccmark{} & \ccmark{} & \ccmark{} & \ccmark{} & \ccmark{} & \ccmark{} & \ccmark{} & \ccmark{} & \ccmark{} & \ccmark{} & \ccmark{} & \ccmark{} & \ccmark{} & \ccmark{} & \ccmark{} & \ccmark{}\\
\texttt{ct\_select\_v2} & \ccmark{} & \textsc{\textcolor{red}{c}} & \textsc{\textcolor{red}{c}} & \textsc{\textcolor{red}{c}} & - & \ccmark{} & \textsc{\textcolor{red}{c}} & \textsc{\textcolor{red}{c}} & \textsc{\textcolor{red}{c}} & - & \ccmark{} & \textsc{\textcolor{red}{c}} & \textsc{\textcolor{red}{c}} & \textsc{\textcolor{red}{c}} & \textsc{\textcolor{red}{c}} & \ccmark{} & \textsc{\textcolor{red}{m}} & \textsc{\textcolor{red}{m}} & \textsc{\textcolor{red}{m}} & \textsc{\textcolor{red}{m}} & \ccmark{} & \textsc{\textcolor{red}{c}} & \textsc{\textcolor{red}{c}} & \textsc{\textcolor{red}{c}} & \textsc{\textcolor{red}{c}} & \ccmark{} & \textsc{\textcolor{red}{m}} & \textsc{\textcolor{red}{m}} & \textsc{\textcolor{red}{m}} & \textsc{\textcolor{red}{m}} & \ccmark{} & \ccmark{} & \ccmark{} & \ccmark{} & \ccmark{} & \ccmark{} & \ccmark{} & \ccmark{} & \ccmark{} & \ccmark{} & \ccmark{} & \ccmark{} & \ccmark{} & \ccmark{} & \ccmark{} & \ccmark{} & \ccmark{} & \ccmark{}\\
\texttt{ct\_select\_v3} & \ccmark{} & \ccmark{} & \ccmark{} & \ccmark{} & - & \ccmark{} & \textsc{\textcolor{red}{c}} & \textsc{\textcolor{red}{c}} & \textsc{\textcolor{red}{c}} & - & \ccmark{} & \textsc{\textcolor{red}{c}} & \textsc{\textcolor{red}{c}} & \textsc{\textcolor{red}{c}} & \textsc{\textcolor{red}{c}} & \ccmark{} & \textsc{\textcolor{red}{m}} & \textsc{\textcolor{red}{m}} & \textsc{\textcolor{red}{m}} & \textsc{\textcolor{red}{m}} & \ccmark{} & \textsc{\textcolor{red}{c}} & \textsc{\textcolor{red}{c}} & \textsc{\textcolor{red}{c}} & \textsc{\textcolor{red}{c}} & \ccmark{} & \textsc{\textcolor{red}{m}} & \textsc{\textcolor{red}{m}} & \textsc{\textcolor{red}{m}} & \textsc{\textcolor{red}{m}} & \ccmark{} & \ccmark{} & \ccmark{} & \ccmark{} & \ccmark{} & \ccmark{} & \ccmark{} & \ccmark{} & \ccmark{} & \ccmark{} & \ccmark{} & \ccmark{} & \ccmark{} & \ccmark{} & \ccmark{} & \ccmark{} & \ccmark{} & \ccmark{}\\
\texttt{ct\_select\_v4} & \ccmark{} & \textsc{\textcolor{red}{c}} & \textsc{\textcolor{red}{c}} & \textsc{\textcolor{red}{c}} & - & \ccmark{} & \textsc{\textcolor{red}{c}} & \textsc{\textcolor{red}{c}} & \textsc{\textcolor{red}{c}} & - & \ccmark{} & \textsc{\textcolor{red}{c}} & \textsc{\textcolor{red}{c}} & \textsc{\textcolor{red}{c}} & \textsc{\textcolor{red}{c}} & \ccmark{} & \textsc{\textcolor{red}{m}} & \textsc{\textcolor{red}{m}} & \textsc{\textcolor{red}{m}} & \textsc{\textcolor{red}{m}} & \ccmark{} & \textsc{\textcolor{red}{c}} & \textsc{\textcolor{red}{c}} & \textsc{\textcolor{red}{c}} & \textsc{\textcolor{red}{c}} & \ccmark{} & \textsc{\textcolor{red}{m}} & \textsc{\textcolor{red}{m}} & \textsc{\textcolor{red}{m}} & \textsc{\textcolor{red}{m}} & \ccmark{} & \ccmark{} & \ccmark{} & \ccmark{} & \ccmark{} & \ccmark{} & \ccmark{} & \ccmark{} & \ccmark{} & \ccmark{} & \ccmark{} & \ccmark{} & \ccmark{} & \ccmark{} & \ccmark{} & \ccmark{} & \ccmark{} & \ccmark{}\\
\texttt{naive\_select} & \textsc{\textcolor{red}{c}} & \textsc{\textcolor{red}{c}} & \textsc{\textcolor{red}{c}} & \textsc{\textcolor{red}{c}} & - & \textsc{\textcolor{red}{c}} & \textsc{\textcolor{red}{c}} & \textsc{\textcolor{red}{c}} & \textsc{\textcolor{red}{c}} & - & \textsc{\textcolor{red}{c}} & \textsc{\textcolor{red}{c}} & \textsc{\textcolor{red}{c}} & \textsc{\textcolor{red}{c}} & \textsc{\textcolor{red}{c}} & \textsc{\textcolor{red}{c}} & \textsc{\textcolor{red}{m}} & \textsc{\textcolor{red}{m}} & \textsc{\textcolor{red}{m}} & \textsc{\textcolor{red}{m}} & \textsc{\textcolor{red}{c}} & \textsc{\textcolor{red}{c}} & \textsc{\textcolor{red}{c}} & \textsc{\textcolor{red}{c}} & \textsc{\textcolor{red}{c}} & \textsc{\textcolor{red}{c}} & \textsc{\textcolor{red}{m}} & \textsc{\textcolor{red}{m}} & \textsc{\textcolor{red}{m}} & \textsc{\textcolor{red}{m}} & \textsc{\textcolor{red}{c}} & \textsc{\textcolor{red}{c}} & \textsc{\textcolor{red}{c}} & \textsc{\textcolor{red}{c}} & \textsc{\textcolor{red}{c}} & \textsc{\textcolor{red}{c}} & \textsc{\textcolor{red}{c}} & \textsc{\textcolor{red}{c}} & \ccmark{} & \ccmark{} & \ccmark{} & \textsc{\textcolor{red}{c}} & \textsc{\textcolor{red}{c}} & \textsc{\textcolor{red}{c}} & \ccmark{} & \ccmark{} & \ccmark{} & \textsc{\textcolor{red}{c}}\\
\midrule
\texttt{sort} & \ccmark{} & \textsc{\textcolor{red}{c}} & \textsc{\textcolor{red}{c}} & \textsc{\textcolor{red}{c}} & - & \ccmark{} & \textsc{\textcolor{red}{c}} & \textsc{\textcolor{red}{c}} & \textsc{\textcolor{red}{c}} & - & \ccmark{} & \ccmark{} & \ccmark{} & \ccmark{} & \ccmark{} & \ccmark{} & \ccmark{} & \ccmark{} & \ccmark{} & \ccmark{} & \ccmark{} & \textsc{\textcolor{red}{c}} & \textsc{\textcolor{red}{c}} & \textsc{\textcolor{red}{c}} & \textsc{\textcolor{red}{c}} & \ccmark{} & \textsc{\textcolor{red}{m}} & \textsc{\textcolor{red}{m}} & \textsc{\textcolor{red}{m}} & \textsc{\textcolor{red}{m}} & \textsc{\textcolor{red}{c}} & \textsc{\textcolor{red}{c}} & \ccmark{} & \ccmark{} & \ccmark{} & \ccmark{} & \textsc{\textcolor{red}{c}} & \textsc{\textcolor{red}{c}} & \ccmark{} & \ccmark{} & \ccmark{} & \ccmark{} & \textsc{\textcolor{red}{c}} & \textsc{\textcolor{red}{c}} & \ccmark{} & \ccmark{} & \ccmark{} & \ccmark{}\\
\texttt{sort\_multiplex} & \ccmark{} & \textsc{\textcolor{red}{c}} & \textsc{\textcolor{red}{c}} & \textsc{\textcolor{red}{c}} & - & \ccmark{} & \textsc{\textcolor{red}{c}} & \textsc{\textcolor{red}{c}} & \textsc{\textcolor{red}{c}} & - & \ccmark{} & \ccmark{} & \ccmark{} & \ccmark{} & \ccmark{} & \ccmark{} & \ccmark{} & \ccmark{} & \ccmark{} & \ccmark{} & \ccmark{} & \textsc{\textcolor{red}{c}} & \textsc{\textcolor{red}{c}} & \textsc{\textcolor{red}{c}} & \textsc{\textcolor{red}{c}} & \ccmark{} & \textsc{\textcolor{red}{m}} & \textsc{\textcolor{red}{m}} & \textsc{\textcolor{red}{m}} & \textsc{\textcolor{red}{m}} & \textsc{\textcolor{red}{c}} & \textsc{\textcolor{red}{c}} & \ccmark{} & \ccmark{} & \ccmark{} & \ccmark{} & \textsc{\textcolor{red}{c}} & \textsc{\textcolor{red}{c}} & \ccmark{} & \ccmark{} & \ccmark{} & \ccmark{} & \textsc{\textcolor{red}{c}} & \textsc{\textcolor{red}{c}} & \ccmark{} & \ccmark{} & \ccmark{} & \ccmark{}\\
\texttt{sort\_naive} & \textsc{\textcolor{red}{c}} & \textsc{\textcolor{red}{c}} & \textsc{\textcolor{red}{c}} & \textsc{\textcolor{red}{c}} & - & \textsc{\textcolor{red}{c}} & \textsc{\textcolor{red}{c}} & \textsc{\textcolor{red}{c}} & \textsc{\textcolor{red}{c}} & - & \textsc{\textcolor{red}{c}} & \textsc{\textcolor{red}{c}} & \textsc{\textcolor{red}{c}} & \textsc{\textcolor{red}{c}} & \textsc{\textcolor{red}{c}} & \textsc{\textcolor{red}{c}} & \textsc{\textcolor{red}{m}} & \textsc{\textcolor{red}{m}} & \textsc{\textcolor{red}{m}} & \textsc{\textcolor{red}{m}} & \textsc{\textcolor{red}{c}} & \textsc{\textcolor{red}{c}} & \textsc{\textcolor{red}{c}} & \textsc{\textcolor{red}{c}} & \textsc{\textcolor{red}{c}} & \textsc{\textcolor{red}{c}} & \textsc{\textcolor{red}{m}} & \textsc{\textcolor{red}{m}} & \textsc{\textcolor{red}{m}} & \textsc{\textcolor{red}{m}} & \textsc{\textcolor{red}{c}} & \textsc{\textcolor{red}{c}} & \textsc{\textcolor{red}{c}} & \textsc{\textcolor{red}{c}} & \textsc{\textcolor{red}{c}} & \textsc{\textcolor{red}{c}} & \textsc{\textcolor{red}{c}} & \textsc{\textcolor{red}{c}} & \textsc{\textcolor{red}{c}} & \textsc{\textcolor{red}{c}} & \textsc{\textcolor{red}{c}} & \textsc{\textcolor{red}{c}} & \textsc{\textcolor{red}{c}} & \textsc{\textcolor{red}{c}} & \ccmark{} & \ccmark{} & \ccmark{} & \textsc{\textcolor{red}{c}}\\
\midrule
\texttt{HACL*-utility \(\times 11\)} & \ccmark{} & \ccmark{} & \ccmark{} & \ccmark{} & - & \ccmark{} & \ccmark{} & \ccmark{} & \ccmark{} & - & \ccmark{} & \ccmark{} & \ccmark{} & \ccmark{} & \ccmark{} & \ccmark{} & \ccmark{} & \ccmark{} & \ccmark{} & \ccmark{} & \ccmark{} & \ccmark{} & \ccmark{} & \ccmark{} & \ccmark{} & \ccmark{} & \ccmark{} & \ccmark{} & \ccmark{} & \ccmark{} & \ccmark{} & \ccmark{} & \ccmark{} & \ccmark{} & \ccmark{} & \ccmark{} & \ccmark{} & \ccmark{} & \ccmark{} & \ccmark{} & \ccmark{} & \ccmark{} & \ccmark{} & \ccmark{} & \ccmark{} & \ccmark{} & \ccmark{} & \ccmark{}\\
\texttt{OpenSSL-utility \(\times 13\)} & \ccmark{} & \ccmark{} & \ccmark{} & \ccmark{} & - & \ccmark{} & \ccmark{} & \ccmark{} & \ccmark{} & - & \ccmark{} & \ccmark{} & \ccmark{} & \ccmark{} & \ccmark{} & \ccmark{} & \ccmark{} & \ccmark{} & \ccmark{} & \ccmark{} & \ccmark{} & \ccmark{} & \ccmark{} & \ccmark{} & \ccmark{} & \ccmark{} & \ccmark{} & \ccmark{} & \ccmark{} & \ccmark{} & \ccmark{} & \ccmark{} & \ccmark{} & \ccmark{} & \ccmark{} & \ccmark{} & \ccmark{} & \ccmark{} & \ccmark{} & \ccmark{} & \ccmark{} & \ccmark{} & \ccmark{} & \ccmark{} & \ccmark{} & \ccmark{} & \ccmark{} & \ccmark{}\\
\midrule
\texttt{tea-enc/dec \(\times 2\)} & \ccmark{} & \ccmark{} & \ccmark{} & \ccmark{} & - & \ccmark{} & \ccmark{} & \ccmark{} & \ccmark{} & - & \ccmark{} & \ccmark{} & \ccmark{} & \ccmark{} & \ccmark{} & \ccmark{} & \ccmark{} & \ccmark{} & \ccmark{} & \ccmark{} & \ccmark{} & \ccmark{} & \ccmark{} & \ccmark{} & \ccmark{} & \ccmark{} & \ccmark{} & \ccmark{} & \ccmark{} & \ccmark{} & \ccmark{} & \ccmark{} & \ccmark{} & \ccmark{} & \ccmark{} & \ccmark{} & \ccmark{} & \ccmark{} & \ccmark{} & \ccmark{} & \ccmark{} & \ccmark{} & \ccmark{} & \ccmark{} & \ccmark{} & \ccmark{} & \ccmark{} & \ccmark{} \\ \bottomrule

%% file: ressources/tables/scale_total_summary.tex
\begin{table}[ht]
  \begin{tabular}{lrrrrrrrrr}
   \toprule
   &
  \#\(\text{I}_{unr}\) &
  \#\(\text{I}_{unr}\)/s &
  \#Q &
  \begin{tabular}{@{}c@{}}\(\text{\#Q}_{\text{e}}\)\end{tabular} &
  \begin{tabular}{@{}c@{}}\(\text{\#Q}_{\text{i}}\)\end{tabular} &
  Time &
  \hourglass{} &
  \ccmark{} &
  \cxmark{} \\
    \midrule

    \textit{SC}    &           248k &     3.9 &       158k &         14k &       143k &    64296 &        16 &     280 &        42 \\
    \textit{RelSE} &           349k &     6.2 &        90k &         17k &        73k &    56428 &        13 &     283 &        42 \\
    \midrule
    \brelse{}      &          22.8M &  6238 &         3700 &        2408 &       1292 &     3657 &         0 &     296 &        42 \\
  \bottomrule
\end{tabular}
\caption{ \brelse{} vs.\ standard approaches.\label{tab:scale_total_summary}}
\end{table}

%% file: ressources/tables/scale_optims.tex
\begin{table}[!htbp]
\begin{tabularx}{\textwidth}{X l rrrrrrrrr}
  \toprule
  & Version &
  \#\(\text{I}_{unr}\) &
  \#\(\text{I}_{unr}\)/s &
  \#Q &
  \begin{tabular}{c}\(\text{\#Q}_{\text{e}}\)\end{tabular} &
  \begin{tabular}{c}\(\text{\#Q}_{\text{i}}\)\end{tabular} &
  Time &                                                      
  \hourglass{} &
  \ccmark{} &
  \cxmark{} \\
  \midrule
  \multirow{3}{*}{\begin{tabular}{@{}l}
                    \textbf{Standard RelSE}\\ \textbf{(without \textit{FlyRow})}
                  \end{tabular}}
  &\textit{RelSE}       &           349k &     6.2 &        90148 &         17428 &        72720 &    56429 &        13 &     283 &        42 \\
  &+ \textit{Unt}       &           414k &     9.9 &        48648 &         20601 &        28047 &    41852 &         7 &     289 &        42 \\
  &+ \textit{FP}        &           437k &    12.7 &        35100 &         21834 &        13266 &    34471 &         7 &     289 &        42 \\
  \midrule
  \multirow{3}{*}{\begin{tabular}{@{}l}
                    \textbf{\brelse{}}\\ \textbf{(with \textit{FlyRow})}
                  \end{tabular}}
  &\textit{RelSE + FlyRow}  &         22.8M &  4450 &      3738 &          2408 &         1330 &     5127 &         0 &     296 &        42 \\
  &+ \textit{Unt}         &         22.8M &  4429 &      3738 &          2408 &         1330 &     5151 &         0 &     296 &        42 \\
  &+ \textit{FP}          &         22.8M &  6238 &      3700 &          2408 &         1292 &     3658 &         0 &     296 &        42 \\
  \bottomrule
\end{tabularx}

\caption{Performances of \brelse{} simplifications.\label{tab:scale_optims}}
\end{table}

%% file: ressources/tables/sse_relse_overhead2.tex
\begin{table}[htpb]
  \setlength{\tabcolsep}{4pt}
\begin{tabular}{lrrrrrrrr}
  \toprule
  \multicolumn{1}{l}{Version} &
  \multicolumn{1}{c}{\#\(\text{I}_{unr}\)} &
  \multicolumn{1}{c}{\#\(\text{I}_{unr}\)/s} &
  \multicolumn{1}{c}{\#Q} &
  \multicolumn{1}{c}{Time} &
  \multicolumn{1}{c}{\hourglass}\\
  \midrule%
  \textit{SE}           &           522k &     19.5 &        24444 &    26814 &         6 \\
  \textit{SE+PostRow}~\cite{DBLP:conf/lpar/FarinierDBL18}
                        &           628k &     29.2 &        29389 &    21475 &         4 \\
  \textit{SE+FlyRow}    &          22.8M &  12531.1 &          534 &     1817 &         0 \\
  \bottomrule
\end{tabular}
\hfill
\begin{tabular}{lrrrrrrrr}
  \toprule
  \multicolumn{1}{l}{Version} &
  \multicolumn{1}{c}{\#\(\text{I}_{unr}\)} &
  \multicolumn{1}{c}{\#\(\text{I}_{unr}\)/s} &
  \multicolumn{1}{c}{\#Q} &
  \multicolumn{1}{c}{Time} &
  \multicolumn{1}{c}{\hourglass}\\
  \midrule%
\textit{RelSE}         &           349k &      6.2 &        90148 &    56428 &        13 \\
\textit{RelSE+PostRow} &           317k &      5.3 &        65834 &    60295 &        15 \\
\brelse{}              &          22.8M &   6237.7 &         3700 &     3657 &         0 \\

  \bottomrule
\end{tabular}
\caption{Performances of RelSE compared to standard SE with/without
  binary level simplifications.\label{tab:sse_relse_overhead2}}
\end{table}

%% file: ressources/tables/naive_implementations.tex
\bgroup{}
  \small
  \setlength{\tabcolsep}{1.2pt}
 \centering
  \begin{tabularx}{\linewidth}{X *{2}{*{4}{c} @{\hskip 10pt}} *{4}{c}}
    \toprule
    & \multicolumn{4}{l}{clang-all-versions}
    & \multicolumn{4}{l}{gcc-5.4/6.2/7.2}
    & \multicolumn{4}{l}{gcc-8.3/10.2}\\
    & \texttt{-O0} & \texttt{-O1} & \texttt{-O2} & \texttt{-O3}
    & \texttt{-O0} & \texttt{-O1} & \texttt{-O2} & \texttt{-O3}
    & \texttt{-O0} & \texttt{-O1} & \texttt{-O2} & \texttt{-O3}\\
    \midrule
    \texttt{loop} & \ccmark{} & \ccmark{} & \cxmark{} & \cxmark{} & \ccmark{} & \ccmark{} & \cxmark{} & \cxmark{} & \ccmark{} & \ccmark{} & \cxmark{} & \cxmark{} \\
    \texttt{memset} & \ccmark{} & \ccmark{} & \cxmark{} & \cxmark{} & \ccmark{} & \ccmark{} & \cxmark{} & \cxmark{} & \ccmark{} & \ccmark{} & \cxmark{} & \cxmark{} \\
    \texttt{bzero} & \ccmark{} & \ccmark{} & \cxmark{} & \cxmark{} & \ccmark{} & \circled{\cxmark{}} & \cxmark{} & \cxmark{} & \ccmark{} & \circled{\ccmark{}} & \cxmark{} & \cxmark{} \\
    \bottomrule
  \end{tabularx}
  \captionof{table}{Preservation of secret-erasure for naive scrubbing
    functions.}\label{tab:naive_implementations}
\egroup{}

%% file: ressources/listings/volatile_function.tex
\begin{lstlisting}[style=nonumbers,caption={OpenSSL scrubbing function~\cite{OpenSSL_cleanse}.
    Relies on volatile function pointer.},label={lst:volatile_func}]
typedef void *(*memset_t)(void *, int, size_t);
static volatile memset_t memset_func = memset;
void scrub(char *buf, size_t size) {
  memset_func(buf, 0, size);
}
\end{lstlisting}

%% file: ressources/tables/volatile_function.tex
  \setlength{\tabcolsep}{2.5pt}
 \centering
    \begin{tabularx}{\textwidth}{ *{4}{c} X *{4}{c} }
    \toprule
    \multicolumn{4}{c}{clang-all-versions} & &
    \multicolumn{4}{c}{gcc-all-versions} \\
    \texttt{-O0} & \texttt{-O1} & \texttt{-O2} & \texttt{-O3} & &
    \texttt{-O0} & \texttt{-O1} & \texttt{-O2} & \texttt{-O3}\\
    \midrule
    \ccmark{} & \ccmark{} & \ccmark{} & \ccmark{} & &
    \ccmark{} & \ccmark{} & \cxmark{} & \cxmark{} \\
    \bottomrule
  \end{tabularx}
  \captionof{table}{Preservation of secret-erasure with volatile function
    pointer as implemented in \cref{lst:volatile_func}.}\label{tab:volatile_function}

%% file: ressources/listings/volatile_ptr.tex
\begin{center}
\begin{lstlisting}[numbers=left,caption={Different usage of \lstinline{volatile} type qualifier before scrubbing memory.},label={lst:volatile_ptr}]
volatile char *vbuf = (volatile char *) buf; // Pointer-to-volatile <@\label{line:ptr-to-volatile}@>
char * volatile vbuf = (char *volatile) buf; // Volatile ptr (pointer is volatile itself)<@\label{line:volatile-ptr}@>
volatile char *volatile vbuf = (volatile char *volatile) buf; // Volatile ptr-to-volatile<@\label{line:volatile-ptr-to-volatile}@>
\end{lstlisting}
\end{center}

%% file: ressources/tables/volatile_ptr.tex
\begin{table}[!htbp]
  \small
  \setlength{\tabcolsep}{1pt}
  \centering
  \begin{tabularx}{.48\textwidth}{X *{4}{c} @{\hskip 12pt} *{4}{c} }
    \toprule
    \multicolumn{1}{c}{}
    & \multicolumn{4}{@{}c}{clang-all-versions}
    & \multicolumn{4}{@{}c}{gcc-all-versions} \\
    \multicolumn{1}{c}{}
    & \texttt{-O0} & \texttt{-O1} & \texttt{-O2} & \texttt{-O3}
    & \texttt{-O0} & \texttt{-O1} & \texttt{-O2} & \texttt{-O3} \\
    \midrule
    \texttt{ptr\_to\_volatile\_loop} & \ccmark{} & \ccmark{} & \ccmark{} & \ccmark{}
    & \ccmark{} & \ccmark{} & \ccmark{} & \ccmark{} \\
    \texttt{volatile\_ptr\_loop} & \ccmark{} & \ccmark{} & \ccmark{} & \ccmark{}
    & \ccmark{} & \ccmark{} & \cxmark{} & \cxmark{} \\
    \texttt{vol\_ptr\_to\_vol\_loop} & \ccmark{} & \ccmark{} & \ccmark{} & \ccmark{}
    & \ccmark{} & \ccmark{} & \ccmark{} & \ccmark{} \\
    \bottomrule
  \end{tabularx}
  \hfill
  \begin{tabularx}{.495\textwidth}{X *{4}{c} @{\hskip 12pt} *{4}{c} }
    \toprule
    \multicolumn{1}{c}{}
    & \multicolumn{4}{@{}c}{clang-all-versions}
    & \multicolumn{4}{@{}c}{gcc-all-versions} \\
    \multicolumn{1}{c}{}
    & \texttt{-O0} & \texttt{-O1} & \texttt{-O2} & \texttt{-O3}
    & \texttt{-O0} & \texttt{-O1} & \texttt{-O2} & \texttt{-O3} \\
    \midrule
    \texttt{ptr\_to\_volatile\_memset} & \ccmark{} & \ccmark{} & \cxmark{} & \cxmark{}
    & \ccmark{} & \ccmark{} & \cxmark{} & \cxmark{} \\
    \texttt{volatile\_ptr\_memset} & \ccmark{} & \ccmark{} & \ccmark{} & \ccmark{}
    & \ccmark{} & \ccmark{} & \cxmark{} & \cxmark{} \\
    \texttt{vol\_ptr\_to\_vol\_memset} & \ccmark{} & \ccmark{} & \ccmark{} & \ccmark{}
    & \ccmark{} & \ccmark{} & \cxmark{} & \cxmark{} \\
    \bottomrule
  \end{tabularx}
  \caption{Preservation of secret-erasure with volatile data pointers. \ccmark{}
    indicates that a program is secure and \cxmark{} that it is
    insecure.}\label{tab:volatile_ptr}
\end{table}

%% file: ressources/listings/memory_barriers.tex
\begin{center}
  \begin{lstlisting}[numbers=left,caption={Different implementation of memory barriers.},label={lst:memory_barriers},mathescape=false]
__asm__ __volatile__("":::"memory");                          // memory_barrier_simple <@\label{line:memory_barriers:simple}@>
__asm__ __volatile__("mfence" ::: "memory");                  // memory_barrier_mfence <@\label{line:memory_barriers:fence}@>
__asm__ __volatile__("lock; addl $0,0(%%esp)" ::: "memory");  // memory_barrier_lock <@\label{line:memory_barriers:lock}@>
__asm__ __volatile__("": :"r"(buf) :"memory");                // memory_barrier_ptr <@\label{line:memory_barriers:ptr}@>
\end{lstlisting}
\end{center}

%% file: ressources/listings/weak_symbols.tex
\begin{center}
  \begin{lstlisting}[numbers=left,caption={Libsodium implementation of weak symbols for memory scrubbing.},label={lst:weak_symbols},float=ht]
__attribute__((weak)) void _sodium_dummy_symbol(void *const  pnt, const size_t len) {
  (void) pnt; (void) len;
}
void scrub(char *buf, size_t size) {
  memset(buf, 0, size);
  _sodium_dummy_symbol(buf, size);
}
\end{lstlisting}
\end{center}

%% file: ressources/tables/secret-erasure_no-dse_gcc.tex
\begin{table}[!htbp]
  \begin{minipage}[t]{0.6\textwidth}
  \vspace{0pt}
  \small
  \setlength{\tabcolsep}{1.5pt}
\begin{tabularx}{\textwidth}{X *{4}{c}@{\hskip 10pt} *{4}{c}}
    \toprule
    \multicolumn{1}{c}{}
    & \multicolumn{4}{c}{gcc-5.4.0 6.2.0 7.2.0}
    & \multicolumn{4}{c}{gcc-8.3.0 10.2.0}\\
    \multicolumn{1}{c}{}
    & \texttt{\ -O3} & \texttt{\ dse} & \texttt{tdse} & \texttt{both}
    & \texttt{\ -O3} & \texttt{\ dse} & \texttt{tdse} & \texttt{both} \\
    \midrule
\textbf{\texttt{loop}} & \cxmark{} & \cxmark{} & \cxmark{} & \cxmark{} & \cxmark{} & \cxmark{} & \cxmark{} & \cxmark{}\\
\texttt{memset} & \cxmark{} & \cxmark{} & \cxmark{} & \circled{\ccmark{}} & \cxmark{} & \cxmark{} & \cxmark{} & \circled{\ccmark{}}\\
\texttt{bzero} & \cxmark{} & \circled{\ccmark{}} & \cxmark{} & \ccmark{} & \cxmark{} & \cxmark{} & \cxmark{} & \circled{\ccmark{}}\\
\midrule
\texttt{explicit\_bzero} & \ccmark{} & \ccmark{} & \ccmark{} & \ccmark{} & \ccmark{} & \ccmark{} & \ccmark{} & \ccmark{}\\
\texttt{memset\_s} & \ccmark{} & \ccmark{} & \ccmark{} & \ccmark{} & \ccmark{} & \ccmark{} & \ccmark{} & \ccmark{}\\
\midrule
\textbf{\texttt{volatile\_func}} & \cxmark{} & \cxmark{} & \cxmark{} & \cxmark{} & \cxmark{} & \cxmark{} & \cxmark{} & \cxmark{}\\
\midrule
\texttt{ptr\_to\_volatile\_loop} & \ccmark{} & \ccmark{} & \ccmark{} & \ccmark{} & \ccmark{} & \ccmark{} & \ccmark{} & \ccmark{}\\
\texttt{ptr\_to\_volatile\_memset} & \cxmark{} & \cxmark{} & \cxmark{} & \circled{\ccmark{}} & \cxmark{} & \cxmark{} & \cxmark{} & \circled{\ccmark{}}\\
\textbf{\texttt{volatile\_ptr\_loop}} & \cxmark{} & \cxmark{} & \cxmark{} & \cxmark{} & \cxmark{} & \cxmark{} & \cxmark{} & \cxmark{}\\
\texttt{volatile\_ptr\_memset} & \cxmark{} & \cxmark{} & \circled{\ccmark{}} & \ccmark{} & \cxmark{} & \cxmark{} & \circled{\ccmark{}} & \ccmark{}\\
\texttt{\texttt{vol\_ptr\_to\_vol\_loop}} & \ccmark{} & \ccmark{} & \ccmark{} & \ccmark{} & \ccmark{} & \ccmark{} & \ccmark{} & \ccmark{}\\
\texttt{vol\_ptr\_to\_vol\_memset} & \cxmark{} & \cxmark{} & \circled{\ccmark{}} & \ccmark{} & \cxmark{} & \cxmark{} & \circled{\ccmark{}} & \ccmark{}\\
\midrule
\texttt{memory\_barrier\_all} (\(\times 4\)) & \ccmark{} & \ccmark{} & \ccmark{} & \ccmark{} & \ccmark{} & \ccmark{} & \ccmark{} & \ccmark{}\\
\midrule
\texttt{weak\_symbols} & \ccmark{} & \ccmark{} & \ccmark{} & \ccmark{} & \ccmark{} & \ccmark{} & \ccmark{} & \ccmark{}\\
\bottomrule
\end{tabularx}
\end{minipage}\hfill
\begin{minipage}[t]{0.37\textwidth}
  \vspace{0pt}
  \caption{Preservation of secret-erasure for several scrubbing functions
    compiled by \texttt{gcc}, with selective optimization disabling.
    \texttt{-O3} is the baseline optimization level, \texttt{dse} (resp.
    \texttt{tdse}) indicates that the \texttt{-fno-dse} (resp.
    \texttt{-fno-tree-dse}) arguments have been passed to \texttt{gcc} to
    disable \texttt{dse} (resp. \texttt{tree-dse}) optimization, and
    ``\texttt{both}'' indicates that both optimizations have been disabled. Circles
    highlight the least set of options that must be disabled to make the program
    secure and \textbf{bold} highlights programs that are insecure in all our
    configurations.}\label{tab:secret-erasure_no-dse_gcc}
\end{minipage}
\end{table}

%% file: discussion.tex
\myparagraph{Limitations of the technique.} %
\newercontent{The relational symbolic execution introduced in this paper handles
  loops and recursion with unrolling. Unrolling still enables exhaustive
  exploration for programs without unbounded loops such as \texttt{tea} or
  \texttt{donna}. However, for programs with unbounded loops, such as stream
  ciphers \texttt{salsa20} or \texttt{chacha20} it leads to unexplored program
  paths, and hence might miss violations\footnote{In our experiments we fix the
    input size for these programs, but we could also keep it symbolic and
    restrict it to a given range, which would extend security guarantees for all
    input sizes in this range.}. A possible solution to enable sound analysis
  for program with unbounded loops would be to use relational loop
  invariants~\cite{DBLP:conf/ccs/BalliuDG14}---however, it would sacrifice
  bug-finding.
  Similarly, indirect jump targets are only enumerated up to a given bound,
  which might lead to unexplored program paths and consequently missed
  violations\footnote{\brelse{} detects and records incomplete jump target
    enumerations and, if it cannot find any vulnerabilities, it returns
    “unknown” instead of “secure”.}. However, we did not encounter incomplete
  enumerations in our experiments: in the cryptographic primitives that we analyzed
  indirect jumps had a single (or few) target. %
  Finally, any register or part of the memory that is concretized in the initial
  state of the symbolic execution might lead to unexplored program behaviors and
  missed violations. In \brelse{}, memory and register are symbolic by default
  and any concretization (e.g.\ setting the initial value of \lstinline{esp}, or
  which memory addresses are initialized from the binary) must be made
  \emph{explicitly} by the user.}

\newercontent{%
  The definition of secret-erasure used in this paper is conservative in the sense that it forbids
  secret-dependent branches (and hence related implicit flows). We leave for
  future work the exploration of alternative (less conservative) definitions
  that could either declassify secret-dependent conditions, or allow
  secret-secret dependent conditions as long as both branches produce the same
  observations. } %
\newcontent{Finally, \brelse{} restricts to a sequential semantics and hence
  cannot detect Spectre vulnerabilities~\cite{DBLP:conf/sp/KocherHFGGHHLM019},
  however the technique has recently been adapted to a speculative
  semantics~\cite{DBLP:conf/sp/DanielBR20}.}

\myparagraph{Implementation limitations.} %
The implementation of \brelse{} shows limitations commonly found in research
prototypes: %
it does not support dynamic libraries (binaries must be statically linked or
stubs must be provided for external function calls), \newercontent{it does not support dynamic
memory allocation (data structures must be statically allocated)}, it does not
implement predefined system call stubs, \newcontent{it does not support
  multi-threading,} and it does not support floating point instructions. These
problems are orthogonal to the core contribution of this paper.
Moreover, the prototype is already efficient on real-world case studies.

\myparagraph{Threats to validity in experimental evaluation.} %
We assessed the effectiveness of our tool on several known secure and
insecure real-world cryptographic binaries, many of them taken from
prior studies. All results have been crosschecked with the expected
output, and manually reviewed in case of deviation.

Our prototype is implemented as part of
\binsec{}~\cite{DBLP:conf/wcre/DavidBTMFPM16}, whose efficiency and robustness
have been demonstrated in prior large scale studies on both adversarial code and
managed
code~\cite{DBLP:conf/sp/BardinDM17,DBLP:conf/kbse/RecoulesBBMP19,DBLP:conf/cav/FarinierBBP18,DBLP:conf/issta/DavidBFMPTM16}.
The IR lifting part has been positively evaluated in an external
study~\cite{DBLP:conf/kbse/KimFJJOLC17} and the symbolic engine features
aggressive formula optimizations~\cite{DBLP:conf/lpar/FarinierDBL18}. All our
experiments use the same search heuristics (depth-first) and, for
bounded-verification, smarter heuristics do not change the performance. %
Regarding the solver, we also tried Z3~\cite{DBLP:conf/tacas/MouraB08} and
confirmed the better performance of Boolector.

Finally, we compare our tool to our own versions of \textit{SC} and
\textit{RelSE}, %
primarily because none of the existing tools can be easily adapted for our
setting, and also because it allows us to compare very close implementations.

%% file: rw.tex
Related work has already been lengthily discussed along the paper. We
add here only a few additional discussions, as well as an overview of
existing SE-based tools for information flow analysis
(\cref{tab:comparison_se}) partly taken
from~\cite{DBLP:conf/uss/AlmeidaBBDE16}.

\input{./ressources/tables/comparison_se}

\myparagraph{Self-composition and SE} %
has first been used by Milushev \emph{et
  al.}~\cite{DBLP:conf/forte/MilushevBC12}. They use type-directed
self-composition and dynamic symbolic execution to find bugs of
\emph{noninterference} but they do not address scalability and their
experiments are limited to toy programs. The main issues here are the
quadratic explosion of the search space (due to the necessity of
considering diverging paths) and the complexity of the underlying
formulas.
Later works~\cite{DBLP:conf/sec/DoBH15,DBLP:conf/ccs/BalliuDG14} suffer from
the same problems.

\emph{Instead of considering the general case of noninterference, we
  focus on properties that relate traces following the same path, and
  we show that it remains tractable for SE with adequate
  optimizations.}

\myparagraph{Relational symbolic execution.} %
\emph{Shadow symbolic
  execution}~\cite{DBLP:conf/icse/PalikarevaKC16,DBLP:conf/icse/CadarP14}
aims at efficiently testing evolving software by focusing on the new
behaviors introduced by a patch.
It introduces the idea of \emph{sharing formulas} across two
executions in the same SE instance. The term \emph{relational symbolic
  execution} has been coined more
recently~\cite{farinaRelationalSymbolicExecution2019} but this work is
limited to a simple toy imperative language and do not address
scalability.

\emph{We maximize sharing between pairs of executions, as ShadowSE
  does, but we also develop specific optimizations tailored to the
  case of information-flow analysis at binary-level. Experiments show
  that our optimizations are crucial in this context. }

\myparagraph{Scaling SE for information flow analysis.} %
Only three previous works in this category achieve scalability, yet at
the cost of either precision or soundness.
Wang {et al.}~\cite{DBLP:conf/uss/WangWLZW17} and
Subramanyan \emph{et al.}~\cite{DBLP:conf/date/SubramanyanMKMF16}
sacrifice soundness for scalability (no bounded-verification).  The
former performs symbolic execution on fully concrete traces and only
symbolizes secrets.
The latter concretizes memory accesses.
In both cases, they may miss feasible paths as well as
vulnerabilities.
Brotzman \emph{et al.}~\cite{DBLP:conf/sp/BrotzmanLZTK19} take the opposite
side and sacrifice precision for scalability (no bug-finding). %
Their analysis scales by over-approximating loops and resetting the
symbolic state at chosen code locations.

\emph{We adopt a different approach and scale by heavy formula optimizations,
  allowing us to keep both correct bug-finding and correct
  bounded-verification.} Interestingly, our method is faster than these
approximated ones. %
Moreover, our technique is compatible with the previous approximations for
extra-scaling.

\myparagraph{Other methods for constant-time analysis.}
\emph{Dynamic approaches} for constant-time are precise (they find real
violations) but limited to a subset of the execution traces, hence they are not
complete. These techniques include statistical
analysis~\cite{DBLP:conf/date/ReparazBV17}, dynamic binary
instrumentation~\cite{langleyImperialVioletCheckingThat2010,DBLP:conf/acsac/WichelmannMES18},
dynamic symbolic execution (DSE)~\cite{DBLP:conf/memocode/ChattopadhyayBR17},
\newcontent{or fuzzing~\cite{DBLP:conf/icst/HeEC20}}.

\emph{Static approaches} %
based on sound static
analysis~\cite{DBLP:conf/popl/Agat00,DBLP:conf/icisc/MolnarPSW05,DBLP:conf/ccs/BartheBCLP14,bacelaralmeidaFormalVerificationSidechannel2013,DBLP:conf/esorics/BlazyPT17,DBLP:conf/cav/KopfMO12,DBLP:conf/uss/DoychevFKMR13,DBLP:conf/pldi/DoychevK17,DBLP:conf/uss/AlmeidaBBDE16,DBLP:conf/cc/RodriguesPA16}
give formal guarantees that a program is free from timing-side-channels but they
cannot find bugs when a program is rejected.

\newcontent{Aside from a posteriori analysis, correct-by-design
approaches~\cite{DBLP:conf/ccs/AlmeidaBBBGLOPS17,DBLP:conf/uss/BondHKLLPRST17,DBLP:conf/secdev/CauligiSBJHJS17,DBLP:conf/ccs/ZinzindohoueBPB17}
require to reimplement cryptographic primitives from scratch.}
\newcontent{Program transformations have been proposed to automatically
  transform insecure programs into (variations of) constant-time
  programs~\cite{DBLP:journals/ijisec/KopfM07,DBLP:conf/popl/Agat00,DBLP:conf/icisc/MolnarPSW05,DBLP:journals/tcad/ChattopadhyayR18,DBLP:conf/issta/WuGS018,DBLP:conf/sp/BrotzmanLZTK19,DBLP:conf/sp/CoppensVBS09,DBLP:conf/uss/RaneLT15,DBLP:conf/issta/WuGS018,DBLP:conf/ccs/BorrelloDQG21,DBLP:conf/cgo/SoaresP21}.
  In particular, Raccoon and Constantine consider a constant-time leakage model
  and seem promising, however they operate at LLVM level and do not protect
  against violations introduced by backend compiler passes. Therefore,
  \brelse{} is complementary to these techniques, as it can be used for
  investigating code patterns and backend optimizations that may introduce
  constant-time violations in backend compiler passes.}

\myparagraph{Other methods for secret-erasure.} %
\newcontent{Compiler or OS-based secure
  deallocation~\cite{DBLP:conf/uss/ChowPGR05,DBLP:conf/sigopsE/GarfinkelPCR04}
  have been proposed but require compiler or OS-support, in contrast this work
  focuses on application-based secret-erasure.}

\citeauthor{DBLP:conf/csfw/ChongM05}~\cite{DBLP:conf/csfw/ChongM05} introduce
the first framework to specify erasure policies which has been later refined to
express richer policies using a knowledge-based
approach~\cite{DBLP:conf/iciss/TedescoHS11}, and cryptographic data
deletion~\cite{DBLP:conf/csfw/AskarovMDC15}. These works focus on expressing
complex secret-erasure policies, but are not directly applicable to concrete
languages.
\citeauthor{hansen2006non}~\cite{hansen2006non} propose the first application of
a simple secret-erasure policy for a concrete language (i.e., Java Card
Bytecode), which ensures that secrets are unavailable after program termination.
Our definition of secret erasure is close to theirs and directly applicable for
binary-level verification.

Most enforcement mechanisms for erasure are language-based and rely on type
systems to enforce information flow
control~\cite{DBLP:conf/esop/HuntS08,DBLP:conf/csfw/ChongM08,DBLP:conf/csfw/AskarovMDC15,DBLP:conf/nordsec/TedescoRS10,DBLP:conf/sp/NanevskiBG11}.
Secretgrind~\cite{Secretgrind2020}, a dynamic taint tracking tool based on
Valgrind~\cite{DBLP:conf/pldi/NethercoteS07} to track secret data in memory, is
the closest work to ours, with the main difference being that it uses dynamic
analysis and permits implicit flows, while we use static analysis and forbid
implicit flows.

The problem of (non-)preservation of secret-erasure by compilers is well
known~\cite{CWE14CompilerRemoval, DBLP:conf/sp/DSilvaPS15,
  DBLP:conf/csfw/BessonDJ19, DBLP:conf/uss/YangJOLL17,
  DBLP:conf/eurosp/SimonCA18}. To remedy it, a notion of information
flow-preserving program transformation has been
proposed~\cite{DBLP:conf/csfw/BessonDJ19} but this approach
requires to compile programs using CompCert~\cite{DBLP:journals/cacm/Leroy09}
and does not apply to already compiled binaries.
Finally, preservation of erasure functions by compilers has been studied
manually~\cite{DBLP:conf/uss/YangJOLL17}, and we further this line of work
by proposing an extensible framework for \emph{automating} the process.

%% file: ressources/tables/comparison_se.tex
\begin{table}[htbp!]
  \begin{tabularx}{\linewidth}{Xlllc@{\,}c@{\,}crr}
    \toprule
    Tool & Target & NI & Technique & P&BV&BF & \(\approx\)XP max & \(\text{I}_{u}\)/s \\
    \midrule
    RelSym~\cite{farinaRelationalSymbolicExecution2019} &
    imp-for & \ccmark & RelSE & \ccmark&\ccmark&\ccmark & 10 LoC & \textsc{na} \\
    IF-exploit~\cite{DBLP:conf/sec/DoBH15} &
    Java & \ccmark & self-composition & \ccmark&\ccmark&\ccmark & 20 LoC & \textsc{na} \\
    Type-SC-SE~\cite{DBLP:conf/forte/MilushevBC12} &
    C & \ccmark & type-based self-comp. + DSE & \cxmark&\cxmark&\ccmark & 20 LoC & \textsc{na} \\
    Casym~\cite{DBLP:conf/sp/BrotzmanLZTK19} &
    LLVM & \ccmark & self-comp. + over-approx & \ccmark&\ccmark&\cxmark & 200 LoC (C) & \textsc{na} \\
    \textsc{ENCoVer}~\cite{DBLP:conf/csfw/BalliuDG12} &
    Java bytecode & \ccmark & epistemic logic + DSE & \cxmark&\ccmark&\ccmark& 33k I\textsubscript{u} & 186 \\
    \midrule
    IF-low-level~\cite{DBLP:conf/ccs/BalliuDG14} &
    \textbf{binary} & \ccmark & self-comp. + invariants & \ccmark&\ccmark&\cxmark & 250 \(\text{I}_{s}\)& \textsc{na}\\
    IF-firmware~\cite{DBLP:conf/date/SubramanyanMKMF16} &
    \textbf{binary} & \cxmark & self-comp. + concretizations & \cxmark&\cxmark&\ccmark & 500k \(\text{I}_{u}\) & 260 \\
    CacheD~\cite{DBLP:conf/uss/WangWLZW17} &
    \textbf{binary} & \cxmark & tainting + concretizations & \cxmark&\cxmark&\ccmark & 31M \(\text{I}_{u}\) & 2010 \\
    \midrule
    \brelse{} &
                \textbf{binary} & \cxmark & RelSE + formula simplifications &
                \cxmark&\ccmark&\ccmark & 10M \(\text{I}_{u}\) & 3861  \\
    \bottomrule
  \end{tabularx}
  \caption{Comparison of SE-based tools for information flow analysis. NI
    indicates whether the tool handles general non-interference (diverging
    paths) or not. In the column ``Technique'', DSE stands for dynamic symbolic
    execution. P stands for Proof, BV for Bounded-Verification, BF for
    Bug-Finding. %
    \(\approx\)XP max indicates the approximate size of the use cases where LoC
    stands for Lines of Code, \(\text{I}_{s}\) for static instructions, and
    \(\text{I}_{u}\) for unrolled instructions. Finally, \textsc{na} indicates
    Non-Applicable. }\label{tab:comparison_se}
\end{table}

%% file: proofs.tex
\section{Proofs}\label{app:proofs}
This section details the proofs of the theorems and lemmas introduced in \cref{sec:proofs}.
\begin{itemize}
  \item The proof of \cref{lemma:stuckinsecure} is given in
        \Cref{app:stuckinsecure},
  \item The proof of \cref{lemma:completeness} is given in
        \Cref{app:completenesslemma}
  \item The proof of \cref{thm:correctness}, which states the correctness of our
        RelSE, is given in \Cref{app:correctness},
  \item The proof of \cref{thm:bv}, which states that our RelSE is correct for
        bounded-verification of constant-time, is given in \Cref{app:bv}.
      \end{itemize}

\subsection{Proof of \cref{lemma:stuckinsecure}}\label{app:stuckinsecure}
\Cref{lemma:stuckinsecure} expresses that when the symbolic evaluation is stuck
on a state \(s_k\), there exist concrete configurations derived from \(s_k\)
which produce distinct leakages.

\stuckinsecure*

\begin{proof}
  Because \(s_k\) is stuck, we know from \Cref{hyp:stuck} that an expression
  \(\rel{\varphi}\) is leaked and that \(\secleak(\rel{\varphi}, \pc)\)
  evaluates to false in the symbolic evaluation of \(s_k\). We also know that
  there exists a model \(M\) such that
  \(M \sat \pc \wedge \lproj{\rel{\varphi}} \neq \rproj{\rel{\varphi}}\). %
  Let \(c_k, c_k'\) be concrete configurations such that %
  \(c_{k} \concsym{l}{M} s_{k}\), and \(c_{k}' \concsym{r}{M} s_{k}\). %
  To show that the program is not constant-time at step \(k\), that is %
  \(c_k \cleval{\leakvar} c_{k+1}\) and \(c_k' \cleval{\leakvar'} c_{k+1}'\)
  with \(\leakvar \neq \leakvar'\), %
  we proceed case by case on the symbolic evaluation, restricting to cases where
  \(\secleak(\rel{\varphi}, \pc)\) might evaluate to false. There are two main
  cases:
  \begin{enumerate}
    \item Symbolic execution is stuck on the evaluation of an expression,
    \item Symbolic evaluation is stuck on the evaluation of an instruction.
  \end{enumerate}

  \myparagraph{SE stuck on an expression.} First, we consider the case where the
  symbolic execution is stuck on the evaluation of an expression, restricting to
  the rule \rulename{load} as other cases cannot be stuck.
  \begin{itemize}
  \item \textbf{Case \rulename{load}}: In the symbolic execution, the
    expression \(\load{\ e_{idx}}\) is evaluated with
    \(\econf{\regmap,\smem,\pc}{e_{idx}} \eeval{} \rel{\iota}\) and
    the index \(\rel{\iota}\) is leaked.
    Assuming \(\secleak(\rel{\iota}, \pc)\) evaluates to false with the model
    \(M\), then
    \(M(\lproj{\rel{\iota}}) \neq M(\rproj{\rel{\iota}})\).
    Moreover, because \(c_{k} \concsym{l}{M} s_{k}\) and
    \(c_{k}' \concsym{r}{M} s_{k}\), we have from \cref{def:concsym}
    that
    \(c_{k}\ e_{idx} \ceeval{} M(\lproj{\rel{\iota}}) \text{ and } %
    c_{k}'\ e_{idx} \ceeval{} M(\rproj{\rel{\iota}})\). Because
    performing a step in the concrete execution leaks the value of
    \(e_{idx}\), we have
    \(c_k \cleval{\leakvar \cdot M(\lproj{\rel{\iota}})} c_{k+1}\) and
    \(c_k' \cleval{\leakvar' \cdot M(\rproj{\rel{\iota}})} c_{k+1}'\)
    with \({M(\lproj{\rel{\iota}}) \neq M(\rproj{\rel{\iota}})}\),
    meaning that the execution is not constant-time at step \(k\).
  \end{itemize}

  \myparagraph{SE stuck on an instruction.} Second, we consider the case where the
  symbolic evaluation is not stuck on the evaluation of an expression, but is
  stuck on the evaluation of an instruction. This can happen when evaluating
  rules \rulename{store}, \rulename{ite} and \rulename{d\_jump}.

  \begin{itemize}
  \item \textbf{Case \rulename{store}}: In the symbolic execution, the
    instruction \(\store{\ e_{idx}}{e_{val}}\) is evaluated with
    \(\econf{\regmap,\smem,\pc}{e_{idx}} \eeval{} \rel{\iota}\) and
    the index \(\rel{\iota}\) is leaked. %
    Similarly as in case \rulename{load}, we can show that we have
    \(c_k \cleval{\leakvar \cdot M(\lproj{\rel{\iota}})} c_{k+1}\) and
    \(c_k' \cleval{\leakvar' \cdot M(\rproj{\rel{\iota}})} c_{k+1}'\)
    with \(M(\lproj{\rel{\iota}}) \neq M(\rproj{\rel{\iota}})\),
    meaning that the execution is not constant-time at step \(k\).

  \item \textbf{Case \rulename{ite-true} and \rulename{ite-false}}: In
    the symbolic execution, the instruction
    \(ite~e~?~l_{true} : l_{false}\) is evaluated with
    \(\econf{\regmap,\smem,\pc}{e} \eeval{} \rel{\varphi}\) and
    \(\rlift{eq_0}\ \rel{\varphi}\) is leaked.
    Assuming \(\secleak(\rlift{eq_0}\ \rel{\varphi}, \pc)\) evaluates
    to false with the model \(M\), then
    \({M(\lproj{\rel{\varphi}} = 0)} \neq M({\rproj{\rel{\varphi}} =
      0})\).
    Moreover, because \(c_{k} \concsym{l}{M} s_{k}\) and
    \(c_{k}' \concsym{r}{M} s_{k}\), we have from \cref{def:concsym}
    that
    \(c_{k}\ e \ceeval{} M(\lproj{\rel{\varphi}}) \text{ and } %
    c_{k}'\ e \ceeval{} M(\rproj{\rel{\varphi}})\).
    Therefore in one of the concrete executions, \(e\) evaluates to 0,
    the rule \rulename{ite-false} is applied and the location
    \(l_{false}\) is leaked; whereas in the other execution, \(e\) does
    not evaluate to 0, the rule \rulename{ite-true} is applied and the
    location \(l_{true}\) is leaked.
    Finally, we have \(c_k \cleval{\leakvar \cdot l_{true}} c_{k+1}\)
    and \(c_k' \cleval{\leakvar' \cdot l_{false}} c_{k+1}'\) (or
    \(c_k \cleval{\leakvar \cdot l_{false}} c_{k+1}\) and
    \(c_k' \cleval{\leakvar' \cdot l_{true}} c_{k+1}'\)), meaning
    that the execution is not constant-time at step \(k\).

  \item \textbf{Case \rulename{d\_jump}}: In the symbolic execution,
    the instruction \(\mathtt{goto}\ e\) is evaluated with
    \(\econf{\regmap,\smem,\pc}{e} \eeval{} \rel{\varphi}\). Let
    \(M(\lproj{\rel{\varphi}}) = \mathtt{bv_l}\) and
    \(M(\rproj{\rel{\varphi}}) = \mathtt{bv_r}\).
    Assume \(\secleak(\rel{\varphi}, \pc)\) evaluates to false with the
    model \(M\), then \(\mathtt{bv_l} \neq \mathtt{bv_r}\).
    Moreover, because \(c_{k} \concsym{l}{M} s_{k}\) and
    \(c_{k}' \concsym{r}{M} s_{k}\), we have from \cref{def:concsym}
    that \(c_{k}\ e \ceeval{} \mathtt{bv_l} \text{ and } %
    c_{k}'\ e \ceeval{} \mathtt{bv_r}\). In the concrete execution,
    \(l_l \mydef \toloc(\mathtt{bv_l})\) and
    \(l_r \mydef \toloc(\mathtt{bv_r})\) are leaked---note that
    \(\toloc\) is defined for \(\mathtt{bv_l}\) and \(\mathtt{bv_r}\)
    from \cref{hyp:stuck}. Because
    \(\mathtt{bv_l} \neq \mathtt{bv_r}\) and \(\toloc\) is a
    one-to-one correspondence, we have \(l_l \neq l_r\). %
    Therefore, we have \(c_k \cleval{\leakvar \cdot l_l} c_{k+1}\) and
    \(c_k' \cleval{\leakvar' \cdot l_r} c_{k+1}'\) with
    \(l_l \neq l_r\), meaning that the execution is not constant-time
    at step \(k\).
  \end{itemize}
\end{proof}

\subsection{Proof of \cref{lemma:completeness}}\label{app:completenesslemma}
\cref{lemma:completeness} expresses that symbolic evaluation does not get stuck
up to \(k\), then for each pair of concrete executions following the same path
up to \(k\), there exists a corresponding symbolic execution.

\completenesslemma*

\begin{proof} \emph{(Induction on the number of steps k).} Case
  \(k = 0\) is trivial.

  Let \(c_{k-1}\) and \(c_{k-1}'\) be concrete configurations and
  \(s_{k-1}\) a symbolic configuration for which the inductive
  hypothesis holds. We need to show that \cref{lemma:completeness}
  still holds at step \(k\), meaning that for each concrete steps
  \(c_{k-1} \cleval{\leakvar} c_{k}\) and
  \(c_{k-1}' \cleval{\leakvar'} c_{k}'\), there exists a symbolic
  configuration \(s_{k}\) and a model \(M'\) such that
  \({c_{k} \concsym{l}{M'} s_{k}} ~\wedge~ {c_{k}' \concsym{r}{M'}
    s_{k}}\) holds. This amounts to show that:
  \begin{enumerate}[label=(\roman*)]
  \item\label{item:loc1} the location in \(s_k\) is the same as the
    location in \(c_k\) and \(c_k'\),
  \item\label{item:model_sat} there exists a model \(M_k\) such that
    \(M_k \sat \pc_k\)
  \item\label{item:equiv1} for all expressions \(e\), either symbolic
    evaluation gets stuck, or
    \(\econf{\regmap_{k}, \smem_{k}}{e} \eeval{} \rel{\varphi}\)
    and
    \(M_k(\proj{\rel{\varphi}}) = \mathtt{bv} \iff c_k'~e \ceeval{}
    \mathtt{bv}\).
  \end{enumerate}
  We can proceed case by case on the concrete evaluation of
  \(c_{k-1}\) and \(c_{k-1}'\). Note that from \cref{def:concsym},
  \(c_{k-1}\) and \(c_{k-1}'\) are at the same program location and
  therefore need to evaluate the same instruction.

  \myparagraph{Case \rulename{store}.} Consider that an instruction
  \(\store{\ e_{idx}}{e_{val}}\) is evaluated at step \(k-1\). %
  Let \(\rel{\iota}\) be the symbolic index and \(\rel{\nu}\) be the
  symbolic value, meaning that
  \(\econf{\regmap,\smem,\pc}{e_{idx}} \eeval{} \rel{\iota} \) and
  \(\econf{\regmap,\smem,\pc}{e_{val}} \eeval{} \rel{\nu}\).
  \cref{item:loc1} directly follows from the concrete and the symbolic
  evaluation rules that both just increment the location by 1.
  Next, we build the new model \(M_k\) as:
  \begin{align*}
    M_k \mydef M[\smem_k \mapsto \pair{m_l}{m_r}] \quad
    \text{ where }& m_l \mydef M(\lproj{\smem})[M(\lproj{\rel{\iota}}) \mapsto M(\lproj{\rel{\nu}})]\\
    \text{ and }&   m_r \mydef M(\rproj{\smem})[M(\rproj{\rel{\iota}}) \mapsto M(\rproj{\rel{\nu}})]
  \end{align*}
  Intuitively, \(M_k\) is equal to the old model \(M\) in which we add
  the new symbolic memory \(\smem_k\), mapping to the concrete value
  of the old memory \(M(\smem)\) where the index \(M(\rel{\iota})\)
  maps to the value \(M(\rel{\nu})\). Notice that \(M_k \sat \pc_k\)
  because \(M \sat \pc\) and we only added new definitions (thus not
  changing satisfiability) from \(M\) and \(\pc\) to \(M_k\) and
  \(\pc_k\). Thus \cref{item:model_sat} holds.
  Finally, we show \cref{item:equiv1} by induction on the structure of
  expressions: for any expression \(e\), if symbolic evaluation does
  not get stuck then
  \(M_k(\lproj{\rel{\varphi}}) = \mathtt{bv} \iff c_k~e \ceeval{}
  \mathtt{bv}\) holds (case of the right projection is analogous).
  Note that only the memory is updated from step \(k-1\) to step
  \(k\), meaning that \(c_k\), \(s_k\), and \(M_k\) only differ from
  \(c_{k-1}\), \(s_{k-1}\) and \(M\) on expressions involving the
  memory. Thus, we only need to consider the rule \rulename{load}, as
  the proof for other rules directly follows from
  \(c_{k-1}\concsym{p}{M} s_{k-1}\), \cref{def:concsym}, and the
  definition of \(M_k\).

  Assume an expression \(\load{e}\) such that \(s_k\) does not get
  stuck and let
  \(\econf{\regmap_{k}, \smem_{k}}{e} \eeval{} \rel{\iota'}\) and
  \(\econf{\regmap_{k}, \smem_{k}}{\load{e}} \eeval{} \rel{\nu'}\).
  We show that if \cref{item:equiv1} holds for the expression \(e\),
  then it holds for the expression \(\load{e}\). Formally, we must
  show that if %
  \({M_k(\lproj{\rel{\iota'}}) = \mathtt{bv_{idx}}} \iff {{c_k~{e}
      \ceeval{} \mathtt{bv_{idx}}}}\) %
  then %
  \({M_k(\lproj{\rel{\nu'}}) = \mathtt{bv_{val}}} \iff {c_k~{\load{e}}
    \ceeval{} \mathtt{bv_{val}}}\).

  \noindent
  First, we can rewrite \(M_k(\lproj{\rel{\nu'}})\) as
  \begin{align*}
    M_k(\lproj{\rel{\nu'}})
    &= M_k(select( \lproj{\smem_k},\lproj{\rel{\iota'}})) \text{ by symbolic rule
      \rulename{load}}\\
    &= M(\lproj{\smem})[M(\lproj{\rel{\iota}}) \mapsto M(\lproj{\rel{\nu}})][M_k (\lproj{\rel{\iota'}})] \text{ by def.\ of \(M_k\)}\\
  \end{align*}
  From this point, there are two cases, either \begin{enumerate*}[label=(\alph*)]
  \item the address of the load is the same as the address of the
    previous store,
  \item the address of the load is different from the address of the
    previous store.
  \end{enumerate*}

  \begin{enumerate}[label=(\alph*)]
  \item The address of the load is the same as the address of the
    previous store, i.e.,
    \(M_k (\lproj{\rel{\iota'}}) = M(\lproj{\rel{\iota}})\).
    Therefore we have \({M_k(\lproj{\rel{\nu'}}) = M(\lproj{\rel{\nu}})}\).
    From the induction hypothesis, the concrete index of the load
    evaluates to \(M_k(\lproj{\rel{\iota'}})\), that is
    \({c_k~{e} \ceeval{} M_k( \lproj{\rel{\iota'}})}\) which can be
    rewritten as \({c_k~{e} \ceeval{} M(\lproj{\rel{\iota}})}\).
    From concrete rule \rulename{store} and
    \(c_{k-1} \concsym{l}{M} s_{k-1}\), we know that the concrete
    memory from \(c_{k-1}\) to \(c_k\) is updated at index
    \(M(\lproj{\rel{\iota}})\) to map to the value
    \(M(\lproj{\rel{\nu}})\). Thus, we have
    \(c_k~\load{e} \ceeval{} M(\lproj{\rel{\nu}})\) and by rewriting,
    \({c_k~{\load{e}} \ceeval{} M_k(\lproj{\rel{\nu'}})}\).  Therefore
    we have shown that
    \({M_k(\lproj{\rel{\nu'}}) = \mathtt{bv_{val}}} \iff {c_k~{\load{e}} \ceeval{} \mathtt{bv_{val}}}\).

  \item The address of the load is different from the address of
    the previous store, i.e.,
    \(M_k (\lproj{\rel{\iota'}}) \neq M(\lproj{\rel{\iota}})\). Therefore we have
    \({M_k (\lproj{\rel{\nu'}}) = M(\lproj{\smem})[M_k (\lproj{\rel{\iota'}})]}\).
    From the induction hypothesis, the concrete index of the load
    evaluates to \(M_k(\lproj{\rel{\iota'}})\), that is
    \(c_k~{e} \ceeval{} M_k(\lproj{\rel{\iota'}})\).
    From concrete rule \rulename{store}, we know that the concrete
    memory from \(c_{k-1}\) to \(c_k\) is only updated at address
    \(M(\lproj{\rel{\iota}})\) and untouched at address
    \(M_k(\lproj{\rel{\iota'}})\). %
    Plus, we know from \(c_{k-1} \concsym{l}{M} s_{k-1}\) that address
    \(M_k(\lproj{\rel{\iota'}})\) maps to
    \(M(\lproj{\smem})[M_k(\lproj{\rel{\iota'}})]\) in \(c_{k-1}\).
    Therefore, in configuration \(c_k\), index
    \(M_k(\lproj{\rel{\iota'}})\) maps to
    \(M(\lproj{\smem})[M_k(\lproj{\rel{\iota'}})]\) which, by
    rewriting, leads to
    \({c_k~{\load{e}} \ceeval{} M_k(\lproj{\rel{\nu'}})}\).  Therefore
    we have shown that
    \({M_k(\lproj{\rel{\nu'}}) = \mathtt{bv_{val}}} \iff {c_k~{\load{e}} \ceeval{} \mathtt{bv_{val}}}\).
  \end{enumerate}

  \myparagraph{Case \rulename{d\_jump}.} Consider that an instruction
  \(\mathtt{goto}\ e\) is evaluated at step \(k-1\). %
  Notice that \(e\) evaluates to the same value in \(c_{k-1}\) and
  \(c_{k-1}'\) because both executions follow the same path, and let
  \(\mathtt{bv}\) be this concrete value. Concrete rule
  \rulename{d\_jump} just sets the next location to
  \(l_c = \toloc(\mathtt{bv})\).
  Symbolic evaluation of rule \rulename{d\_jump}, evaluates \(e\) to a
  symbolic value \(\rel{\varphi}\), computes a model
  \(M' \solver{} \pi \wedge \lproj{\rel{\varphi}} \wedge
  \rproj{\rel{\varphi}}\) and sets the next location to
  \(l_{s} = \toloc(M'(\lproj{\rel{\varphi}}))\).
  Note that from \cref{hyp:stuck} and because symbolic execution does
  not get stuck, the rule can be non-deterministically applied with
  any model \(M'\) satisfying the constraint.

  Therefore, we can apply the rule with \(M' = M\) which gives
  \(l_{s} = \toloc(M(\lproj{\rel{\varphi}}))\). %
  Moreover, from the hypothesis \(c_{k-1} \concsym{p}{M} s_{k-1}\) and
  \cref{def:concsym}, we have
  \(M(\lproj{\rel{\varphi}}) = \mathtt{bv}\) so
  \(l_{s} = \toloc(\mathtt{bv})\). Therefore we have shown how to make
  a symbolic step such that \(l_{s_{k}} = l_{c_{k}}\) and
  \Cref{item:loc1} holds.
  Finally, \cref{item:model_sat,item:equiv1} directly follow from
  \(c_{k-1} \concsym{p}{M} s_{k-1}\) and \cref{def:concsym} because
  symbolic and concrete evaluation of expressions are not modified by
  rule \rulename{d\_jump}.

  \myparagraph{Other cases.} \textbf{Case \rulename{s\_jump}} is trivial
  as only the location is updated to the same static value in both
  concrete and symbolic evaluation. %
  \textbf{Case \rulename{assign}} is similar to case \rulename{store}
  (and simpler because it only requires to reason about the value and
  not the index). %
  Finally \textbf{Cases \rulename{ite\_true} and
    \rulename{ite\_false}} are a bit similar to case
  \rulename{d\_jump} and rely on the fact that both concrete
  executions follow the same path, therefore a symbolic rule (either
  \rulename{ite\_true} or \rulename{ite\_fale}) can be applied to
  match the execution of both \(c_k\) and \(c_k'\).

  \myparagraph{Conclusion.} We have shown that we can perform a step in symbolic
  execution from \(s_{k-1}\) to a state \(s_k\), and there exists a model \(M'\)
  such that \(c_k \concsym{l}{M_k} s_k\) and \(c_k' \concsym{r}{M_k} s_k\).
\end{proof}

\subsection{Proof of \cref{thm:correctness}}\label{app:correctness}
\Cref{thm:correctness} claims the correctness of our symbolic execution, meaning
that for each symbolic execution and model \(M\) satisfying the path predicate,
the concretization of the symbolic execution with \(M\) corresponds to a valid
concrete execution.

\correctness*

\begin{proof} \emph{(Induction on the number of steps k).}  Case \(k = 0\) is trivial. %

  Consider symbolic configurations \(s_0\),
  \(s_{k-1} \mydef \iconf{\locvar}{\regmap}{\smem}{\pc}\) for which
  the induction hypothesis holds. That is, for each model \(M\) and
  configurations \(c_0\), \(c_{k-1}\) such that
  \(c_0 \concsym{p}{M} s_0\) and \(c_{k-1}\concsym{p}{M} s_{k-1}\), we
  have \(c_0 \cleval{}^{k-1} c_{k-1}\).
  Let \(s_k \mydef \iconf{\locvar_k}{\regmap_k}{\smem_k}{\pc_k}\)
  be the symbolic state such that \(s_{k-1} \ieval{} s_{k}\).  We need
  to show that for each model \(M_k\) and configurations \(c_0\) and
  \(c_k\) such that \(c_0 \concsym{p}{M_k} s_0\) and
  \(c_{k}\concsym{p}{M_k} s_{k}\), we have
  \(c_{0} \cleval{}^{k} c_k\).

  \myparagraph{Build steps \(0\) to \(k-1\).}
  First, we build the concrete execution from steps \(0\) to \(k-1\). Because
  \(M_k \sat \pc_k\) (from \cref{def:concsym}) and \(\pc\) is a sub-formula of
  \(\pc_k\) (from symbolic evaluation), we have \(M_k \sat \pc\), thus we can
  build \(c_{k-1}\) such that \(c_{k-1}\concsym{p}{M_k} s_{k-1}\). Similarly we
  can build \(c_0\) such that \(c_0 \concsym{p}{M_k} s_0\). Finally, from the
  induction hypothesis, we have \(c_0 \cleval{}^{k-1} c_{k-1}\).

  \myparagraph{Build step \(k-1\) to \(k\).}
  Second, we build the concrete execution from steps \(k-1\) to \(k\): we need
  to show that there is a step from \(c_{k-1}\) to \(c_k\) where
  \(c_k \concsym{p}{M_k} s_k\). %
  Because concrete semantics never gets stuck (\cref{hyp:concrete-stuck}) there
  is a state \(c_k' \mydef \cconf{\locvar'}{\cregmap'}{\cmem'}\) such that
  \(c_{k-1} \cleval{} c_{k}'\) and because the semantics is deterministic
  (\cref{hyp:deterministic}), this state is unique. Thus we need to show that
  \(c_{k}' = c_{k}\), that is \(c_{k}' \concsym{p}{M_k} s_{k}\). Because
  \(M_k \sat \pc_k\) (from \(c_{k}\concsym{p}{M_k} s_{k}\) and
  \cref{def:concsym}), this amounts to show that:
  \begin{enumerate}[label=(\roman*)]
  \item\label{item:loc2} the location in \(c_k'\) is the same as the location in \(s_k\),
  \item\label{item:equiv2} for all expression \(e\), either symbolic
    evaluation gets stuck, or
    \(\econf{\regmap_{k}, \smem_{k}}{e} \eeval{} \rel{\varphi}\)
    and
    \(M_k(\proj{\rel{\varphi}}) = \mathtt{bv} \iff c_k'~e \ceeval{}
    \mathtt{bv}\).
  \end{enumerate}
  We can proceed case by case on the concrete evaluation of
  \(c_{k-1}\). Note that from \cref{def:concsym}, \(c_{k-1}\) and
  \(s_{k-1}\) are at the same program location and therefore both evaluate
  the same instruction.

  \myparagraph{Case \rulename{store}.} Consider that an instruction
  \(\store{\ e_{idx}}{e_{val}}\) is evaluated at step \(k-1\). %
  Let \(\rel{\iota}\) be the symbolic index and \(\rel{\nu}\) be the
  symbolic value, meaning that
  \(\econf{\regmap,\smem,\pc}{e_{idx}} \eeval{} \rel{\iota} \) and
  \(\econf{\regmap,\smem,\pc}{e_{val}} \eeval{} \rel{\nu}\).
  \Cref{item:loc2} directly follows from the concrete and the symbolic
  evaluation rules that both just increment the location by 1.
  We prove \cref{item:equiv2} by induction on the structure of
  expressions: for any expression \(e\), if symbolic evaluation does
  not get stuck then
  \(M_k(\proj{\rel{\varphi}}) = \mathtt{bv} \iff c_k'~e \ceeval{}
  \mathtt{bv}\) holds. %
  Note that only the memory is updated from step \(k-1\) to step
  \(k\), meaning that \(c_k'\) and \(s_k\) only differ from
  \(c_{k-1}\) and \(s_{k-1}\) on expressions involving the
  memory. Thus, we only need to consider the rule \rulename{load}, as
  the proof for other rules directly follows from
  \(c_{k-1}\concsym{p}{M_k} s_{k-1}\) and \cref{def:concsym}.

  Assume an expression \(\load{e}\) such that \(s_k\) does not get
  stuck and let
  \(\econf{\regmap_{k}, \smem_{k}}{e} \eeval{} \rel{\iota'}\) and
  \(\econf{\regmap_{k}, \smem_{k}}{\load{e}} \eeval{} \rel{\nu'}\).
  We show that if \cref{item:equiv2} holds for the expression
  \(e\), then it holds for the expression \(\load{e}\). Formally, we
  must show that if %
  \({M_k(\proj{\rel{\iota'}}) = \mathtt{bv_{idx}}} \iff {{c_k'~{e}
      \ceeval{} \mathtt{bv_{idx}}}}\) %
  then %
  \({M_k(\proj{\rel{\nu'}}) = \mathtt{bv_{val}}} \iff {c_k'~{\load{e}} \ceeval{} \mathtt{bv_{val}}}\).

  \noindent
  First, we can rewrite \(M_k(\proj{\rel{\nu'}})\) as
  \begin{align*}
    M_k(\proj{\rel{\nu'}})
    &= M_k(select( \proj{\smem_k},\proj{\rel{\iota'}})) \text{ by symbolic rule
      \rulename{load}}\\
    &=  M_k(select(store(\proj{\smem_{k-1}}, \proj{\rel{\iota}}, \proj{\rel{\nu}}),\proj{\rel{\iota'}})) \text{ by symbolic rule
      \rulename{store}}\\
  \end{align*}
  From this point, there are two cases, either \begin{enumerate*}[label=(\alph*)]
  \item the address of the load is the same as the address of the
    previous store,
  \item the address of the load is different from the address of the
    previous store.
  \end{enumerate*}

  \begin{enumerate}[label=(\alph*)]
    \item The address of the load is the same as the address of the previous
          store, i.e., \(M_k (\proj{\rel{\iota'}}) = M_k(\proj{\rel{\iota}})\).
          Therefore, we have \({M_k(\proj{\rel{\nu'}}) = M_k(\proj{\rel{\nu}})}\).
    From the induction hypothesis, the concrete index of the load
    evaluates to \(M_k(\proj{\rel{\iota'}})\), that is
    \({c_k'~{e} \ceeval{} M_k( \proj{\rel{\iota'}})}\) which can be
    rewritten as \({c_k'~{e} \ceeval{} M_k(\proj{\rel{\iota}})}\).
    From concrete rule \rulename{store} and \break
    \mbox{\(c_{k-1} \concsym{p}{M_k} s_{k-1}\)}, we know that the concrete
    memory from \(c_{k-1}\) to \(c_k'\) is updated at index
    \(M_k(\proj{\rel{\iota}})\) to map to the value
    \(M_k(\proj{\rel{\nu}})\). Thus, we have
    \(c_k'~\load{e} \ceeval{} M_k(\proj{\rel{\nu}})\) and by rewriting,
    \({c_k'~{\load{e}} \ceeval{} M_k(\proj{\rel{\nu'}})}\).  Therefore
    we have shown that:
    \[{M_k(\proj{\rel{\nu'}}) = \mathtt{bv_{val}}} \iff {c_k'~{\load{e}} \ceeval{} \mathtt{bv_{val}}}\]

    \item The address of the load is different from the address of the previous
          store, i.e.,
          \(M_k (\proj{\rel{\iota'}}) \neq M_k(\proj{\rel{\iota}})\). Therefore:
          \[{M_k(\proj{\rel{\nu'}})} = M_k(select(\proj{\smem_{k-1}},\proj{\rel{\iota'}})) = M_k(\proj{\smem_{k-1}})[M_k(\proj{\rel{\iota'}})]\]
          From the induction hypothesis, the concrete index of the load
          evaluates to \(M_k(\proj{\rel{\iota'}})\), that is
          \(c_k'~{e} \ceeval{} M_k(\proj{\rel{\iota'}})\).
    From concrete rule \rulename{store}, we know that the concrete
    memory from \(c_{k-1}\) to \(c_k'\) is only updated at address
    \(M_k(\proj{\rel{\iota}})\) and untouched at address
    \(M_k(\proj{\rel{\iota'}})\). %
    Plus, we know from \(c_{k-1} \concsym{p}{M_k} s_{k-1}\) that address
    \(M_k(\proj{\rel{\iota'}})\) maps to
    \(M(\proj{\smem_{k-1}})[M_k(\proj{\rel{\iota'}})]\) in \(c_{k-1}\).
    Therefore, in configuration \(c_k'\), address
    \(M_k(\proj{\rel{\iota'}})\) also maps to
    \(M(\proj{\smem_{k-1}})[M_k(\proj{\rel{\iota'}})]\) which, by
    rewriting, leads to
    \({c_k'~{\load{e}} \ceeval{} M_k(\proj{\rel{\nu'}})}\).  Therefore
    we have shown that:
    \[{M_k(\proj{\rel{\nu'}}) = \mathtt{bv_{val}}} \iff {c_k'~{\load{e}} \ceeval{} \mathtt{bv_{val}}}\]
  \end{enumerate}

  \myparagraph{Case \rulename{d\_jump}.} Consider that an instruction
  \(\mathtt{goto}\ e\) is evaluated at step \(k-1\). %
  Let \(\rel{\varphi}\) be the symbolic value of \(e\), meaning that
  \(\econf{\regmap,\smem,\pc}{e} \eeval{} \rel{\varphi}\).
  Symbolic rule \rulename{d\_jump} sets the location to
  \(l_{s_{k}} = \toloc(M(\lproj{\rel{\varphi}}))\), which is equal to
  \(\toloc(M(\rproj{\rel{\varphi}}))\) because \(M\) satisfies the
  constraint \(\lproj{\rel{\varphi}} = \rproj{\rel{\varphi}}\). %
  Concrete rule \rulename{d\_jump} evaluates expression \(e\) to a
  concrete value \(\mathtt{bv}\) and sets the location to
  \(l_{c_{k}} = \toloc(\mathtt{bv})\). From the hypothesis
  \(c_{k-1} \concsym{p}{M_k} s_{k-1}\) and \cref{def:concsym} we have
  \(\mathtt{bv} = M_k(\proj{\rel{\varphi}})\), therefore
  \(l_{c_{k}} = l_{s_{k}}\).

  Finally, \cref{item:equiv2} directly follows from
  \(c_{k-1} \concsym{p}{M_k} s_{k-1}\) and \cref{def:concsym} because
  symbolic and concrete evaluation of expressions are not modified by
  rule \rulename{d\_jump}.

  \myparagraph{Other cases.} \textbf{Case \rulename{s\_jump}} is trivial as only
  the location is updated to the same static value in both concrete and symbolic
  evaluation. %
  \textbf{Case \rulename{assign}} is similar to case \rulename{store} (and
  simpler because it only requires to reason about the value and not the
  index). %
  Finally \textbf{Cases \rulename{ite\_true} and \rulename{ite\_false}} are
  similar to case \rulename{d\_jump}.

  \textbf{Conclusion.} We have shown that \(c_{k-1} \cleval{} c_k'\) with
  \(s_k \concsym{p}{M_k} c_k'\). Because \(\concsym{p}{M_k}\) is a tight
  relation, we have \(c_k' = c_k\) and thus \(c_{k-1} \cleval{} c_k\).
  Therefore, for each model \(M_k\) and configuration \(c_0\) and \(c_k\)
  such that \(c_0 \concsym{p}{M_k} s_0\) and
  \(c_{k} \concsym{p}{M_k} s_{k}\), we have \(c_{0} \cleval{}^{k} c_k\).
\end{proof}

\subsection{Proof of \cref{thm:bv}}\label{app:bv}
\Cref{thm:bv} claims that if symbolic execution does not get stuck due to a
satisfiable insecurity query, then the program is constant-time.

\boundedverif*

\begin{proof} \emph{(Induction on the number of steps k).} Case
  \(k = 0\) is trivial.

  Let $s_0$ be an initial symbolic configuration for which the
  symbolic evaluation never gets stuck. Let us consider a model \(M_0\)
  and concrete configurations \(\cconfvar_0 \concsym{l}{M_0} s_0\),
  \(\cconfvar'_0 \concsym{r}{M_0} s_0\), for which the induction
  hypothesis holds at step \(k-1\), meaning that for all
  \(c_{k-1} \mydef \cconf{\locvar}{\cregmap}{\cmem}\) and
  \(c_{k-1}' \mydef \cconf{\locvar'}{\cregmap'}{\cmem'}\) such that
  $c_0 \cleval{\leakvar}^{k-1} c_{k-1}$, $c_0' \cleval{\leakvar'}^{k-1} c_{k-1}'$,
  then \(\leakvar = \leakvar'\). We show that \cref{thm:bv} still
  holds at step \(k\).

  From \cref{lemma:completeness}, there exists a model \(M\) and a
  symbolic configuration
  \(s_{k-1} \mydef \iconf{\locvar_s}{\regmap}{\smem}{\pc}\) such
  that:
  \begin{equation}
    \label{eq:2}
    s_0 \ieval{}^{k-1} s_{k-1} ~\wedge~ c_{k-1} \concsym{l}{M} s_{k-1} ~\wedge~ c_{k-1}' \concsym{r}{M} s_{k-1}
  \end{equation}

  We show by contradiction that the leakages \(\mathtt{bv}\)
  and \(\mathtt{bv'}\) produced by
  \(c_{k-1} \cleval{\mathtt{bv}} c_{k}\) and
  \(c_{k-1}' \cleval{\mathtt{bv'}} c_{k}'\) are equal.
  Note that from \cref{eq:2} and \cref{def:concsym}, we have
  \(l_s = l = l'\), therefore the same instruction and expressions are
  evaluated in configurations \(c_{k-1}, c_{k-1}', \text{ and } s_{k-1}\).
  Suppose that \(c_{k-1}\) and \(c_{k-1}'\) produce distinct
  leakages. %
  This can happen during the evaluation of rules \rulename{load},
  \rulename{d\_jump}, \rulename{ite\_true}, \rulename{ite\_false},
  \rulename{store}.

  \myparagraph{Case \rulename{load}.} Concrete evaluation of an expression
  \(\load{e}\) in configurations \(c_{k-1}\) and \(c_{k-1}'\) produces
  leakages \(\mathtt{bv}\) and \(\mathtt{bv'}\) and, assuming the load
  is insecure, we have \(\mathtt{bv} \neq \mathtt{bv'}\).
  Symbolic evaluation evaluates the index to a symbolic expression
  \(\rel{\iota}\) and ensures \(\secleak(\rel{\iota}, \pc)\) holds. %
  From \cref{eq:2} and \cref{def:concsym}, we have \(M \sat \pc\),
  \(\mathtt{bv} = M(\lproj{\rel{\iota}})\) and
  \(\mathtt{bv'} = M(\rproj{\rel{\iota}})\).
  Because we assumed \(\mathtt{bv} \neq \mathtt{bv'}\), then
  \(M(\lproj{\rel{\iota}}) \neq M(\rproj{\rel{\iota}})\).  Therefore,
  we have
  \(M \sat \pc \wedge \lproj{\rel{\iota}} \neq \rproj{\rel{\iota}}\),
  meaning that \(\secleak(\rel{\iota}, \pc)\) evaluates to false and
  the symbolic execution is stuck, which is a contradiction.
  Therefore \(\mathtt{bv} = \mathtt{bv'}\).

  \myparagraph{Cases \rulename{d\_jump}, \rulename{ite\_true},
    \rulename{ite\_false}, \rulename{store}.}  The reasoning is
  analogous.

  \myparagraph{Conclusion.}
  We have shown that the hypothesis holds at step \(k\). If
  $s_0 \ieval{}^{k} s_{k}$, then for all model \(M\) and initial configurations
  \(\cconfvar_0 \concsym{l}{M} s_0 \) and \(\cconfvar'_0 \concsym{r}{M} s_0\)
  such that %
  \(\cconfvar_0 \cleval{\leakvar}^{k-1} \cconfvar_{k-1} \cleval{\leakvar_k}
  \cconfvar_{k}\) %
  and %
  \(\cconfvar_0' \cleval{\leakvar'}^{k-1} \cconfvar'_{k-1} \cleval{\leakvar_k'} \cconfvar'_{k}\)
  where \(\leakvar = \leakvar'\), %
  then \(\leakvar \cdot \leakvar_k = \leakvar' \cdot \leakvar_k'\).
\end{proof}

%% file: appendix.tex
\section{Definition of the \(\lookup\) function for relational indexes}\label{app:rel-lookup}
\begin{newercontentblock}
This section re-defines the \(\lookup\) function, introduced in \cref{sec:row}
for simple indexes, in order to consider relational indexes. Note that simple
indexes are enough as long as the leakage model that is considered subsumes the
memory obliviousness leakage model, and hence constraints memory indexes to be
equal in both executions---this is for instance the case for the constant-time
leakage model.

We recall that \(\compare(\iota,\kappa)\) is a comparison function based on
\emph{syntactic term equality}, which returns true (resp.\ false) only if
\(\iota\) and \(\kappa\) are equal (resp.\ different) in any interpretation, and
is undefined (denoted \(\bot\)) if the terms are not comparable.

We first define an auxiliary function \(\slookup(\smemvar_n, \iota)\) that
checks if the (non-relational) index \(\iota\) matches the \emph{most recent}
store operation in the (non-relational) memory \(\smemvar_n\). If it can
determine that the index of the store matches \(\iota\), it returns the
corresponding value; if it can determine that the index of the store is distinct
from \(\iota\), it is undefined; and if it cannot determine the (dis-)equality
of both indexes, or if \(\smemvar_{n}\) is the initial memory, it returns a
symbolic \(select\) operation from \(\smemvar_{n}\):
\begin{alignat*}{2}
  \slookup(\smemvar_0, \iota) &= select(\smemvar_0, \iota)\\
  \slookup(\smemvar_n, \iota) &=
  \begin{cases}
    \varphi & {\sf if~} \smemvar_n = store(\smemvar_{n-1},\kappa,\varphi) \wedge
    \compare(\iota,\kappa) \\
    \bot & {\sf if~} \smemvar_n = store(\smemvar_{n-1},\kappa,\varphi) \wedge
    \neg\compare(\iota,\kappa)\\
    select(\smemvar_n, \iota) & {\sf if} \smemvar_n = store(\smemvar_{n-1},\kappa,\varphi) \wedge \compare(\iota,\kappa) = \bot
 \end{cases}
\end{alignat*}
We also define a deep lookup function \(\dlookup(\smemvar_n, \iota)\) that
returns the value corresponding to the index \(\iota\) in the memory
\(\smemvar_n\). Contrary to \(\slookup\), \(\dlookup\) recursively iterates
through the memory when it can determine that the index of a store is distinct
from \(\iota\) (i.e., when \(\slookup\) returns \(\bot\)). This function
corresponds to the standard read-over-write~\cite{DBLP:conf/lpar/FarinierDBL18}
optimization:
\begin{alignat*}{2}
  \dlookup(\smemvar_n, \iota) &=
  \begin{cases}
    \slookup(\smemvar_n, \iota) & {\sf if~} \slookup(\smemvar_n, \iota) \neq \bot \\
    \dlookup(\smemvar_{n-1}, \iota) & {\sf if~} \slookup(\smemvar_n, \iota) = \bot \wedge \smemvar_n = store(\smemvar_{n-1},\kappa,\varphi)\\
 \end{cases}
\end{alignat*}

Next, we define a function \(\dedup\) that deduplicates a pair of expression
\(\pair{\varphi_{l}}{\varphi_{r}}\) if it can syntactically determine the
equality of \(\varphi_{l}\) and \(\varphi_{r}\):
\begin{alignat*}{2}
  \dedup(\pair{\varphi_{l}}{\varphi_{r}}) &=
  \begin{cases}
    \simple{\varphi_{l}} & {\sf if~} \compare(\varphi_l,\varphi_r) \\
    \pair{\varphi_{l}}{\varphi_{r}} & {\sf otherwise}
 \end{cases}
\end{alignat*}

Finally, the function \(\lookup\) for relational indexes is defined as follows:
\begin{mathpar}
  \inferrule*[left=match]{%
    \slookup(\proj{\smem_n}, \proj{\rel{\iota}}) = \varphi_{p} \\%
    \forall p \in \{l,r\}
  }{%
    \lookup(\smem_n, \rel{\iota}) = \dedup(\pair{\varphi_{l}}{\varphi_{r}})%
  }%
  \and
  \inferrule*[left=miss]{%
    \slookup(\proj{\smem_n}, \proj{\rel{\iota}}) = \bot \\%
    \forall p \in \{l,r\}
  }{%
    \lookup(\smem_n, \rel{\iota}) = \lookup(\smem_{n-1}, \rel{\iota})%
  }%
  \and
  \inferrule*[left=match-l]{%
    \slookup(\lproj{\smem_n}, \lproj{\rel{\iota}}) = \varphi_{l} \\%
    \slookup(\rproj{\smem_n}, \rproj{\rel{\iota}}) = \bot \\%
    \dlookup(\rproj{\smem_n}, \rproj{\rel{\iota}})
  }{%
    \lookup(\smem_n, \rel{\iota}) = \dedup(\pair{\varphi_{l}}{\varphi_{r}})%
  }%
\end{mathpar}
If both sides of the (relational) index \(\rel{\iota}\) match the most recent
store---or if the equality of the index with the most recent store cannot be
determined---then \(\lookup\) returns the corresponding deduplicated relational
expression (rule \textsc{match}). Notice than the \textsc{match} rule also
encompasses simple indexes as \(\proj{\simple{\iota}} = \iota\). In this case,
the \(\dedup\) function ensures that the returned value is simple.

If both sides of the current index \(\rel{\iota}\) are syntactically distinct
from the most recent store, then the store can be safely bypassed and \(lookup\)
is recursively called on the next store (rule \textsc{miss}). Notice that this
rule keeps the lookup for the right and left side \emph{synchronized}.

Finally, if the index \(\rel{\iota}\) matches the most recent store on the left
side but not on the right side, the rule continues the lookup only on the right
side using \(\dlookup\) (rule \textsc{match-l}). A symmetric rule
\textsc{match-r}, applies when the index matches the most recent store on the
right side.
\end{newercontentblock}

\section{Application: Violations Introduced by \texttt{clang} in a Constant-Time Sort Function}\label{app:ct-sort}
\begin{newercontentblock}
This section details vulnerabilities introduced by the \texttt{clang} compiler
in a constant-time sort function, (partly) given in \cref{lst:ct-sort}, where
\lstinline{in} contains secret data. The LLVM code, given in
\cref{lst:ct-sort_llvm}, contains two \lstinline{select} operations at
line~\ref{lst:ct-sort_llvm:select1} (resp.\
line~\ref{lst:ct-sort_llvm:select2}), which conditionally select the lowest
(resp.\ greatest) element of \lstinline{in} and store it to \lstinline{out[0]}
(resp. \lstinline{out[1]}). Notice that the source code and the LLVM code are
both constant-time: they do not contain secret-dependent control-flow and memory
accesses do not depend on secret data.

\input{./ressources/listings/ct-sort}

The compiled version of \cref{lst:ct-sort_llvm} for the \texttt{i386}
architecture---which does not feature \dbainline{cmov} instructions---is given
in \cref{lst:ct-sort_i386}. The \dbainline{select} LLVM instructions of
\cref{lst:ct-sort_llvm} are compiled to conditional jumps. Consequently, the
assembly code in \cref{lst:ct-sort_i386} contains two secret dependent jumps at
line~\ref{lst:ct-sort:branch1} and line~\ref{lst:ct-sort:branch2} and is
therefore insecure.

More interestingly, for the more recent architecture \texttt{i686} featuring
\dbainline{cmov}, the compiler introduces a secret dependent memory access. In
the assembly code for \texttt{i686}, given in \cref{lst:ct-sort_i686},
\lstinline{edx} is set at line~\ref{lst:ct-sort_llvm:cmov} to either the address
\lstinline{in+0} or \lstinline{in+4}, depending on whether %
\lstinline{in[0] <= in[1]} (i.e., the result of the \lstinline{cmp} instruction
at line~\ref{lst:ct-sort_llvm:cmp}). At line~\ref{lst:ct-sort_llvm:leak},
\lstinline{edx} is then used as a \dbainline{load} address. Consequently, the
address of the \lstinline{load} depends on the outcome of the condition %
\lstinline{in[0] <= in[1]}, which is secret, hence the assembly code in
\cref{lst:ct-sort_i686} is insecure.

These two vulnerabilities were automatically found by \brelse{} during our study
on preservation of constant-time by compilers detailed in \cref{sec:compilers}.
This example illustrates that reasoning about constant-time at source or LLVM
level in not sufficient.

\input{./ressources/listings/ct-sort-asm}

\end{newercontentblock}

\section{Scalability of \brelse{} according to input data
  size.}\label{app:scale_size}

\begin{newercontentblock}
In this section, we evaluate the scalability of \brelse{} according to the size
of the input data. %
We conduct this experiment on the subset of programs given in
\cref{tab:bounded-verif} that has variable input size. For these programs, we
report the median execution time for 5 runs of \brelse{} for a base size
(\(S_1\)) and for multiples of base size, ranging from $S_{2} = 2 \times S_1$ to
$S_{6} = 6 \times S_1$. For stream ciphers, we set \(S_1\) to 500; for block
ciphers we, we set \(S_1\) to the size of a block; and for the remove padding
function we set \(S_1\) to 64 in order to get size values around 256, which is
the maximum padding size. Experiments were performed on a laptop with an AMD
Ryzen 7 PRO 3700U @ 2.30GHz processor and 14GB of RAM\footnote{Note that the
  architecture used in this evaluation is different from the architecture used
  in \cref{sec:expes}, hence results are not drectly comparable to
  \cref{tab:scale_size}}. Results are reported in \cref{tab:scale_size}.
Additionally, \cref{fig:scale_size} shows the evolution of the execution time
according to the size.

\input{./ressources/tables/scale_size}

\input{./ressources/figures/scale_size}

\noindent\textbf{Results}. For most programs, the execution time is proportional to the size of
the input. Indeed for most programs, almost all queries are simplified and the
number of queries stays the same when increasing size of the input, meaning that
the variation of execution time is due to the symbolic execution engine (and not
to the solver). Consequently, the execution time is proportional to the number
of instruction explored, which is in turn proportional to the input size.

Three notable exceptions are \texttt{des\_ct},
\texttt{tls1\_cbc\_remove\_pad\_patch}, and \texttt{libsodium\_chacha20} for an
input size of \(S_5\). For \texttt{des\_ct}, the execution time does not grow
linearly with the input size because the number (and complexity) of queries sent
to the solver increases with the size of the input. For
\texttt{tls1\_cbc\_remove\_pad\_patch} the execution time is dominated by a
large number of queries, which are independent from the input size. The
execution time slowly increases and reaches a plateau after a size of 256, which
is the maximum padding size---at this point increasing the size of the input
does not change the symbolic exploration. Finally, for
\texttt{libsodium\_chacha20}, the execution time seems to increase linearly but
drops for the size \(S_5\) (i.e., 2500). Manual investigation revealed that the
size 2500 results in a final loop iteration on a 4 bytes chunk, which triggers a
shortcut (introduced by the compiler) and results in less queries sent to the
solver (0 instead of 3).

To conclude, the scalability of \brelse{} according to the input size mostly
depends on whether the number of queries depends on the input size. Fortunately,
in most of our cases, the number of queries does not vary with the input size
and the execution time of \brelse{} scales linearly with the size of the input.

\end{newercontentblock}

%% file: ressources/listings/ct-sort.tex
\begin{minipage}[b]{.45\linewidth}
  \begin{lstlisting}[style=nonumbers,caption={Constant-time sort: sort the two integers in \lstinline{in} (secret) and write the resulting array in \lstinline{out}.},label={lst:ct-sort}, morekeywords={out, in}]
void sort2(int out[2], int in[2]) {
  signed char c = (in[0] < in[1]) - 1;
  out[0] = (~c & in[0]) | (c & in[1]);
  out[1] = (~c & in[1]) | (c & in[0]);
}
\end{lstlisting}
\end{minipage}
\hfill
\begin{minipage}[b]{.45\linewidth}
\begin{lstlisting}[numbers=left,caption={Simplified llvm code of \Cref{lst:ct-sort}.},label={lst:ct-sort_llvm}, morekeywords={out, in}]
void sort2(i32* out, i32* in) {
  a0 = load in[0]
  a1 = load in[1]
  a = select (a0 < a1) a0 a1 <@\label{lst:ct-sort_llvm:select1}@>
  store a out[0]
  b1 = load in[1]
  b0 = load in[0]
  b = select (a0 < a1) b1 b0 <@\label{lst:ct-sort_llvm:select2}@>
  store b out[1] }
\end{lstlisting}
\end{minipage}

%% file: ressources/listings/ct-sort-asm.tex
\begin{minipage}[b]{0.45\linewidth}
\begin{lstlisting}[numbers=left,caption={\Cref{lst:ct-sort} compiled with \texttt{clang-9.0.1 -m32 -fno-stack-protector -O3 -march=i386.}},label={lst:ct-sort_i386}, morekeywords={out, in, c}]
sort2:
        ecx := in+4
        esi := load (in+0)
        edi := load (in+4)
        cmp esi edi
        jle .L0 <@\label{lst:ct-sort:branch1}@>
        esi := edi
.L0:    store (out+0) esi
        jl .L1 <@\label{lst:ct-sort:branch2}@>
        ecx := in+0
.L1:    ecx := load ecx
        store (out+4) ecx
\end{lstlisting}
\end{minipage}
\hfill
\begin{minipage}[b]{0.45\linewidth}
\begin{lstlisting}[numbers=left,caption={\Cref{lst:ct-sort} compiled with \texttt{clang-9.0.1 -m32 -fno-stack-protector -O3 -march=i686.}},label={lst:ct-sort_i686}, morekeywords={out, in, c}]
sort2:
        esi := load (in+0)
        edi := load (in+4)
        cmp esi edi <@\label{lst:ct-sort_llvm:cmp}@>
        edi := cmovle esi
        store (out+0) edi
        ecx := in+0
        edx := in+4
        edx := cmovge ecx <@\label{lst:ct-sort_llvm:cmov}@>
        ecx := load edx   <@\label{lst:ct-sort_llvm:leak}@>
        store (out+4) ecx
\end{lstlisting}
\end{minipage}

%% file: ressources/tables/scale_size.tex
\begin{table}[ht]
  \centering
  \begin{tabularx}{\textwidth}{lXrrrrrrr}
\toprule
\multicolumn{2}{l}{Program} &  \(S_{1}\) &  $T_1$ &   $T_2$ &   $T_3$ &    $T_4$ &    $T_5$ &    $T_6$ \\
\midrule
    \multirow{2}{*}{Hacl*-utility}
    & cmp-bytes \texttt{-O0} & 500 &   3.1 &   6.2 &   9.3 &   12.2 &   15.4 &   18.5 \\
    & cmp-bytes \texttt{-O3} & 500 &   1.2 &   2.4 &   3.7 &    4.9 &    6.1 &    7.4 \\
\midrule
    \multirow{3}{*}{Hacl*}
    & chacha20 & 500 &   2.7 &   5.3 &   7.8 &   10.4 &   13.0 &   15.3 \\
    & sha256 & 500 &   5.7 &  11.4 &  17.2 &   22.9 &   28.6 &   34.3 \\
    & sha512 & 500 &  11.2 &  18.2 &  27.4 &   36.5 &   45.0 &   54.7 \\
\midrule
    \multirow{4}{*}{Libsodium}
    & chacha20 & 500 &   7.5 &  16.3 &  25.5 &   34.3 &   12.9 &   53.0 \\
    & salsa20 & 500 &   2.5 &   4.9 &   7.3 &    9.7 &   12.0 &   14.2 \\
    & sha256 & 500 &   6.6 &  13.2 &  19.8 &   26.3 &   32.9 &   39.3 \\
    & sha512 & 500 &  13.2 &  21.6 &  32.6 &   43.4 &   53.7 &   64.9 \\
\midrule
    \multirow{2}{*}{BearSSL}
    & aes-ct-cbcenc &  16 &   0.3 &   0.6 &   0.9 &    1.2 &    1.4 &    1.7 \\
    & des-ct-cbcenc &   8 &   6.2 &  16.2 &  30.2 &   48.7 &   70.8 &   98.1 \\
\midrule
    \multirow{1}{*}{OpenSSL}
    & tls-remove-padding-patch &  64 & 790.5 & 846.7 & 944.9 & 1083.6 & 1084.0 & 1094.8 \\
\bottomrule
\end{tabularx}

\caption{Execution time according to input data size where \(S_{1}\) is the base
  size and \(T_{n}\) is the execution time of \brelse{} for input size
  \(S_{n}\).}\label{tab:scale_size}
\end{table}

%% file: ressources/figures/scale_size.tex
\begin{figure}[ht]
  \centering
  \includegraphics[scale=.8]{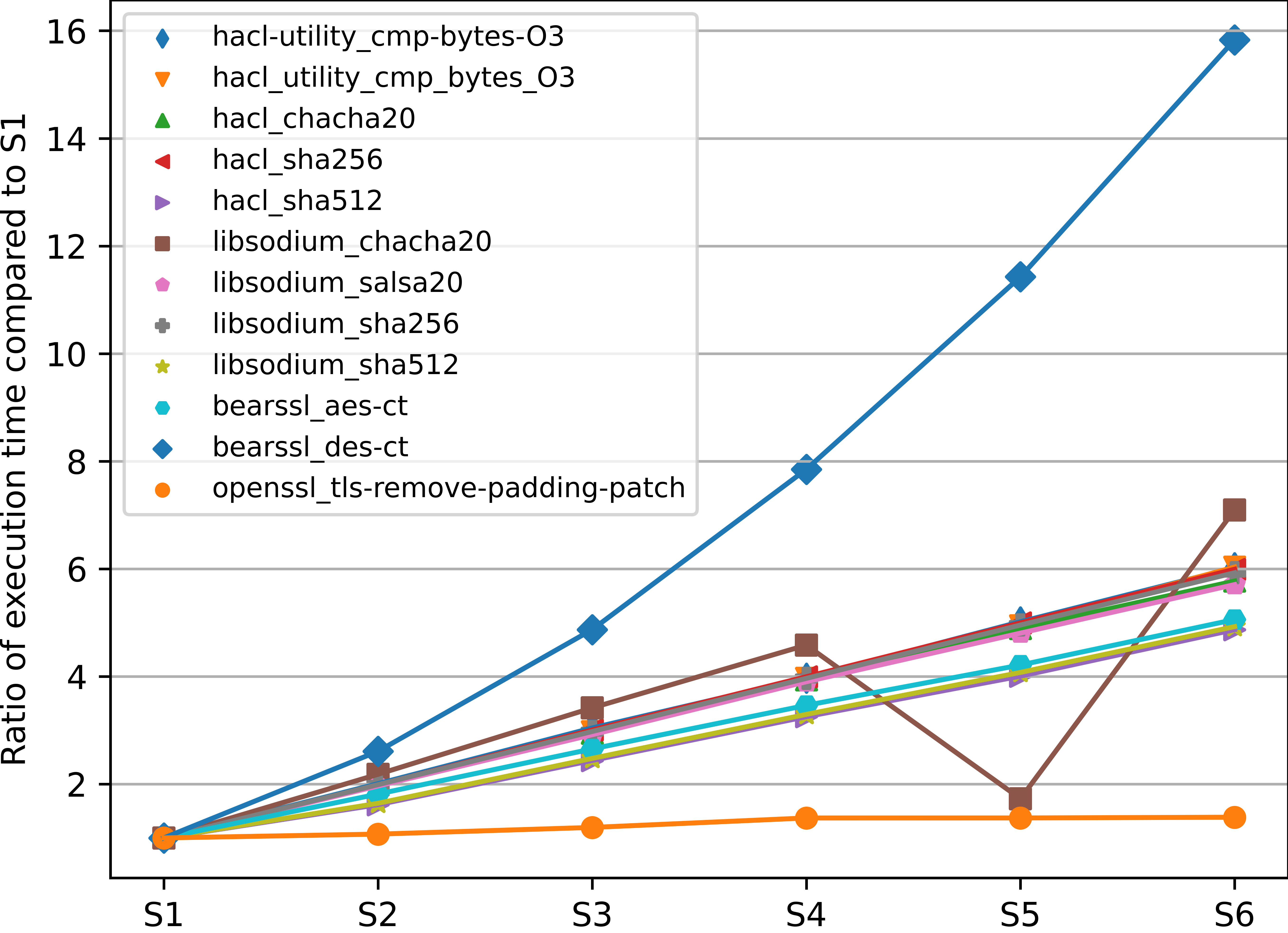}
  \caption{Evolution of the execution time according to the input size.}\label{fig:scale_size}
\end{figure}

%% file: ressources/SE_hypersafety-std.bib
@preamble{"\ifdefined\DeclarePrefChars\DeclarePrefChars{'’-}\else\fi"}

@online{ArraysExTheorySMTLIB,
    title = "{{ArraysEx Theory}}, {{SMT-LIB}}",
    url = "http://smtlib.cs.uiowa.edu/theories-ArraysEx.shtml",
    urldate = "2019-04-02"
}

@article{bacelaralmeidaFormalVerificationSidechannel2013,
    author = "Bacelar Almeida, J. and Barbosa, Manuel and Pinto, Jorge S. and Vieira, Bárbara",
    title = "Formal Verification of Side-Channel Countermeasures Using Self-Composition",
    date = "2013-07",
    journaltitle = "Science of Computer Programming",
    volume = "78",
    number = "7",
    pages = "796--812",
    issn = "01676423",
    doi = "10.1016/j.scico.2011.10.008",
    url = "https://linkinghub.elsevier.com/retrieve/pii/S0167642311001857",
    urldate = "2019-03-14",
    langid = "english"
}

@inproceedings{bardinBINCOAFrameworkBinary2011,
    author = "Bardin, Sébastien and Herrmann, Philippe and Leroux, Jérôme and Ly, Olivier and Tabary, Renaud and Vincent, Aymeric",
    title = "The {{BINCOA}} Framework for Binary Code Analysis",
    booktitle = "{{CAV}}",
    date = "2011",
    series = "Lecture Notes in Computer Science",
    volume = "6806",
    pages = "165--170",
    publisher = "{Springer}"
}

@report{barrettSMTLIBStandardVersion2017,
    author = "Barrett, Clark and Fontaine, Pascal and Tinelli, Cesare",
    title = "The {{SMT-LIB Standard}}: {{Version}} 2.6",
    date = "2017",
    institution = "{Department of Computer Science, The University of Iowa}"
}

@online{BearSSLConstantTimeCrypto,
    title = "{{BearSSL}} - {{Constant-Time Crypto}}",
    url = "https://bearssl.org/constanttime.html",
    urldate = "2019-05-07"
}

@online{CWE14CompilerRemoval,
    title = "{{CWE-14}}: {{Compiler Removal}} of {{Code}} to {{Clear Buffers}}",
    url = "https://cwe.mitre.org/data/definitions/14.html",
    urldate = "2020-09-29"
}

@inproceedings{DBLP:conf/acsac/WichelmannMES18,
    author = "Wichelmann, Jan and Moghimi, Ahmad and Eisenbarth, Thomas and Sunar, Berk",
    title = "MicroWalk: {A} Framework for Finding Side Channels in Binaries",
    booktitle = "Proceedings of the 34th Annual Computer Security Applications Conference, {ACSAC} 2018, San Juan, PR, USA, December 03-07, 2018",
    pages = "161--173",
    publisher = "{ACM}",
    year = "2018",
    url = "https://doi.org/10.1145/3274694.3274741",
    doi = "10.1145/3274694.3274741",
    bibsource = "dblp computer science bibliography, https://dblp.org"
}

@inproceedings{DBLP:conf/cans/KaufmannPVV16,
    author = "Kaufmann, Thierry and Pelletier, Herv{\'{e}} and Vaudenay, Serge and Villegas, Karine",
    editor = "Foresti, Sara and Persiano, Giuseppe",
    title = "When Constant-Time Source Yields Variable-Time Binary: Exploiting Curve25519-donna Built with {MSVC} 2015",
    booktitle = "Cryptology and Network Security - 15th International Conference, {CANS} 2016, Milan, Italy, November 14-16, 2016, Proceedings",
    series = "Lecture Notes in Computer Science",
    volume = "10052",
    pages = "573--582",
    year = "2016",
    url = "https://doi.org/10.1007/978-3-319-48965-0_36",
    doi = "10.1007/978-3-319-48965-0_36",
    bibsource = "dblp computer science bibliography, https://dblp.org"
}

@inproceedings{DBLP:conf/cav/BrumleyJAS11,
    author = "Brumley, David and Jager, Ivan and Avgerinos, Thanassis and Schwartz, Edward J.",
    editor = "Gopalakrishnan, Ganesh and Qadeer, Shaz",
    title = "{BAP:} {A} Binary Analysis Platform",
    booktitle = "Computer Aided Verification - 23rd International Conference, {CAV} 2011, Snowbird, UT, USA, July 14-20, 2011. Proceedings",
    series = "Lecture Notes in Computer Science",
    volume = "6806",
    pages = "463--469",
    publisher = "Springer",
    year = "2011",
    url = "https://doi.org/10.1007/978-3-642-22110-1_37",
    doi = "10.1007/978-3-642-22110-1_37",
    bibsource = "dblp computer science bibliography, https://dblp.org"
}

@inproceedings{DBLP:conf/cav/FarinierBBP18,
    author = "Farinier, Benjamin and Bardin, S{\'{e}}bastien and Bonichon, Richard and Potet, Marie{-}Laure",
    editor = "Chockler, Hana and Weissenbacher, Georg",
    title = "Model Generation for Quantified Formulas: {A} Taint-Based Approach",
    booktitle = "Computer Aided Verification - 30th International Conference, {CAV} 2018, Held as Part of the Federated Logic Conference, FloC 2018, Oxford, UK, July 14-17, 2018, Proceedings, Part {II}",
    series = "Lecture Notes in Computer Science",
    volume = "10982",
    pages = "294--313",
    publisher = "Springer",
    year = "2018",
    url = "https://doi.org/10.1007/978-3-319-96142-2_19",
    doi = "10.1007/978-3-319-96142-2_19",
    bibsource = "dblp computer science bibliography, https://dblp.org"
}

@inproceedings{DBLP:conf/cav/KopfMO12,
    author = {K{\"{o}}pf, Boris and Mauborgne, Laurent and Ochoa, Mart{\'{i}}n},
    editor = "Madhusudan, P. and Seshia, Sanjit A.",
    title = "Automatic Quantification of Cache Side-Channels",
    booktitle = "Computer Aided Verification - 24th International Conference, {CAV} 2012, Berkeley, CA, USA, July 7-13, 2012 Proceedings",
    series = "Lecture Notes in Computer Science",
    volume = "7358",
    pages = "564--580",
    publisher = "Springer",
    year = "2012",
    url = "https://doi.org/10.1007/978-3-642-31424-7_40",
    doi = "10.1007/978-3-642-31424-7_40",
    bibsource = "dblp computer science bibliography, https://dblp.org"
}

@inproceedings{DBLP:conf/cc/RodriguesPA16,
    author = "Rodrigues, Bruno and Pereira, Fernando Magno Quint{\\textasciitilde {a}}o and Aranha, Diego F.",
    editor = "Zaks, Ayal and Hermenegildo, Manuel V.",
    title = "Sparse representation of implicit flows with applications to side-channel detection",
    booktitle = "Proceedings of the 25th International Conference on Compiler Construction, {CC} 2016, Barcelona, Spain, March 12-18, 2016",
    pages = "110--120",
    publisher = "{ACM}",
    year = "2016",
    url = "https://doi.org/10.1145/2892208.2892230",
    doi = "10.1145/2892208.2892230",
    bibsource = "dblp computer science bibliography, https://dblp.org"
}

@inproceedings{DBLP:conf/ccs/AlmeidaBBBGLOPS17,
    author = "Almeida, Jos{\'{e}} Bacelar and Barbosa, Manuel and Barthe, Gilles and Blot, Arthur and Gr{\'{e}}goire, Benjamin and Laporte, Vincent and Oliveira, Tiago and Pacheco, Hugo and Schmidt, Benedikt and Strub, Pierre{-}Yves",
    editor = "Thuraisingham, Bhavani and Evans, David and Malkin, Tal and Xu, Dongyan",
    title = "Jasmin: High-Assurance and High-Speed Cryptography",
    booktitle = "Proceedings of the 2017 {ACM} {SIGSAC} Conference on Computer and Communications Security, {CCS} 2017, Dallas, TX, USA, October 30 - November 03, 2017",
    pages = "1807--1823",
    publisher = "{ACM}",
    year = "2017",
    url = "https://doi.org/10.1145/3133956.3134078",
    doi = "10.1145/3133956.3134078",
    bibsource = "dblp computer science bibliography, https://dblp.org"
}

@inproceedings{DBLP:conf/ccs/BalliuDG14,
    author = "Balliu, Musard and Dam, Mads and Guanciale, Roberto",
    editor = "Ahn, Gail{-}Joon and Yung, Moti and Li, Ninghui",
    title = "Automating Information Flow Analysis of Low Level Code",
    booktitle = "Proceedings of the 2014 {ACM} {SIGSAC} Conference on Computer and Communications Security, Scottsdale, AZ, USA, November 3-7, 2014",
    pages = "1080--1091",
    publisher = "{ACM}",
    year = "2014",
    url = "https://doi.org/10.1145/2660267.2660322",
    doi = "10.1145/2660267.2660322",
    bibsource = "dblp computer science bibliography, https://dblp.org"
}

@inproceedings{DBLP:conf/ccs/BartheBCLP14,
    author = "Barthe, Gilles and Betarte, Gustavo and Campo, Juan Diego and Luna, Carlos Daniel and Pichardie, David",
    editor = "Ahn, Gail{-}Joon and Yung, Moti and Li, Ninghui",
    title = "System-level Non-interference for Constant-time Cryptography",
    booktitle = "Proceedings of the 2014 {ACM} {SIGSAC} Conference on Computer and Communications Security, Scottsdale, AZ, USA, November 3-7, 2014",
    pages = "1267--1279",
    publisher = "{ACM}",
    year = "2014",
    url = "https://doi.org/10.1145/2660267.2660283",
    doi = "10.1145/2660267.2660283",
    bibsource = "dblp computer science bibliography, https://dblp.org"
}

@inproceedings{DBLP:conf/ccs/BorrelloDQG21,
    author = "Borrello, Pietro and D'Elia, Daniele Cono and Querzoni, Leonardo and Giuffrida, Cristiano",
    editor = "Kim, Yongdae and Kim, Jong and Vigna, Giovanni and Shi, Elaine",
    title = "Constantine: Automatic Side-Channel Resistance Using Efficient Control and Data Flow Linearization",
    booktitle = "{CCS} '21: 2021 {ACM} {SIGSAC} Conference on Computer and Communications Security, Virtual Event, Republic of Korea, November 15 - 19, 2021",
    pages = "715--733",
    publisher = "{ACM}",
    year = "2021",
    url = "https://doi.org/10.1145/3460120.3484583",
    doi = "10.1145/3460120.3484583",
    bibsource = "dblp computer science bibliography, https://dblp.org"
}

@inproceedings{DBLP:conf/ccs/ZinzindohoueBPB17,
    author = "Zinzindohou{\'{e}}, Jean Karim and Bhargavan, Karthikeyan and Protzenko, Jonathan and Beurdouche, Benjamin",
    editor = "Thuraisingham, Bhavani and Evans, David and Malkin, Tal and Xu, Dongyan",
    title = "HACL*: {A} Verified Modern Cryptographic Library",
    booktitle = "Proceedings of the 2017 {ACM} {SIGSAC} Conference on Computer and Communications Security, {CCS} 2017, Dallas, TX, USA, October 30 - November 03, 2017",
    pages = "1789--1806",
    publisher = "{ACM}",
    year = "2017",
    url = "https://doi.org/10.1145/3133956.3134043",
    doi = "10.1145/3133956.3134043",
    bibsource = "dblp computer science bibliography, https://dblp.org"
}

@inproceedings{DBLP:conf/cgo/SoaresP21,
    author = "Soares, Luigi and Pereira, Fernando Magno Quint{\\textasciitilde {a}}o",
    editor = "Lee, Jae W. and Soffa, Mary Lou and Zaks, Ayal",
    title = "Memory-Safe Elimination of Side Channels",
    booktitle = "{IEEE/ACM} International Symposium on Code Generation and Optimization, {CGO} 2021, Seoul, South Korea, February 27 - March 3, 2021",
    pages = "200--210",
    publisher = "{IEEE}",
    year = "2021",
    url = "https://doi.org/10.1109/CGO51591.2021.9370305",
    doi = "10.1109/CGO51591.2021.9370305",
    bibsource = "dblp computer science bibliography, https://dblp.org"
}

@inproceedings{DBLP:conf/csfw/AskarovMDC15,
    author = "Askarov, Aslan and Moore, Scott and Dimoulas, Christos and Chong, Stephen",
    editor = "Fournet, C{\'{e}}dric and Hicks, Michael W. and Vigan{\`{o}}, Luca",
    title = "Cryptographic Enforcement of Language-Based Information Erasure",
    booktitle = "{IEEE} 28th Computer Security Foundations Symposium, {CSF} 2015, Verona, Italy, 13-17 July, 2015",
    pages = "334--348",
    publisher = "{IEEE} Computer Society",
    year = "2015",
    url = "https://doi.org/10.1109/CSF.2015.30",
    doi = "10.1109/CSF.2015.30",
    bibsource = "dblp computer science bibliography, https://dblp.org"
}

@inproceedings{DBLP:conf/csfw/BalliuDG12,
    author = "Balliu, Musard and Dam, Mads and Guernic, Gurvan Le",
    editor = "Chong, Stephen",
    title = "ENCoVer: Symbolic Exploration for Information Flow Security",
    booktitle = "25th {IEEE} Computer Security Foundations Symposium, {CSF} 2012, Cambridge, MA, USA, June 25-27, 2012",
    pages = "30--44",
    publisher = "{IEEE} Computer Society",
    year = "2012",
    url = "https://doi.org/10.1109/CSF.2012.24",
    doi = "10.1109/CSF.2012.24",
    bibsource = "dblp computer science bibliography, https://dblp.org"
}

@inproceedings{DBLP:conf/csfw/BartheDR04,
    author = "Barthe, Gilles and D'Argenio, Pedro R. and Rezk, Tamara",
    title = "Secure Information Flow by Self-Composition",
    booktitle = "17th {IEEE} Computer Security Foundations Workshop, {(CSFW-17} 2004), 28-30 June 2004, Pacific Grove, CA, {USA}",
    pages = "100--114",
    publisher = "{IEEE} Computer Society",
    year = "2004",
    url = "http://doi.ieeecomputersociety.org/10.1109/CSFW.2004.17",
    doi = "10.1109/CSFW.2004.17",
    bibsource = "dblp computer science bibliography, https://dblp.org"
}

@inproceedings{DBLP:conf/csfw/BartheGL18,
    author = "Barthe, Gilles and Gr{\'{e}}goire, Benjamin and Laporte, Vincent",
    title = {Secure Compilation of Side-Channel Countermeasures: The Case of Cryptographic "Constant-Time"},
    booktitle = "31st {IEEE} Computer Security Foundations Symposium, {CSF} 2018, Oxford, United Kingdom, July 9-12, 2018",
    pages = "328--343",
    publisher = "{IEEE} Computer Society",
    year = "2018",
    url = "https://doi.org/10.1109/CSF.2018.00031",
    doi = "10.1109/CSF.2018.00031",
    bibsource = "dblp computer science bibliography, https://dblp.org"
}

@inproceedings{DBLP:conf/csfw/BessonDJ19,
    author = "Besson, Fr{\'{e}}d{\'{e}}ric and Dang, Alexandre and Jensen, Thomas P.",
    title = "Information-Flow Preservation in Compiler Optimisations",
    booktitle = "32nd {IEEE} Computer Security Foundations Symposium, {CSF} 2019, Hoboken, NJ, USA, June 25-28, 2019",
    pages = "230--242",
    publisher = "{IEEE}",
    year = "2019",
    url = "https://doi.org/10.1109/CSF.2019.00023",
    doi = "10.1109/CSF.2019.00023",
    bibsource = "dblp computer science bibliography, https://dblp.org"
}

@inproceedings{DBLP:conf/csfw/ChongM05,
    author = "Chong, Stephen and Myers, Andrew C.",
    title = "Language-Based Information Erasure",
    booktitle = "18th {IEEE} Computer Security Foundations Workshop, {(CSFW-18} 2005), 20-22 June 2005, Aix-en-Provence, France",
    pages = "241--254",
    publisher = "{IEEE} Computer Society",
    year = "2005",
    url = "https://doi.org/10.1109/CSFW.2005.19",
    doi = "10.1109/CSFW.2005.19",
    bibsource = "dblp computer science bibliography, https://dblp.org"
}

@inproceedings{DBLP:conf/csfw/ChongM08,
    author = "Chong, Stephen and Myers, Andrew C.",
    title = "End-to-End Enforcement of Erasure and Declassification",
    booktitle = "Proceedings of the 21st {IEEE} Computer Security Foundations Symposium, {CSF} 2008, Pittsburgh, Pennsylvania, USA, 23-25 June 2008",
    pages = "98--111",
    publisher = "{IEEE} Computer Society",
    year = "2008",
    url = "https://doi.org/10.1109/CSF.2008.12",
    doi = "10.1109/CSF.2008.12",
    bibsource = "dblp computer science bibliography, https://dblp.org"
}

@inproceedings{DBLP:conf/csfw/ClarksonS08,
    author = "Clarkson, Michael R. and Schneider, Fred B.",
    title = "Hyperproperties",
    booktitle = "Proceedings of the 21st {IEEE} Computer Security Foundations Symposium, {CSF} 2008, Pittsburgh, Pennsylvania, USA, 23-25 June 2008",
    pages = "51--65",
    publisher = "{IEEE} Computer Society",
    year = "2008",
    url = "https://doi.org/10.1109/CSF.2008.7",
    doi = "10.1109/CSF.2008.7",
    bibsource = "dblp computer science bibliography, https://dblp.org"
}

@inproceedings{DBLP:conf/date/ReparazBV17,
    author = "Reparaz, Oscar and Balasch, Josep and Verbauwhede, Ingrid",
    editor = "Atienza, David and Natale, Giorgio Di",
    title = "Dude, is my code constant time?",
    booktitle = "Design, Automation {\&} Test in Europe Conference {\&} Exhibition, {DATE} 2017, Lausanne, Switzerland, March 27-31, 2017",
    pages = "1697--1702",
    publisher = "{IEEE}",
    year = "2017",
    url = "https://doi.org/10.23919/DATE.2017.7927267",
    doi = "10.23919/DATE.2017.7927267",
    bibsource = "dblp computer science bibliography, https://dblp.org"
}

@inproceedings{DBLP:conf/date/SubramanyanMKMF16,
    author = "Subramanyan, Pramod and Malik, Sharad and Khattri, Hareesh and Maiti, Abhranil and Fung, Jason M.",
    editor = {Fanucci, Luca and Teich, J{\"{u}}rgen},
    title = "Verifying information flow properties of firmware using symbolic execution",
    booktitle = "2016 Design, Automation {\&} Test in Europe Conference {\&} Exhibition, {DATE} 2016, Dresden, Germany, March 14-18, 2016",
    pages = "337--342",
    publisher = "{IEEE}",
    year = "2016",
    url = "https://ieeexplore.ieee.org/document/7459333/",
    bibsource = "dblp computer science bibliography, https://dblp.org"
}

@inproceedings{DBLP:conf/dimva/SalwanBP18,
    author = "Salwan, Jonathan and Bardin, S{\'{e}}bastien and Potet, Marie{-}Laure",
    editor = "Giuffrida, Cristiano and Bardin, S{\'{e}}bastien and Blanc, Gregory",
    title = "Symbolic Deobfuscation: From Virtualized Code Back to the Original",
    booktitle = "Detection of Intrusions and Malware, and Vulnerability Assessment - 15th International Conference, {DIMVA} 2018, Saclay, France, June 28-29, 2018, Proceedings",
    series = "Lecture Notes in Computer Science",
    volume = "10885",
    pages = "372--392",
    publisher = "Springer",
    year = "2018",
    url = "https://doi.org/10.1007/978-3-319-93411-2_17",
    doi = "10.1007/978-3-319-93411-2_17",
    bibsource = "dblp computer science bibliography, https://dblp.org"
}

@inproceedings{DBLP:conf/esop/CousotCFMMMR05,
    author = "Cousot, Patrick and Cousot, Radhia and Feret, J{\'{e}}r{\^{o}}me and Mauborgne, Laurent and Min{\'{e}}, Antoine and Monniaux, David and Rival, Xavier",
    editor = "Sagiv, Shmuel",
    title = "The ASTRE{\'{E}} Analyzer",
    booktitle = "Programming Languages and Systems, 14th European Symposium on Programming,ESOP 2005, Held as Part of the Joint European Conferences on Theory and Practice of Software, {ETAPS} 2005, Edinburgh, UK, April 4-8, 2005, Proceedings",
    series = "Lecture Notes in Computer Science",
    volume = "3444",
    pages = "21--30",
    publisher = "Springer",
    year = "2005",
    url = "https://doi.org/10.1007/978-3-540-31987-0_3",
    doi = "10.1007/978-3-540-31987-0_3",
    bibsource = "dblp computer science bibliography, https://dblp.org"
}

@inproceedings{DBLP:conf/esop/HuntS08,
    author = "Hunt, Sebastian and Sands, David",
    editor = "Drossopoulou, Sophia",
    title = "Just Forget It - The Semantics and Enforcement of Information Erasure",
    booktitle = "Programming Languages and Systems, 17th European Symposium on Programming, {ESOP} 2008, Held as Part of the Joint European Conferences on Theory and Practice of Software, {ETAPS} 2008, Budapest, Hungary, March 29-April 6, 2008. Proceedings",
    series = "Lecture Notes in Computer Science",
    volume = "4960",
    pages = "239--253",
    publisher = "Springer",
    year = "2008",
    url = "https://doi.org/10.1007/978-3-540-78739-6_19",
    doi = "10.1007/978-3-540-78739-6_19",
    bibsource = "dblp computer science bibliography, https://dblp.org"
}

@inproceedings{DBLP:conf/esorics/BlazyPT17,
    author = "Blazy, Sandrine and Pichardie, David and Trieu, Alix",
    editor = "Foley, Simon N. and Gollmann, Dieter and Snekkenes, Einar",
    title = "Verifying Constant-Time Implementations by Abstract Interpretation",
    booktitle = "Computer Security - {ESORICS} 2017 - 22nd European Symposium on Research in Computer Security, Oslo, Norway, September 11-15, 2017, Proceedings, Part {I}",
    series = "Lecture Notes in Computer Science",
    volume = "10492",
    pages = "260--277",
    publisher = "Springer",
    year = "2017",
    url = "https://doi.org/10.1007/978-3-319-66402-6_16",
    doi = "10.1007/978-3-319-66402-6_16",
    bibsource = "dblp computer science bibliography, https://dblp.org"
}

@inproceedings{DBLP:conf/eurosp/SimonCA18,
    author = "Simon, Laurent and Chisnall, David and Anderson, Ross J.",
    title = "What You Get is What You {C:} Controlling Side Effects in Mainstream {C} Compilers",
    booktitle = "2018 {IEEE} European Symposium on Security and Privacy, EuroS{\&}P 2018, London, United Kingdom, April 24-26, 2018",
    pages = "1--15",
    publisher = "{IEEE}",
    year = "2018",
    url = "https://doi.org/10.1109/EuroSP.2018.00009",
    doi = "10.1109/EuroSP.2018.00009",
    bibsource = "dblp computer science bibliography, https://dblp.org"
}

@inproceedings{DBLP:conf/fm/BartheCK11,
    author = "Barthe, Gilles and Crespo, Juan Manuel and Kunz, C{\'{e}}sar",
    editor = "Butler, Michael J. and Schulte, Wolfram",
    title = "Relational Verification Using Product Programs",
    booktitle = "{FM} 2011: Formal Methods - 17th International Symposium on Formal Methods, Limerick, Ireland, June 20-24, 2011. Proceedings",
    series = "Lecture Notes in Computer Science",
    volume = "6664",
    pages = "200--214",
    publisher = "Springer",
    year = "2011",
    url = "https://doi.org/10.1007/978-3-642-21437-0_17",
    doi = "10.1007/978-3-642-21437-0_17",
    bibsource = "dblp computer science bibliography, https://dblp.org"
}

@inproceedings{DBLP:conf/fm/DjoudiBG16,
    author = "Djoudi, Adel and Bardin, S{\'{e}}bastien and Goubault, {\'{E}}ric",
    editor = "Fitzgerald, John S. and Heitmeyer, Constance L. and Gnesi, Stefania and Philippou, Anna",
    title = "Recovering High-Level Conditions from Binary Programs",
    booktitle = "{FM} 2016: Formal Methods - 21st International Symposium, Limassol, Cyprus, November 9-11, 2016, Proceedings",
    series = "Lecture Notes in Computer Science",
    volume = "9995",
    pages = "235--253",
    year = "2016",
    url = "https://doi.org/10.1007/978-3-319-48989-6_15",
    doi = "10.1007/978-3-319-48989-6_15",
    bibsource = "dblp computer science bibliography, https://dblp.org"
}

@inproceedings{DBLP:conf/forte/MilushevBC12,
    author = "Milushev, Dimiter and Beck, Wim and Clarke, Dave",
    editor = "Giese, Holger and Rosu, Grigore",
    title = "Noninterference via Symbolic Execution",
    booktitle = "Formal Techniques for Distributed Systems - Joint 14th {IFIP} {WG} 6.1 International Conference, {FMOODS} 2012 and 32nd {IFIP} {WG} 6.1 International Conference, {FORTE} 2012, Stockholm, Sweden, June 13-16, 2012. Proceedings",
    series = "Lecture Notes in Computer Science",
    volume = "7273",
    pages = "152--168",
    publisher = "Springer",
    year = "2012",
    url = "https://doi.org/10.1007/978-3-642-30793-5_10",
    doi = "10.1007/978-3-642-30793-5_10",
    bibsource = "dblp computer science bibliography, https://dblp.org"
}

@inproceedings{DBLP:conf/fse/WheelerN94,
    author = "Wheeler, David J. and Needham, Roger M.",
    editor = "Preneel, Bart",
    title = "TEA, a Tiny Encryption Algorithm",
    booktitle = "Fast Software Encryption: Second International Workshop. Leuven, Belgium, 14-16 December 1994, Proceedings",
    series = "Lecture Notes in Computer Science",
    volume = "1008",
    pages = "363--366",
    publisher = "Springer",
    year = "1994",
    url = "https://doi.org/10.1007/3-540-60590-8_29",
    doi = "10.1007/3-540-60590-8_29",
    bibsource = "dblp computer science bibliography, https://dblp.org"
}

@inproceedings{DBLP:conf/iccsw/Phan13,
    author = "Phan, Quoc{-}Sang",
    editor = "Jones, Andrew V. and Ng, Nicholas",
    title = "Self-composition by Symbolic Execution",
    booktitle = "2013 Imperial College Computing Student Workshop, {ICCSW} 2013, September 26/27, 2013, London, United Kingdom",
    series = "OASIcs",
    volume = "35",
    pages = "95--102",
    publisher = "Schloss Dagstuhl - Leibniz-Zentrum fuer Informatik, Germany",
    year = "2013",
    url = "https://doi.org/10.4230/OASIcs.ICCSW.2013.95",
    doi = "10.4230/OASIcs.ICCSW.2013.95",
    bibsource = "dblp computer science bibliography, https://dblp.org"
}

@inproceedings{DBLP:conf/icisc/MolnarPSW05,
    author = "Molnar, David and Piotrowski, Matt and Schultz, David and Wagner, David A.",
    editor = "Won, Dongho and Kim, Seungjoo",
    title = "The Program Counter Security Model: Automatic Detection and Removal of Control-Flow Side Channel Attacks",
    booktitle = "Information Security and Cryptology - {ICISC} 2005, 8th International Conference, Seoul, Korea, December 1-2, 2005, Revised Selected Papers",
    series = "Lecture Notes in Computer Science",
    volume = "3935",
    pages = "156--168",
    publisher = "Springer",
    year = "2005",
    url = "https://doi.org/10.1007/11734727_14",
    doi = "10.1007/11734727_14",
    bibsource = "dblp computer science bibliography, https://dblp.org"
}

@inproceedings{DBLP:conf/iciss/TedescoHS11,
    author = "Tedesco, Filippo Del and Hunt, Sebastian and Sands, David",
    editor = "Jajodia, Sushil and Mazumdar, Chandan",
    title = "A Semantic Hierarchy for Erasure Policies",
    booktitle = "Information Systems Security - 7th International Conference, {ICISS} 2011, Kolkata, India, December 15-19, 2011, Procedings",
    series = "Lecture Notes in Computer Science",
    volume = "7093",
    pages = "352--369",
    publisher = "Springer",
    year = "2011",
    url = "https://doi.org/10.1007/978-3-642-25560-1_24",
    doi = "10.1007/978-3-642-25560-1_24",
    bibsource = "dblp computer science bibliography, https://dblp.org"
}

@inproceedings{DBLP:conf/icse/BounimovaGM13,
    author = "Bounimova, Ella and Godefroid, Patrice and Molnar, David A.",
    editor = "Notkin, David and Cheng, Betty H. C. and Pohl, Klaus",
    title = "Billions and billions of constraints: whitebox fuzz testing in production",
    booktitle = "35th International Conference on Software Engineering, {ICSE} '13, San Francisco, CA, USA, May 18-26, 2013",
    pages = "122--131",
    publisher = "{IEEE} Computer Society",
    year = "2013",
    url = "https://doi.org/10.1109/ICSE.2013.6606558",
    doi = "10.1109/ICSE.2013.6606558",
    bibsource = "dblp computer science bibliography, https://dblp.org"
}

@inproceedings{DBLP:conf/icse/CadarP14,
    author = "Cadar, Cristian and Palikareva, Hristina",
    editor = "Jalote, Pankaj and Briand, Lionel C. and van der Hoek, Andr{\'{e}}",
    title = "Shadow symbolic execution for better testing of evolving software",
    booktitle = "36th International Conference on Software Engineering, {ICSE} '14, Companion Proceedings, Hyderabad, India, May 31 - June 07, 2014",
    pages = "432--435",
    publisher = "{ACM}",
    year = "2014",
    url = "https://doi.org/10.1145/2591062.2591104",
    doi = "10.1145/2591062.2591104",
    bibsource = "dblp computer science bibliography, https://dblp.org"
}

@inproceedings{DBLP:conf/icse/PalikarevaKC16,
    author = "Palikareva, Hristina and Kuchta, Tomasz and Cadar, Cristian",
    editor = "Dillon, Laura K. and Visser, Willem and Williams, Laurie A.",
    title = "Shadow of a doubt: testing for divergences between software versions",
    booktitle = "Proceedings of the 38th International Conference on Software Engineering, {ICSE} 2016, Austin, TX, USA, May 14-22, 2016",
    pages = "1181--1192",
    publisher = "{ACM}",
    year = "2016",
    url = "https://doi.org/10.1145/2884781.2884845",
    doi = "10.1145/2884781.2884845",
    bibsource = "dblp computer science bibliography, https://dblp.org"
}

@inproceedings{DBLP:conf/icst/HeEC20,
    author = "He, Shaobo and Emmi, Michael and Ciocarlie, Gabriela F.",
    title = "ct-fuzz: Fuzzing for Timing Leaks",
    booktitle = "13th {IEEE} International Conference on Software Testing, Validation and Verification, {ICST} 2020, Porto, Portugal, October 24-28, 2020",
    pages = "466--471",
    publisher = "{IEEE}",
    year = "2020",
    url = "https://doi.org/10.1109/ICST46399.2020.00063",
    doi = "10.1109/ICST46399.2020.00063",
    bibsource = "dblp computer science bibliography, https://dblp.org"
}

@inproceedings{DBLP:conf/issta/DavidBFMPTM16,
    author = "David, Robin and Bardin, S{\'{e}}bastien and Feist, Josselin and Mounier, Laurent and Potet, Marie{-}Laure and Ta, Thanh Dinh and Marion, Jean{-}Yves",
    editor = "Zeller, Andreas and Roychoudhury, Abhik",
    title = "Specification of concretization and symbolization policies in symbolic execution",
    booktitle = {Proceedings of the 25th International Symposium on Software Testing and Analysis, {ISSTA} 2016, Saarbr{\"{u}}cken, Germany, July 18-20, 2016},
    pages = "36--46",
    publisher = "{ACM}",
    year = "2016",
    url = "https://doi.org/10.1145/2931037.2931048",
    doi = "10.1145/2931037.2931048",
    bibsource = "dblp computer science bibliography, https://dblp.org"
}

@inproceedings{DBLP:conf/issta/WuGS018,
    author = "Wu, Meng and Guo, Shengjian and Schaumont, Patrick and Wang, Chao",
    editor = "Tip, Frank and Bodden, Eric",
    title = "Eliminating timing side-channel leaks using program repair",
    booktitle = "Proceedings of the 27th {ACM} {SIGSOFT} International Symposium on Software Testing and Analysis, {ISSTA} 2018, Amsterdam, The Netherlands, July 16-21, 2018",
    pages = "15--26",
    publisher = "{ACM}",
    year = "2018",
    url = "https://doi.org/10.1145/3213846.3213851",
    doi = "10.1145/3213846.3213851",
    bibsource = "dblp computer science bibliography, https://dblp.org"
}

@inproceedings{DBLP:conf/kbse/KimFJJOLC17,
    author = "Kim, Soomin and Faerevaag, Markus and Jung, Minkyu and Jung, Seungil and Oh, DongYeop and Lee, JongHyup and Cha, Sang Kil",
    editor = "Rosu, Grigore and Penta, Massimiliano Di and Nguyen, Tien N.",
    title = "Testing intermediate representations for binary analysis",
    booktitle = "Proceedings of the 32nd {IEEE/ACM} International Conference on Automated Software Engineering, {ASE} 2017, Urbana, IL, USA, October 30 - November 03, 2017",
    pages = "353--364",
    publisher = "{IEEE} Computer Society",
    year = "2017",
    url = "https://doi.org/10.1109/ASE.2017.8115648",
    doi = "10.1109/ASE.2017.8115648",
    bibsource = "dblp computer science bibliography, https://dblp.org"
}

@inproceedings{DBLP:conf/kbse/RecoulesBBMP19,
    author = "Recoules, Fr{\'{e}}d{\'{e}}ric and Bardin, S{\'{e}}bastien and Bonichon, Richard and Mounier, Laurent and Potet, Marie{-}Laure",
    title = "Get Rid of Inline Assembly through Verification-Oriented Lifting",
    booktitle = "34th {IEEE/ACM} International Conference on Automated Software Engineering, {ASE} 2019, San Diego, CA, USA, November 11-15, 2019",
    pages = "577--589",
    publisher = "{IEEE}",
    year = "2019",
    url = "https://doi.org/10.1109/ASE.2019.00060",
    doi = "10.1109/ASE.2019.00060",
    bibsource = "dblp computer science bibliography, https://dblp.org"
}

@inproceedings{DBLP:conf/latincrypt/BernsteinLS12,
    author = "Bernstein, Daniel J. and Lange, Tanja and Schwabe, Peter",
    editor = "Hevia, Alejandro and Neven, Gregory",
    title = "The Security Impact of a New Cryptographic Library",
    booktitle = "Progress in Cryptology - {LATINCRYPT} 2012 - 2nd International Conference on Cryptology and Information Security in Latin America, Santiago, Chile, October 7-10, 2012. Proceedings",
    series = "Lecture Notes in Computer Science",
    volume = "7533",
    pages = "159--176",
    publisher = "Springer",
    year = "2012",
    url = "https://doi.org/10.1007/978-3-642-33481-8_9",
    doi = "10.1007/978-3-642-33481-8_9",
    bibsource = "dblp computer science bibliography, https://dblp.org"
}

@inproceedings{DBLP:conf/lpar/FarinierDBL18,
    author = "Farinier, Benjamin and David, Robin and Bardin, S{\'{e}}bastien and Lemerre, Matthieu",
    editor = "Barthe, Gilles and Sutcliffe, Geoff and Veanes, Margus",
    title = "Arrays Made Simpler: An Efficient, Scalable and Thorough Preprocessing",
    booktitle = "{LPAR-22.} 22nd International Conference on Logic for Programming, Artificial Intelligence and Reasoning, Awassa, Ethiopia, 16-21 November 2018",
    series = "EPiC Series in Computing",
    volume = "57",
    pages = "363--380",
    publisher = "EasyChair",
    year = "2018",
    url = "https://doi.org/10.29007/dc9b",
    doi = "10.29007/dc9b",
    bibsource = "dblp computer science bibliography, https://dblp.org"
}

@inproceedings{DBLP:conf/memocode/ChattopadhyayBR17,
    author = "Chattopadhyay, Sudipta and Beck, Moritz and Rezine, Ahmed and Zeller, Andreas",
    editor = "Talpin, Jean{-}Pierre and Derler, Patricia and Schneider, Klaus",
    title = "Quantifying the information leak in cache attacks via symbolic execution",
    booktitle = "Proceedings of the 15th {ACM-IEEE} International Conference on Formal Methods and Models for System Design, {MEMOCODE} 2017, Vienna, Austria, September 29 - October 02, 2017",
    pages = "25--35",
    publisher = "{ACM}",
    year = "2017",
    url = "https://doi.org/10.1145/3127041.3127044",
    doi = "10.1145/3127041.3127044",
    bibsource = "dblp computer science bibliography, https://dblp.org"
}

@inproceedings{DBLP:conf/ndss/AvgerinosCHB11,
    author = "Avgerinos, Thanassis and Cha, Sang Kil and Hao, Brent Lim Tze and Brumley, David",
    title = "{AEG:} Automatic Exploit Generation",
    booktitle = "Proceedings of the Network and Distributed System Security Symposium, {NDSS} 2011, San Diego, California, USA, 6th February - 9th February 2011",
    publisher = "The Internet Society",
    year = "2011",
    url = "https://www.ndss-symposium.org/ndss2011/aeg-automatic-exploit-generation",
    bibsource = "dblp computer science bibliography, https://dblp.org"
}

@inproceedings{DBLP:conf/nordsec/TedescoRS10,
    author = "Tedesco, Filippo Del and Russo, Alejandro and Sands, David",
    editor = {Aura, Tuomas and J{\"{a}}rvinen, Kimmo and Nyberg, Kaisa},
    title = "Implementing Erasure Policies Using Taint Analysis",
    booktitle = "Information Security Technology for Applications - 15th Nordic Conference on Secure {IT} Systems, NordSec 2010, Espoo, Finland, October 27-29, 2010, Revised Selected Papers",
    series = "Lecture Notes in Computer Science",
    volume = "7127",
    pages = "193--209",
    publisher = "Springer",
    year = "2010",
    url = "https://doi.org/10.1007/978-3-642-27937-9_14",
    doi = "10.1007/978-3-642-27937-9_14",
    bibsource = "dblp computer science bibliography, https://dblp.org"
}

@inproceedings{DBLP:conf/osdi/CadarDE08,
    author = "Cadar, Cristian and Dunbar, Daniel and Engler, Dawson R.",
    editor = "Draves, Richard and van Renesse, Robbert",
    title = "{KLEE:} Unassisted and Automatic Generation of High-Coverage Tests for Complex Systems Programs",
    booktitle = "8th {USENIX} Symposium on Operating Systems Design and Implementation, {OSDI} 2008, December 8-10, 2008, San Diego, California, USA, Proceedings",
    pages = "209--224",
    publisher = "{USENIX} Association",
    year = "2008",
    url = "http://www.usenix.org/events/osdi08/tech/full_papers/cadar/cadar.pdf",
    bibsource = "dblp computer science bibliography, https://dblp.org"
}

@inproceedings{DBLP:conf/pkc/Bernstein06,
    author = "Bernstein, Daniel J.",
    editor = "Yung, Moti and Dodis, Yevgeniy and Kiayias, Aggelos and Malkin, Tal",
    title = "Curve25519: New Diffie-Hellman Speed Records",
    booktitle = "Public Key Cryptography - {PKC} 2006, 9th International Conference on Theory and Practice of Public-Key Cryptography, New York, NY, USA, April 24-26, 2006, Proceedings",
    series = "Lecture Notes in Computer Science",
    volume = "3958",
    pages = "207--228",
    publisher = "Springer",
    year = "2006",
    url = "https://doi.org/10.1007/11745853_14",
    doi = "10.1007/11745853_14",
    bibsource = "dblp computer science bibliography, https://dblp.org"
}

@inproceedings{DBLP:conf/pldi/DoychevK17,
    author = {Doychev, Goran and K{\"{o}}pf, Boris},
    editor = "Cohen, Albert and Vechev, Martin T.",
    title = "Rigorous analysis of software countermeasures against cache attacks",
    booktitle = "Proceedings of the 38th {ACM} {SIGPLAN} Conference on Programming Language Design and Implementation, {PLDI} 2017, Barcelona, Spain, June 18-23, 2017",
    pages = "406--421",
    publisher = "{ACM}",
    year = "2017",
    url = "https://doi.org/10.1145/3062341.3062388",
    doi = "10.1145/3062341.3062388",
    bibsource = "dblp computer science bibliography, https://dblp.org"
}

@inproceedings{DBLP:conf/pldi/NethercoteS07,
    author = "Nethercote, Nicholas and Seward, Julian",
    editor = "Ferrante, Jeanne and McKinley, Kathryn S.",
    title = "Valgrind: a framework for heavyweight dynamic binary instrumentation",
    booktitle = "Proceedings of the {ACM} {SIGPLAN} 2007 Conference on Programming Language Design and Implementation, San Diego, California, USA, June 10-13, 2007",
    pages = "89--100",
    publisher = "{ACM}",
    year = "2007",
    url = "https://doi.org/10.1145/1250734.1250746",
    doi = "10.1145/1250734.1250746",
    bibsource = "dblp computer science bibliography, https://dblp.org"
}

@inproceedings{DBLP:conf/popl/Agat00,
    author = "Agat, Johan",
    editor = "Wegman, Mark N. and Reps, Thomas W.",
    title = "Transforming Out Timing Leaks",
    booktitle = "{POPL} 2000, Proceedings of the 27th {ACM} {SIGPLAN-SIGACT} Symposium on Principles of Programming Languages, Boston, Massachusetts, USA, January 19-21, 2000",
    pages = "40--53",
    publisher = "{ACM}",
    year = "2000",
    url = "https://doi.org/10.1145/325694.325702",
    doi = "10.1145/325694.325702",
    bibsource = "dblp computer science bibliography, https://dblp.org"
}

@inproceedings{DBLP:conf/popl/AustinF12,
    author = "Austin, Thomas H. and Flanagan, Cormac",
    editor = "Field, John and Hicks, Michael",
    title = "Multiple facets for dynamic information flow",
    booktitle = "Proceedings of the 39th {ACM} {SIGPLAN-SIGACT} Symposium on Principles of Programming Languages, {POPL} 2012, Philadelphia, Pennsylvania, USA, January 22-28, 2012",
    pages = "165--178",
    publisher = "{ACM}",
    year = "2012",
    url = "https://doi.org/10.1145/2103656.2103677",
    doi = "10.1145/2103656.2103677",
    bibsource = "dblp computer science bibliography, https://dblp.org"
}

@inproceedings{DBLP:conf/popl/Benton04,
    author = "Benton, Nick",
    editor = "Jones, Neil D. and Leroy, Xavier",
    title = "Simple relational correctness proofs for static analyses and program transformations",
    booktitle = "Proceedings of the 31st {ACM} {SIGPLAN-SIGACT} Symposium on Principles of Programming Languages, {POPL} 2004, Venice, Italy, January 14-16, 2004",
    pages = "14--25",
    publisher = "{ACM}",
    year = "2004",
    url = "https://doi.org/10.1145/964001.964003",
    doi = "10.1145/964001.964003",
    bibsource = "dblp computer science bibliography, https://dblp.org"
}

@inproceedings{DBLP:conf/popl/JourdanLBLP15,
    author = "Jourdan, Jacques{-}Henri and Laporte, Vincent and Blazy, Sandrine and Leroy, Xavier and Pichardie, David",
    editor = "Rajamani, Sriram K. and Walker, David",
    title = "A Formally-Verified {C} Static Analyzer",
    booktitle = "Proceedings of the 42nd Annual {ACM} {SIGPLAN-SIGACT} Symposium on Principles of Programming Languages, {POPL} 2015, Mumbai, India, January 15-17, 2015",
    pages = "247--259",
    publisher = "{ACM}",
    year = "2015",
    url = "https://doi.org/10.1145/2676726.2676966",
    doi = "10.1145/2676726.2676966",
    bibsource = "dblp computer science bibliography, https://dblp.org"
}

@inproceedings{DBLP:conf/sas/TerauchiA05,
    author = "Terauchi, Tachio and Aiken, Alexander",
    editor = "Hankin, Chris and Siveroni, Igor",
    title = "Secure Information Flow as a Safety Problem",
    booktitle = "Static Analysis, 12th International Symposium, {SAS} 2005, London, UK, September 7-9, 2005, Proceedings",
    series = "Lecture Notes in Computer Science",
    volume = "3672",
    pages = "352--367",
    publisher = "Springer",
    year = "2005",
    url = "https://doi.org/10.1007/11547662_24",
    doi = "10.1007/11547662_24",
    bibsource = "dblp computer science bibliography, https://dblp.org"
}

@inproceedings{DBLP:conf/sec/DoBH15,
    author = {Do, Quoc Huy and Bubel, Richard and H{\"{a}}hnle, Reiner},
    editor = "Federrath, Hannes and Gollmann, Dieter",
    title = "Exploit Generation for Information Flow Leaks in Object-Oriented Programs",
    booktitle = "{ICT} Systems Security and Privacy Protection - 30th {IFIP} {TC} 11 International Conference, {SEC} 2015, Hamburg, Germany, May 26-28, 2015, Proceedings",
    series = "{IFIP} Advances in Information and Communication Technology",
    volume = "455",
    pages = "401--415",
    publisher = "Springer",
    year = "2015",
    url = "https://doi.org/10.1007/978-3-319-18467-8_27",
    doi = "10.1007/978-3-319-18467-8_27",
    bibsource = "dblp computer science bibliography, https://dblp.org"
}

@inproceedings{DBLP:conf/secdev/CauligiSBJHJS17,
    author = "Cauligi, Sunjay and Soeller, Gary and Brown, Fraser and Johannesmeyer, Brian and Huang, Yunlu and Jhala, Ranjit and Stefan, Deian",
    title = "FaCT: {A} Flexible, Constant-Time Programming Language",
    booktitle = "{IEEE} Cybersecurity Development, SecDev 2017, Cambridge, MA, USA, September 24-26, 2017",
    pages = "69--76",
    publisher = "{IEEE} Computer Society",
    year = "2017",
    url = "https://doi.org/10.1109/SecDev.2017.24",
    doi = "10.1109/SecDev.2017.24",
    bibsource = "dblp computer science bibliography, https://dblp.org"
}

@inproceedings{DBLP:conf/sigopsE/GarfinkelPCR04,
    author = "Garfinkel, Tal and Pfaff, Ben and Chow, Jim and Rosenblum, Mendel",
    editor = "Berbers, Yolande and Castro, Miguel",
    title = "Data lifetime is a systems problem",
    booktitle = "Proceedings of the 11st {ACM} {SIGOPS} European Workshop, Leuven, Belgium, September 19-22, 2004",
    pages = "10",
    publisher = "{ACM}",
    year = "2004",
    url = "https://doi.org/10.1145/1133572.1133599",
    doi = "10.1145/1133572.1133599",
    bibsource = "dblp computer science bibliography, https://dblp.org"
}

@inproceedings{DBLP:conf/sp/AlFardanP13,
    author = "AlFardan, Nadhem J. and Paterson, Kenneth G.",
    title = "Lucky Thirteen: Breaking the {TLS} and {DTLS} Record Protocols",
    booktitle = "2013 {IEEE} Symposium on Security and Privacy, {SP} 2013, Berkeley, CA, USA, May 19-22, 2013",
    pages = "526--540",
    publisher = "{IEEE} Computer Society",
    year = "2013",
    url = "https://doi.org/10.1109/SP.2013.42",
    doi = "10.1109/SP.2013.42",
    bibsource = "dblp computer science bibliography, https://dblp.org"
}

@inproceedings{DBLP:conf/sp/BardinDM17,
    author = "Bardin, S{\'{e}}bastien and David, Robin and Marion, Jean{-}Yves",
    title = "Backward-Bounded {DSE:} Targeting Infeasibility Questions on Obfuscated Codes",
    booktitle = "2017 {IEEE} Symposium on Security and Privacy, {SP} 2017, San Jose, CA, USA, May 22-26, 2017",
    pages = "633--651",
    publisher = "{IEEE} Computer Society",
    year = "2017",
    url = "https://doi.org/10.1109/SP.2017.36",
    doi = "10.1109/SP.2017.36",
    bibsource = "dblp computer science bibliography, https://dblp.org"
}

@inproceedings{DBLP:conf/sp/BrotzmanLZTK19,
    author = "Brotzman, Robert and Liu, Shen and Zhang, Danfeng and Tan, Gang and Kandemir, Mahmut T.",
    title = "CaSym: Cache Aware Symbolic Execution for Side Channel Detection and Mitigation",
    booktitle = "2019 {IEEE} Symposium on Security and Privacy, {SP} 2019, San Francisco, CA, USA, May 19-23, 2019",
    pages = "505--521",
    publisher = "{IEEE}",
    year = "2019",
    url = "https://doi.org/10.1109/SP.2019.00022",
    doi = "10.1109/SP.2019.00022",
    bibsource = "dblp computer science bibliography, https://dblp.org"
}

@inproceedings{DBLP:conf/sp/CoppensVBS09,
    author = "Coppens, Bart and Verbauwhede, Ingrid and Bosschere, Koen De and Sutter, Bjorn De",
    title = "Practical Mitigations for Timing-Based Side-Channel Attacks on Modern x86 Processors",
    booktitle = "30th {IEEE} Symposium on Security and Privacy (S{\&}P 2009), 17-20 May 2009, Oakland, California, {USA}",
    pages = "45--60",
    publisher = "{IEEE} Computer Society",
    year = "2009",
    url = "https://doi.org/10.1109/SP.2009.19",
    doi = "10.1109/SP.2009.19",
    bibsource = "dblp computer science bibliography, https://dblp.org"
}

@inproceedings{DBLP:conf/sp/DanielBR20,
    author = "Daniel, Lesly{-}Ann and Bardin, S{\'{e}}bastien and Rezk, Tamara",
    title = "Binsec/Rel: Efficient Relational Symbolic Execution for Constant-Time at Binary-Level",
    booktitle = "2020 {IEEE} Symposium on Security and Privacy, {SP} 2020, San Francisco, CA, USA, May 18-21, 2020",
    pages = "1021--1038",
    publisher = "{IEEE}",
    year = "2020",
    url = "https://doi.org/10.1109/SP40000.2020.00074",
    doi = "10.1109/SP40000.2020.00074",
    bibsource = "dblp computer science bibliography, https://dblp.org"
}

@inproceedings{DBLP:conf/sp/DSilvaPS15,
    author = "D'Silva, Vijay and Payer, Mathias and Song, Dawn Xiaodong",
    title = "The Correctness-Security Gap in Compiler Optimization",
    booktitle = "2015 {IEEE} Symposium on Security and Privacy Workshops, {SPW} 2015, San Jose, CA, USA, May 21-22, 2015",
    pages = "73--87",
    publisher = "{IEEE} Computer Society",
    year = "2015",
    url = "https://doi.org/10.1109/SPW.2015.33",
    doi = "10.1109/SPW.2015.33",
    bibsource = "dblp computer science bibliography, https://dblp.org"
}

@inproceedings{DBLP:conf/sp/KocherHFGGHHLM019,
    author = "Kocher, Paul and Horn, Jann and Fogh, Anders and Genkin, Daniel and Gruss, Daniel and Haas, Werner and Hamburg, Mike and Lipp, Moritz and Mangard, Stefan and Prescher, Thomas and Schwarz, Michael and Yarom, Yuval",
    title = "Spectre Attacks: Exploiting Speculative Execution",
    booktitle = "2019 {IEEE} Symposium on Security and Privacy, {SP} 2019, San Francisco, CA, USA, May 19-23, 2019",
    pages = "1--19",
    publisher = "{IEEE}",
    year = "2019",
    url = "https://doi.org/10.1109/SP.2019.00002",
    doi = "10.1109/SP.2019.00002",
    bibsource = "dblp computer science bibliography, https://dblp.org"
}

@inproceedings{DBLP:conf/sp/NanevskiBG11,
    author = "Nanevski, Aleksandar and Banerjee, Anindya and Garg, Deepak",
    title = "Verification of Information Flow and Access Control Policies with Dependent Types",
    booktitle = "32nd {IEEE} Symposium on Security and Privacy, S{\&}P 2011, 22-25 May 2011, Berkeley, California, {USA}",
    pages = "165--179",
    publisher = "{IEEE} Computer Society",
    year = "2011",
    url = "https://doi.org/10.1109/SP.2011.12",
    doi = "10.1109/SP.2011.12",
    bibsource = "dblp computer science bibliography, https://dblp.org"
}

@inproceedings{DBLP:conf/sp/Shoshitaishvili16,
    author = {Shoshitaishvili, Yan and Wang, Ruoyu and Salls, Christopher and Stephens, Nick and Polino, Mario and Dutcher, Andrew and Grosen, John and Feng, Siji and Hauser, Christophe and Kr{\"{u}}gel, Christopher and Vigna, Giovanni},
    title = "{SOK:} (State of) The Art of War: Offensive Techniques in Binary Analysis",
    booktitle = "{IEEE} Symposium on Security and Privacy, {SP} 2016, San Jose, CA, USA, May 22-26, 2016",
    pages = "138--157",
    publisher = "{IEEE} Computer Society",
    year = "2016",
    url = "https://doi.org/10.1109/SP.2016.17",
    doi = "10.1109/SP.2016.17",
    bibsource = "dblp computer science bibliography, https://dblp.org"
}

@inproceedings{DBLP:conf/sp/YadegariJWD15,
    author = "Yadegari, Babak and Johannesmeyer, Brian and Whitely, Ben and Debray, Saumya",
    title = "A Generic Approach to Automatic Deobfuscation of Executable Code",
    booktitle = "2015 {IEEE} Symposium on Security and Privacy, {SP} 2015, San Jose, CA, USA, May 17-21, 2015",
    pages = "674--691",
    publisher = "{IEEE} Computer Society",
    year = "2015",
    url = "https://doi.org/10.1109/SP.2015.47",
    doi = "10.1109/SP.2015.47",
    bibsource = "dblp computer science bibliography, https://dblp.org"
}

@inproceedings{DBLP:conf/tacas/DjoudiB15,
    author = "Djoudi, Adel and Bardin, S{\'{e}}bastien",
    editor = "Baier, Christel and Tinelli, Cesare",
    title = "{BINSEC:} Binary Code Analysis with Low-Level Regions",
    booktitle = "Tools and Algorithms for the Construction and Analysis of Systems - 21st International Conference, {TACAS} 2015, Held as Part of the European Joint Conferences on Theory and Practice of Software, {ETAPS} 2015, London, UK, April 11-18, 2015. Proceedings",
    series = "Lecture Notes in Computer Science",
    volume = "9035",
    pages = "212--217",
    publisher = "Springer",
    year = "2015",
    url = "https://doi.org/10.1007/978-3-662-46681-0_17",
    doi = "10.1007/978-3-662-46681-0_17",
    bibsource = "dblp computer science bibliography, https://dblp.org"
}

@inproceedings{DBLP:conf/tacas/MouraB08,
    author = "de Moura, Leonardo Mendon{\c{c}}a and Bj{\o}rner, Nikolaj S.",
    editor = "Ramakrishnan, C. R. and Rehof, Jakob",
    title = "{Z3:} An Efficient {SMT} Solver",
    booktitle = "Tools and Algorithms for the Construction and Analysis of Systems, 14th International Conference, {TACAS} 2008, Held as Part of the Joint European Conferences on Theory and Practice of Software, {ETAPS} 2008, Budapest, Hungary, March 29-April 6, 2008. Proceedings",
    series = "Lecture Notes in Computer Science",
    volume = "4963",
    pages = "337--340",
    publisher = "Springer",
    year = "2008",
    url = "https://doi.org/10.1007/978-3-540-78800-3_24",
    doi = "10.1007/978-3-540-78800-3_24",
    bibsource = "dblp computer science bibliography, https://dblp.org"
}

@inproceedings{DBLP:conf/uss/AlmeidaBBDE16,
    author = "Almeida, Jos{\'{e}} Bacelar and Barbosa, Manuel and Barthe, Gilles and Dupressoir, Fran{\c{c}}ois and Emmi, Michael",
    editor = "Holz, Thorsten and Savage, Stefan",
    title = "Verifying Constant-Time Implementations",
    booktitle = "25th {USENIX} Security Symposium, {USENIX} Security 16, Austin, TX, USA, August 10-12, 2016",
    pages = "53--70",
    publisher = "{USENIX} Association",
    year = "2016",
    url = "https://www.usenix.org/conference/usenixsecurity16/technical-sessions/presentation/almeida",
    bibsource = "dblp computer science bibliography, https://dblp.org"
}

@inproceedings{DBLP:conf/uss/BondHKLLPRST17,
    author = "Bond, Barry and Hawblitzel, Chris and Kapritsos, Manos and Leino, K. Rustan M. and Lorch, Jacob R. and Parno, Bryan and Rane, Ashay and Setty, Srinath T. V. and Thompson, Laure",
    editor = "Kirda, Engin and Ristenpart, Thomas",
    title = "Vale: Verifying High-Performance Cryptographic Assembly Code",
    booktitle = "26th {USENIX} Security Symposium, {USENIX} Security 2017, Vancouver, BC, Canada, August 16-18, 2017",
    pages = "917--934",
    publisher = "{USENIX} Association",
    year = "2017",
    url = "https://www.usenix.org/conference/usenixsecurity17/technical-sessions/presentation/bond",
    bibsource = "dblp computer science bibliography, https://dblp.org"
}

@inproceedings{DBLP:conf/uss/ChowPGCR04,
    author = "Chow, Jim and Pfaff, Ben and Garfinkel, Tal and Christopher, Kevin and Rosenblum, Mendel",
    editor = "Blaze, Matt",
    title = "Understanding Data Lifetime via Whole System Simulation (Awarded Best Paper!)",
    booktitle = "Proceedings of the 13th {USENIX} Security Symposium, August 9-13, 2004, San Diego, CA, {USA}",
    pages = "321--336",
    publisher = "{USENIX}",
    year = "2004",
    url = "http://www.usenix.org/publications/library/proceedings/sec04/tech/chow.html",
    bibsource = "dblp computer science bibliography, https://dblp.org"
}

@inproceedings{DBLP:conf/uss/ChowPGR05,
    author = "Chow, Jim and Pfaff, Ben and Garfinkel, Tal and Rosenblum, Mendel",
    editor = "McDaniel, Patrick D.",
    title = "Shredding Your Garbage: Reducing Data Lifetime Through Secure Deallocation",
    booktitle = "Proceedings of the 14th {USENIX} Security Symposium, Baltimore, MD, USA, July 31 - August 5, 2005",
    publisher = "{USENIX} Association",
    year = "2005",
    url = "https://www.usenix.org/conference/14th-usenix-security-symposium/shredding-your-garbage-reducing-data-lifetime-through",
    bibsource = "dblp computer science bibliography, https://dblp.org"
}

@inproceedings{DBLP:conf/uss/DoychevFKMR13,
    author = {Doychev, Goran and Feld, Dominik and K{\"{o}}pf, Boris and Mauborgne, Laurent and Reineke, Jan},
    editor = "King, Samuel T.",
    title = "CacheAudit: {A} Tool for the Static Analysis of Cache Side Channels",
    booktitle = "Proceedings of the 22th {USENIX} Security Symposium, Washington, DC, USA, August 14-16, 2013",
    pages = "431--446",
    publisher = "{USENIX} Association",
    year = "2013",
    url = "https://www.usenix.org/conference/usenixsecurity13/technical-sessions/paper/doychev",
    bibsource = "dblp computer science bibliography, https://dblp.org"
}

@inproceedings{DBLP:conf/uss/RaneLT15,
    author = "Rane, Ashay and Lin, Calvin and Tiwari, Mohit",
    editor = "Jung, Jaeyeon and Holz, Thorsten",
    title = "Raccoon: Closing Digital Side-Channels through Obfuscated Execution",
    booktitle = "24th {USENIX} Security Symposium, {USENIX} Security 15, Washington, D.C., USA, August 12-14, 2015",
    pages = "431--446",
    publisher = "{USENIX} Association",
    year = "2015",
    url = "https://www.usenix.org/conference/usenixsecurity15/technical-sessions/presentation/rane",
    bibsource = "dblp computer science bibliography, https://dblp.org"
}

@inproceedings{DBLP:conf/uss/SchwartzAB11,
    author = "Schwartz, Edward J. and Avgerinos, Thanassis and Brumley, David",
    title = "{Q:} Exploit Hardening Made Easy",
    booktitle = "20th {USENIX} Security Symposium, San Francisco, CA, USA, August 8-12, 2011, Proceedings",
    publisher = "{USENIX} Association",
    year = "2011",
    url = "http://static.usenix.org/events/sec11/tech/full_papers/Schwartz.pdf",
    bibsource = "dblp computer science bibliography, https://dblp.org"
}

@inproceedings{DBLP:conf/uss/WangWLZW17,
    author = "Wang, Shuai and Wang, Pei and Liu, Xiao and Zhang, Danfeng and Wu, Dinghao",
    editor = "Kirda, Engin and Ristenpart, Thomas",
    title = "CacheD: Identifying Cache-Based Timing Channels in Production Software",
    booktitle = "26th {USENIX} Security Symposium, {USENIX} Security 2017, Vancouver, BC, Canada, August 16-18, 2017",
    pages = "235--252",
    publisher = "{USENIX} Association",
    year = "2017",
    url = "https://www.usenix.org/conference/usenixsecurity17/technical-sessions/presentation/wang-shuai",
    bibsource = "dblp computer science bibliography, https://dblp.org"
}

@inproceedings{DBLP:conf/uss/YangJOLL17,
    author = "Yang, Zhaomo and Johannesmeyer, Brian and Olesen, Anders Trier and Lerner, Sorin and Levchenko, Kirill",
    editor = "Kirda, Engin and Ristenpart, Thomas",
    title = "Dead Store Elimination (Still) Considered Harmful",
    booktitle = "26th {USENIX} Security Symposium, {USENIX} Security 2017, Vancouver, BC, Canada, August 16-18, 2017",
    pages = "1025--1040",
    publisher = "{USENIX} Association",
    year = "2017",
    url = "https://www.usenix.org/conference/usenixsecurity17/technical-sessions/presentation/yang",
    bibsource = "dblp computer science bibliography, https://dblp.org"
}

@inproceedings{DBLP:conf/wcre/DavidBTMFPM16,
    author = "David, Robin and Bardin, S{\'{e}}bastien and Ta, Thanh Dinh and Mounier, Laurent and Feist, Josselin and Potet, Marie{-}Laure and Marion, Jean{-}Yves",
    title = "{BINSEC/SE:} {A} Dynamic Symbolic Execution Toolkit for Binary-Level Analysis",
    booktitle = "{IEEE} 23rd International Conference on Software Analysis, Evolution, and Reengineering, {SANER} 2016, Suita, Osaka, Japan, March 14-18, 2016 - Volume 1",
    pages = "653--656",
    publisher = "{IEEE} Computer Society",
    year = "2016",
    url = "https://doi.org/10.1109/SANER.2016.43",
    doi = "10.1109/SANER.2016.43",
    bibsource = "dblp computer science bibliography, https://dblp.org"
}

@inproceedings{DBLP:conf/woot/VanegueH12,
    author = "Vanegue, Julien and Heelan, Sean",
    editor = "Bursztein, Elie and Dullien, Thomas",
    title = "{SMT} Solvers in Software Security",
    booktitle = "6th {USENIX} Workshop on Offensive Technologies, WOOT'12, August 6-7, 2012, Bellevue, WA, USA, Proceedings",
    pages = "85--96",
    publisher = "{USENIX} Association",
    year = "2012",
    url = "http://www.usenix.org/conference/woot12/smt-solvers-software-security",
    bibsource = "dblp computer science bibliography, https://dblp.org"
}

@inproceedings{DBLP:conf/www/NgoBFRRS18,
    author = "Ngo, Minh and Bielova, Nataliia and Flanagan, Cormac and Rezk, Tamara and Russo, Alejandro and Schmitz, Thomas",
    editor = "Champin, Pierre{-}Antoine and Gandon, Fabien and Lalmas, Mounia and Ipeirotis, Panagiotis G.",
    title = "A Better Facet of Dynamic Information Flow Control",
    booktitle = "Companion of the The Web Conference 2018 on The Web Conference 2018, {WWW} 2018, Lyon , France, April 23-27, 2018",
    pages = "731--739",
    publisher = "{ACM}",
    year = "2018",
    url = "https://doi.org/10.1145/3184558.3185979",
    doi = "10.1145/3184558.3185979",
    bibsource = "dblp computer science bibliography, https://dblp.org"
}

@article{DBLP:journals/cacm/CadarS13,
    author = "Cadar, Cristian and Sen, Koushik",
    title = "Symbolic execution for software testing: three decades later",
    journal = "Commun. {ACM}",
    volume = "56",
    number = "2",
    pages = "82--90",
    year = "2013",
    url = "https://doi.org/10.1145/2408776.2408795",
    doi = "10.1145/2408776.2408795",
    bibsource = "dblp computer science bibliography, https://dblp.org"
}

@article{DBLP:journals/cacm/DenningD77,
    author = "Denning, Dorothy E. and Denning, Peter J.",
    title = "Certification of Programs for Secure Information Flow",
    journal = "Commun. {ACM}",
    volume = "20",
    number = "7",
    pages = "504--513",
    year = "1977",
    url = "https://doi.org/10.1145/359636.359712",
    doi = "10.1145/359636.359712",
    bibsource = "dblp computer science bibliography, https://dblp.org"
}

@article{DBLP:journals/cacm/GodefroidLM12,
    author = "Godefroid, Patrice and Levin, Michael Y. and Molnar, David A.",
    title = "{SAGE:} whitebox fuzzing for security testing",
    journal = "Commun. {ACM}",
    volume = "55",
    number = "3",
    pages = "40--44",
    year = "2012",
    url = "https://doi.org/10.1145/2093548.2093564",
    doi = "10.1145/2093548.2093564",
    bibsource = "dblp computer science bibliography, https://dblp.org"
}

@article{DBLP:journals/cacm/King76,
    author = "King, James C.",
    title = "Symbolic Execution and Program Testing",
    journal = "Commun. {ACM}",
    volume = "19",
    number = "7",
    pages = "385--394",
    year = "1976",
    url = "https://doi.org/10.1145/360248.360252",
    doi = "10.1145/360248.360252",
    bibsource = "dblp computer science bibliography, https://dblp.org"
}

@article{DBLP:journals/cacm/Leroy09,
    author = "Leroy, Xavier",
    title = "Formal verification of a realistic compiler",
    journal = "Commun. {ACM}",
    volume = "52",
    number = "7",
    pages = "107--115",
    year = "2009",
    url = "https://doi.org/10.1145/1538788.1538814",
    doi = "10.1145/1538788.1538814",
    bibsource = "dblp computer science bibliography, https://dblp.org"
}

@article{DBLP:journals/dc/AlpernS87,
    author = "Alpern, Bowen and Schneider, Fred B.",
    title = "Recognizing Safety and Liveness",
    journal = "Distributed Comput.",
    volume = "2",
    number = "3",
    pages = "117--126",
    year = "1987",
    url = "https://doi.org/10.1007/BF01782772",
    doi = "10.1007/BF01782772",
    bibsource = "dblp computer science bibliography, https://dblp.org"
}

@article{DBLP:journals/fac/KirchnerKPSY15,
    author = "Kirchner, Florent and Kosmatov, Nikolai and Prevosto, Virgile and Signoles, Julien and Yakobowski, Boris",
    title = "Frama-C: {A} software analysis perspective",
    journal = "Formal Aspects Comput.",
    volume = "27",
    number = "3",
    pages = "573--609",
    year = "2015",
    url = "https://doi.org/10.1007/s00165-014-0326-7",
    doi = "10.1007/s00165-014-0326-7",
    bibsource = "dblp computer science bibliography, https://dblp.org"
}

@article{DBLP:journals/ieeesp/AvgerinosBDGNRW18,
    author = "Avgerinos, Thanassis and Brumley, David and Davis, John and Goulden, Ryan and Nighswander, Tyler and Rebert, Alexandre and Williamson, Ned",
    title = "The Mayhem Cyber Reasoning System",
    journal = "{IEEE} Secur. Priv.",
    volume = "16",
    number = "2",
    pages = "52--60",
    year = "2018",
    url = "https://doi.org/10.1109/MSP.2018.1870873",
    doi = "10.1109/MSP.2018.1870873",
    bibsource = "dblp computer science bibliography, https://dblp.org"
}

@article{DBLP:journals/ijisec/KopfM07,
    author = {K{\"{o}}pf, Boris and Mantel, Heiko},
    title = "Transformational typing and unification for automatically correcting insecure programs",
    journal = "Int. J. Inf. Sec.",
    volume = "6",
    number = "2-3",
    pages = "107--131",
    year = "2007",
    url = "https://doi.org/10.1007/s10207-007-0016-z",
    doi = "10.1007/s10207-007-0016-z",
    bibsource = "dblp computer science bibliography, https://dblp.org"
}

@article{DBLP:journals/sttt/HavelundP00,
    author = "Havelund, Klaus and Pressburger, Thomas",
    title = "Model Checking {JAVA} Programs using {JAVA} PathFinder",
    journal = "Int. J. Softw. Tools Technol. Transf.",
    volume = "2",
    number = "4",
    pages = "366--381",
    year = "2000",
    url = "https://doi.org/10.1007/s100090050043",
    doi = "10.1007/s100090050043",
    bibsource = "dblp computer science bibliography, https://dblp.org"
}

@article{DBLP:journals/tcad/ChattopadhyayR18,
    author = "Chattopadhyay, Sudipta and Roychoudhury, Abhik",
    title = "Symbolic Verification of Cache Side-Channel Freedom",
    journal = "{IEEE} Trans. Comput. Aided Des. Integr. Circuits Syst.",
    volume = "37",
    number = "11",
    pages = "2812--2823",
    year = "2018",
    url = "https://doi.org/10.1109/TCAD.2018.2858402",
    doi = "10.1109/TCAD.2018.2858402",
    bibsource = "dblp computer science bibliography, https://dblp.org"
}

@article{DBLP:journals/tocs/ChipounovKC12,
    author = "Chipounov, Vitaly and Kuznetsov, Volodymyr and Candea, George",
    title = "The {S2E} Platform: Design, Implementation, and Applications",
    journal = "{ACM} Trans. Comput. Syst.",
    volume = "30",
    number = "1",
    pages = "2:1--2:49",
    year = "2012",
    url = "https://doi.org/10.1145/2110356.2110358",
    doi = "10.1145/2110356.2110358",
    bibsource = "dblp computer science bibliography, https://dblp.org"
}

@article{DBLP:journals/toplas/BalakrishnanR10,
    author = "Balakrishnan, Gogul and Reps, Thomas W.",
    title = "{WYSINWYX:} What you see is not what you eXecute",
    journal = "{ACM} Trans. Program. Lang. Syst.",
    volume = "32",
    number = "6",
    pages = "23:1--23:84",
    year = "2010",
    url = "https://doi.org/10.1145/1749608.1749612",
    doi = "10.1145/1749608.1749612",
    bibsource = "dblp computer science bibliography, https://dblp.org"
}

@article{DBLP:journals/toplas/NelsonO79,
    author = "Nelson, Greg and Oppen, Derek C.",
    title = "Simplification by Cooperating Decision Procedures",
    journal = "{ACM} Trans. Program. Lang. Syst.",
    volume = "1",
    number = "2",
    pages = "245--257",
    year = "1979",
    url = "https://doi.org/10.1145/357073.357079",
    doi = "10.1145/357073.357079",
    bibsource = "dblp computer science bibliography, https://dblp.org"
}

@inproceedings{farinaRelationalSymbolicExecution2019,
    author = "Farina, Gian Pietro and Chong, Stephen and Gaboardi, Marco",
    title = "Relational Symbolic Execution",
    booktitle = "{{PPDP}}",
    date = "2019",
    pages = "10:1--10:14",
    publisher = "{ACM}"
}

@online{FixedSizeBitVectorsTheorySMTLIB,
    title = "{{FixedSizeBitVectors Theory}}, {{SMT-LIB}}",
    url = "http://smtlib.cs.uiowa.edu/theories-FixedSizeBitVectors.shtml",
    urldate = "2019-04-02"
}

@inproceedings{hansen2006non,
    author = "Hansen, René Rydhof and Probst, Christian W",
    title = "Non-Interference and Erasure Policies for Java Card Bytecode",
    booktitle = "6th International Workshop on Issues in the Theory of Security ({{WITS}}’06)",
    date = "2006"
}

@online{ImdeasoftwareVerifyingconstanttime,
    title = "Imdea-Software/Verifying-Constant-Time",
    url = "https://github.com/imdea-software/verifying-constant-time",
    urldate = "2019-10-13",
    langid = "english",
    organization = "{GitHub}"
}

@online{langleyImperialVioletCheckingThat2010,
    author = "Langley, Adam",
    title = "{{ImperialViolet}} - {{Checking}} That Functions Are Constant Time with {{Valgrind}}",
    date = "2010",
    url = "https://www.imperialviolet.org/2010/04/01/ctgrind.html",
    urldate = "2019-03-14"
}

@article{niemetzBoolectorSystemDescription2014,
    author = "Niemetz, Aina and Preiner, Mathias and Biere, Armin",
    title = "Boolector 2.0 System Description",
    date = "2014",
    journaltitle = "Journal on Satisfiability, Boolean Modeling and Computation",
    volume = "9",
    pages = "53--58"
}

@online{porninBearSSL,
    author = "Pornin, Thomas",
    title = "{{BearSSL}}",
    url = "https://www.bearssl.org/",
    urldate = "2019-05-23"
}

@software{Secretgrind2020,
    title = "Secretgrind",
    date = "2020-03-22T22:26:24Z",
    origdate = "2016-12-20T15:33:14Z",
    url = "https://github.com/lmrs2/secretgrind",
    urldate = "2020-09-29"
}

@online{SMTCOMP,
    title = "{{SMT-COMP}}",
    url = "https://smt-comp.github.io/2019/results.html",
    urldate = "2019-10-11",
    langid = "american",
    organization = "{SMT-COMP}"
}


%% file: ressources/software.bib
@online{OpenSSL_cleanse,
  title = {OpenSSL, {\texttt{OPENSSL\_cleanse} function}},
  url = {https://github.com/openssl/openssl/blob/master/crypto/mem_clr.c},
  urldate = {2021-04-24},
}

@online{OpenSSL,
  title = {{OpenSSL, Cryptography and SSL/TLS Toolkit}},
  url = {https://www.openssl.org/},
  urldate = {2021-04-24},
}

@online{Safeclib_memset_s,
  author = {Reini Urban},
  title = {{Safeclib}, \texttt{memset\_s} function},
  url = {https://github.com/rurban/safeclib/blob/v31082020/src/mem/memset_s.c},
  urldate = {2021-04-24},
}

@online{Safeclib_memory_barriers,
  author = {Reini Urban},
  title = {{Safeclib}, \texttt{MEMORY\_BARRIER} macro},
  url = {https://github.com/rurban/safeclib/blob/v31082020/src/mem/mem_primitives_lib.h},
  urldate = {2021-04-24},
}

@online{Libgcrypt_wipememory,
  title = {{Libgcrypt, \texttt{wipememory} function}},
  url = {https://github.com/equalitie/libgcrypt/blob/libgcrypt-1.6.3/src/gcryptrnd.c},
  urldate = {2021-04-24},
}

@online{wolfSSL_ForceZero,
  title = {{wolfSSL, \texttt{ForceZero} function}},
  url = {https://github.com/equalitie/libgcrypt/blob/libgcrypt-1.6.3/src/gcryptrnd.c},
  urldate = {2021-04-24},
}

@online{sudo_explicit_bzero,
  author = {Todd C. Miller},
  title = {{sudo, \texttt{explicit\_bzero} function}},
  url = {https://github.com/sudo-project/sudo/blob/SUDO_1_9_6/lib/util/explicit_bzero.c},
  urldate = {2021-04-24},
}

@online{libsodium_memzero,
  title = {{libsodium, \texttt{sodium\_memzero} function}},
  url = {https://github.com/jedisct1/libsodium/blob/1.0.18/src/libsodium/sodium/utils.c},
  urldate = {2021-04-24},
}

@online{hacl_memzero,
  title = {{HACL*, \texttt{Lib\_Memzero0\_memzero} function}},
  url = {https://github.com/project-everest/hacl-star/blob/v0.3.0/lib/c/Lib_Memzero0.c},
  urldate = {2021-04-24},
}

@online{clang_bug15495,
  title = {{Bug 15495 - dead store pass ignores memory clobbering asm statement}},
  url = {https://bugs.llvm.org/show_bug.cgi?id=15495},
  urldate = {2021-04-24},
}

@online{gcc_extended_asm,
  title = {{6.47.2 Extended Asm - Assembler Instructions with C Expression Operands}},
  url = {https://gcc.gnu.org/onlinedocs/gcc/Extended-Asm.html},
  urldate = {2021-04-24},
}

@online{cmov-conversion,
  title = {{LLVM provides no side-channel resistance}},
  author = {Daan Sprenkels},
  url = {https://dsprenkels.com/cmov-conversion.html},
  urldate = {2021-10-01},
}
